\documentclass{ws-m3as}

\usepackage{xcolor}
\usepackage{cases,empheq}
\usepackage{hyperref}

\def\eps{\varepsilon}
\newcommand{\beqar}{\begin{eqnarray*}}
\newcommand{\eeqar}{\end{eqnarray*}}
\newcommand{\beqarl}{\begin{eqnarray}}
\newcommand{\eeqarl}{\end{eqnarray}}
\newcommand{\be}{\begin{equation}}
\newcommand{\ee}{\end{equation}}
\newcommand{\lp}{\left(}
\newcommand{\rp}{\right)}

\def\R{\mathbb{R}}

\allowdisplaybreaks[3]

\begin{document}

\markboth{Pierre Degond, Sara Merino-Aceituno}{Nematic alignment of self-propelled particles in the macroscopic regime}

%
\catchline{}{}{}{}{}
%

\title{Nematic alignment of self-propelled particles in the macroscopic regime}

\author{Pierre Degond}

\address{Department of Mathematics, Imperial College London, 
South Kensington Campus,\\
 London SW7 2AZ, United Kingdom\\
pdegond@imperial.ac.uk}

\author{Sara Merino-Aceituno}

\address{Faculty of Mathematics, University of Vienna, 
Oskar-Morgenstern-Platz 1, 1090 Vienna, Austria \\
sara.merino@univie.ac.at\\
Department of Mathematics, University of Sussex Falmer,
Brighton BN1 9RH United Kingdom
}

\maketitle

\begin{abstract}
Starting from a particle model describing self-propelled particles interacting through nematic alignment, we derive a macroscopic model for the particle density and mean direction of motion. We first propose a mean-field kinetic model of the particle dynamics. After diffusive rescaling of the kinetic equation, we formally show that the distribution function converges to an equilibrium distribution in particle direction, whose local density and mean direction satisfies a cross-diffusion system. We show that the system is consistent with symmetries typical of a nematic material. The derivation is carried over by means of a Hilbert expansion. It requires the inversion of the linearized collision operator for which we show that the generalized collision invariants, a concept introduced to overcome the lack of momentum conservation of the system, plays a central role. This cross diffusion system poses many new challenging questions.  
\end{abstract}



\section{Introduction}

Systems of active (or self-propelled) particles have received a great deal of attention in the last decade due to their potential for explaining emergent phenomena occurring for instance in animal collective behavior,\cite{couzin2002collective} development and cancer\cite{kabla2012collective} or social mass phenomena.\cite{helbing1995social} We refer to Ref.~\refcite{vicsek2012collective} for a review on the subject. Among all the models, the Vicsek model\cite{vicsek1995novel} has been particularly studied due to its simplicity. In the Vicsek model, self-propelled particles tend to align with their neighbors up to some random uncertainty. To be more specific about the type of considered alignment, we clarify first the difference between `orientation' and 'direction': two vectors have the same orientation if after normalization, they are equal; two vectors have the same direction if after normalization, they are equal or opposite (so every direction has two orientations). The alignment in the Vicsek model is so called polar in the sense that if a particle's orientation and the neighbors' mean particle orientation are opposite, the particle will make a U-turn  to adopt the same orientation as the neighbors' mean particle orientation. With this model, Vicsek and followers\cite{chate2008collective} exhibited a wealth of intriguing patterns which attracted a lot of literature. 

In a series of papers \refcite{chate2008modeling,ginelli2010large}, Chat\'e and his team proposed a variant of the Vicsek model in which the particles interact nematically. In this case, returning to particle alignment as described above, the particle would not undertake a U-turn because the particle's direction and the mean particle direction are the same (even though in the example considered they have opposite orientations). In other words, what matters in a nematic interaction is the angle of lines between the two directions and not the angle of vectors betwen the two orientations. The word ``nematic'' originates from the physics of liquid crystals, in which this kind of interaction is a model for the excluded volume interaction between rod-like polymers.\cite{ball2017mathematics,ball2010nematic} In Ref.~\refcite{chate2008modeling,ginelli2010large}, new patterns were seen compared to the Vicsek model, which suggests that the change from polar to nematic alignment makes a big difference. The present work aims at studying nematic alignment further by means of macroscopic models. 

Macroscopic, i.e. fluid-like, models of large particle systems are important tools in the analysis of such systems. Indeed, macroscopic models consist of partial differential equations which are amenable to different kinds of qualitative and quantitative studies such as stability and bifurcation analyses, asymptotic behavior, rate of convergence towards equilibria, etc., that the discrete particle models do not allow. However, a key requirement is to derive the macroscopic models from the particle ones as rigorously as possible, otherwise results derived from the macroscopic level could lack relevance for the particle system.  

The first macroscopic version of the Vicsek model was proposed by Toner and Tu in Ref.~\refcite{toner1995long} from pure symmetry consideration. We will see below that symmetry considerations are quite important. However, Toner \& Tu's model was not -per se- derived from the Vicsek model. To overcome this question, Bertin and coworkers in Ref.~\refcite{bertin2006boltzmann} proposed a binary collision mechanism supposed to mimic the Vicsek interaction and used a Boltzmann approach to derive Toner \& Tu's model. However, beside the fact that their derivation has not been performed on the original model, the approach itself leaves a lot of unanswered mathematical questions some of which have been addressed in Ref.~\refcite{carlen2015boltzmann}. The first rigorous derivation of a macroscopic model for the Vicsek model has been performed in Ref.~\refcite{degond2008continuum} using the techniques of kinetic theory. A fully rigorous treatment of this derivation  can be found in Ref.~\refcite{jiang2016hydrodynamic} and related mathematical investigations in Ref.~\refcite{figalli2018global,gamba2016global,zhang2017local}. The resulting model now referred to as ``Self-Organized Hydrodynamics (SOH)'' is not the Toner \& Tu model, although the latter can be related to an approximation of the former by relaxation. The SOH model has been elaborated further to accommodate other kinds of interactions. A noticeable one is that performed in Ref.~\refcite{degond2017new,degond2018quaternions} for full body-attitude coordination. In particular, Ref.~\refcite{degond2018quaternions} highlights the connection between full body-attitude coordination and nematic alignment of the corresponding quaternions (body attitude can be encoded in a unit quaternion, i.e. a normalized vector in dimension $4$). In the present work, we will rely on Ref.~\refcite{degond2018quaternions} for several technical aspects. 

Using the same approach as in Ref.~\refcite{bertin2006boltzmann}, the article \refcite{peshkov2012nonlinear} proposes a model for nematically interacting particles. A similar approach based on a slightly different collision mechanism is developed in Ref.~\refcite{baskaran2008enhanced}. But these approaches suffer from the same drawback as in Ref.~\refcite{bertin2006boltzmann}: they do not start from the genuine Vicsek model for nematic particles and a rigorous mathematical framework for their derivation is still missing. We also note a mean-field approach in Ref.~\refcite{peruani2008mean}. Here, we aim to derive a macroscopic model from the genuine nematic Vicsek dynamics based on rigorous asymptotic theory in which the small parameter $\varepsilon$ is related to the change of scale from the microscopic to the macroscopic scale. We will show that the relevant scaling is a diffusive scaling by which the dilation parameter between the micro and macro time scales is $\varepsilon^{-2}$ while the corresponding dilation parameter for the spatial scales $\varepsilon^{-1}$. 

Our approach is valid for any dimension $d \geq 2$. It relies first on the derivation of an associated mean-field kinetic model and second on a diffusion approximation of that model. The derivation of the mean-field model from the particle model is not rigorous but, based on previous results in the Vicsek case,\cite{bolley2012mean} we conjecture that the former is the limit of the latter when the number of particles tends to infinity. To perform the diffusion approximation of the kinetic model, we use a classical Hilbert expansion method (see e.g. Ref.~\refcite{degond2004macroscopic} for a review, REf.~\refcite{cercignani2013mathematical} for a general presentation of mathematical kinetic theory and Ref.~\refcite{guo2010global} for a recent application of the Hilbert expansion technique). However, there are several technical difficulties. One of them lies in the inversion of the linearized collision operator (which describes the combined influence of alignment and noise within the kinetic model). As usual, solvability conditions need to be satisfied for this linearized operator to be invertible. We show that these conditions involve the so-called Generalized Collision Invariants (GCI) which were first introduced in Ref.~\refcite{degond2008continuum} to overcome the lack of momentum conservation in the model (indeed, the alignment interaction does not preserve momentum, a feature related to the self-propulsion of the particles). The fact that the GCI span the kernel of the adjoint of the linearized collision operator has already been noticed in Ref.~\refcite{aceves2019hydrodynamic} and is also verified here. The GCI for the nematic alignment collision operator were first derived in Ref.~\refcite{degond2018quaternions}. 

The macroscopic model is a system of cross-diffusion equations for the particle density $\rho$ and for the mean nematic direction $u$. The mean nematic direction is a direction of anisotropy for the system. Therefore, the local response of the system is different whether it is acted upon along the direction $u$ or across it, but such responses are equivalent when it is acted in different directions lying in the subspace $\{u\}^\bot$. Therefore, gradients need to be decomposed along the $u$ direction or across it which generates a large combinatoric complexity of different second order derivatives in the model. Likewise, gradients in $\rho$ and $u$ fuel the dynamics of the system, which results in the presence of quadratic terms in first order derivatives. Again, due to the large number of ways to multiply first order gradients in $\rho$ and $u$ decomposed in their parallel and transverse components to $u$, this results in a large combinatoric complexity of first-order terms as well. However, there is an order in this apparent complexity. This order is powered by the symmetries of the system and we will show that some combinations of derivatives which would superficially appear as possible are turned off as incompatible with the symmetries of the system. 

There are different variants of these nematic alignment models. For instance, in Ref.~\refcite{chate2006simple}, the motion of the particles is also nematic: they have a certain probability of reversing i.e. of changing the orientation of their motion along their given direction in such a way that there is no preferred orientation along the direction of motion. In Ref.~\refcite{degond2017continuum,degond2018age}, a similar model has been proposed to model colonies of myxobacteria. In this model, the reversal probability was weak compared to the nematic alignment probability. Also, a feature of the noise in the interaction term allowed the two densities of particles moving along a given direction in the two possible orientations to be different. This left the possibility of a net mean motion and consequently, the macroscopic limit was of hydrodynamic type. It led to a hyperbolic model which corresponded to a coupled system of two SOH models with identical mean directions and with reaction terms describing the reversals. Here, we do not leave the possibility to the densities of these two populations to be different. So the net mean motion is actually zero and what the macroscopic model captures are the fluctuations around this zero-average motion in the form of a diffusion system. So, the resulting model is completely different. 

We also mention Ref.~\refcite{degond2010diffusion} in which an asymptotic expansion to the solution of the kinetic Vicsek model up to the first order in $\varepsilon$ were given. This led to an SOH model perturbed by diffusion terms of order $\varepsilon$. These diffusion terms had similar structure as those presented in this paper, with a decomposition of the gradients along and normal to the mean direction of motion. This is not surprising as the structure of these terms were conditioned by the symmetries of the system which were, for the second order terms, the same as the ones we encounter here. In Ref.~\refcite{degond2010diffusion}, instead of relying on a Hilbert expansion, the methodology was based on a micro-macro decomposition. In the end, the two approaches should be equivalent and in the present paper, we chose to investigate the Hilbert expansion approach. In doing so, important structural properties were revealed, such as the relation between the GCI and the inversion of the linearized collision operator. 

As already mentioned, the body orientation model of Ref.~\refcite{degond2018quaternions} mostly corresponded to nematic alignment in dimension $4$. However, there is a major difference with the  model investigated here, which lies in the motion term. In Ref.~\refcite{degond2018quaternions}, the particle velocity was a quadratic function of the unit quaternion. Thus, two opposite quaternions gave rise to the same direction of motion. So, in this model, a net motion was achieved in average. The macroscopic limit was of hydrodynamic type and the limit model was of SOH type, i.e. was hyperbolic. Here, two opposite orientations give rise on average to no net motion. Therefore, the macroscopic limit is of diffusive type and again, completely different from what we get in Ref.~\refcite{degond2018quaternions}. 

The paper is organized as follows. In Section \ref{sec:modeling} we introduce the modelling framework, i.e. the particle model and the associated mean-field model. In Section \ref{sec:main_result}, we state the main result, i.e. Theorem \ref{th:macro} which gives the macroscopic model and discuss the properties of the model. In Section \ref{sec:main_proof}, we give the proof of the main result. Finally, in Section \ref{sec:summary}, we provide a conclusion and some perpectives.

\section{Modelling framework}
\label{sec:modeling}

\subsection{Individual Based Model}

We present a particle (or Individual-Based Model (IBM)) of collective motion where agents move at a constant speed while undergoing nematic alignment with their neighbours.
Consider $N$ agents described by their positions $X_i\in {\mathbb R}^d$ and orientations $\omega_i\in \mathbb{S}^{d-1}$, $i=1,\hdots, N$, where $\mathbb{S}^{d-1}$ is the $d-1$-sphere. In all this document, we assume $d \geq 2$.  
The evolution of the system is given by:
\begin{subequations}\label{eq:IBM}
\begin{numcases}{}
\label{eq:IBM_x}
dX_i = \omega_i dt,\\
\label{eq:IBM_omega}
d\omega_i = P_{\omega_i^\perp} \circ \left[ \nu (\omega_i \cdot \bar \omega_i) \bar \omega_i + \sqrt{2D} dB^i_t \right],
\end{numcases}
\end{subequations}
where $\nu, D>0$ are given constants, $(B_t^i)_{i=1,\hdots, N}$ denotes $N$ independent Brownian motions in ${\mathbb R}^d$ and $P_{\omega^\perp}$ denotes the orthogonal projection onto the orthogonal space to $\omega$ in ${\mathbb R}^d$ denoted by $\{\omega\}^\bot$. More generally, for any unit vector $\xi \in {\mathbb R}^d$, $|\xi|=1$, we will denote by  $P_{\xi^\perp}$ the orthogonal projection of ${\mathbb R}^d$ onto  $\{\xi\}^\bot$, namely
$$  P_{\xi^\perp} = \mbox{Id} - \xi \otimes \xi, $$
where $\otimes$ denotes the tensor product and $\mbox{Id}$ the identity matrix. The symbol '$\circ$' in Eq. \eqref{eq:IBM_omega} indicates that the stochastic differential equation (SDE)  \eqref{eq:IBM_omega} must be understood in the Stratonovich sense. Indeed, it is shown in Ref.~\refcite{hsu2002stochastic} that a SDE involving a Brownian motion projected on the tangent space to a manifold provides a Brownian motion on this manifold provided the SDE is understood in the Stratonovich sense. Finally, $\bar \omega_i$ denotes any of the two unitary leading eigenvectors of the matrix $Q_i$ defined by:
\begin{equation}  \label{eq:def_Q}
Q_i = \frac{1}{N} \sum_{j=1}^N \frac{1}{R^d} K \left( \frac{|X_i-X_j|}{R} \right)  \left( \omega_j \otimes \omega_j - \frac{1}{d}\mbox{Id} \right),
\end{equation}
where the function $K$ corresponds to a sensing kernel and $R>0$ is the typical radius of the sensing region. We assume that $K\geq 0$ and
$$\int_{{\mathbb R}^d}\frac{1}{R^d} K\left( \frac{|x|}{R} \right)\, dx =1.$$
We assume that the leading eigenvalue of $Q_i$ is simple. Therefore, there are only two unitary leading eigenvectors which are opposite to each other. However, the expression $(\omega_i \cdot \bar \omega_i) \bar \omega_i$ in Eq. \eqref{eq:IBM_omega} is independent of the choice of sign for $\bar \omega_i$ and is well-defined. System \eqref{eq:IBM} is supplemented with initial conditions, namely $(X_i(0),\omega_i(0))= (X_{i0}, \omega_{i0})$, $\forall i \in \{1, \ldots, N\}$, where $(X_{i0}, \omega_{i0})$ are points in the phase space ${\mathbb R}^d \times {\mathbb S}^{d-1}$ which are independently and identically distributed according to a probability distribution having density $f_0(x,\omega)$ with respect to the Lebesgue measure. 

Eq. \eqref{eq:IBM_x} expresses that agent $i$ moves in the direction and orientation of $\omega_i$ at speed $1$. The constancy of the speed is a way to express the particles' self-propulsion (think of fish which would be able to instantaneously adjust their stroke to maintain a constant cruising speed). The specification of a unit speed is possible by choosing a convenient ratio between the time and space units. Eq. \eqref{eq:IBM_omega} expresses how the orientation $\omega_i$ changes over time: it is the sum of two competing phenomena, a noise term given by the Brownian motion on the one hand, and an alignment term corresponding to the term involving $\bar \omega_i$. Eq. \eqref{eq:IBM_omega} without the noise term can be written 
\begin{equation}
\frac{d\omega_i}{dt}=\nu P_{\omega_i^\perp}\Big((\omega_i\cdot \bar\omega_i) \bar \omega_i\Big)= \frac{\nu}{2}\nabla_\omega \Big( (\omega_i\cdot \bar\omega_i)^2 \Big),
\label{eq:liqcrystsimpl}
\end{equation}
where $\nabla_\omega$ is the gradient in the sphere ${\mathbb S}^{d-1}$. Eq. \eqref{eq:liqcrystsimpl} describes the relaxation of the orientation $\omega_i$ towards a maximum of the potential given by $(\omega_i\cdot \bar\omega_i)^2$. These maxima corresponds to either $\bar \omega_i$ if $(\omega_i \cdot \bar \omega_i) \geq 0$ or $-\bar \omega_i$ if $(\omega_i \cdot \bar \omega_i) \leq 0$ which is what is called ``nematic alignment'' in reference to nematic liquid crystal theory. Alignment occurs with intensity $\nu$ (in the fish example above, they would change orientation at time intervals of average duration $1/\nu$). 

The direction of $\bar \omega_i$ corresponds to the mean nematic direction of the particles. Indeed, to be consistent with the fact that the particles tend to adopt the orientation of $\bar \omega_i$ or $- \bar \omega_i$ according to the sign of $(\omega_i \cdot \bar \omega_i)$, one must compute an average of the mean orientations $\omega_j$ which is invariant under the change $\omega_j \to - \omega_j$.This is the purpose of constructing the tensor $Q_i$, which is called the Q-tensor in the language of liquid crystals.\cite{ball2017mathematics,ball2010nematic}  The expression \eqref{eq:def_Q} of $Q_i$ is quadratic with respect to any of the vectors $\omega_j$ involved in the sum, and consequently respects this invariance. On the other hand, if there is only one particle $j$ involved in the sum (for instance if $K$ is compactly supported and only particle $j$ different from $i$ lies in the support of $K ( \frac{|X_i \, - \, \cdot|}{R} )$), then the nematic alignment direction should be $\pm \omega_j$. The corresponding $Q_i$ is proportional to $(\omega_j \otimes \omega_j - \frac{1}{d}\mbox{Id})$ and its leading eigenvectors are precisely $\pm \omega_j$. So, it makes sense to retain this property and refer to the mean alignment direction as the direction of the leading eigenvector of $Q_i$. 
For an alternative explanation of the relation between the $Q$-tensor and the mean nematic direction through a minimisation of a potential, the reader is referred to Ref.~\refcite[Sec. 3]{degond2018quaternions}.

This nematic alignment model is different from models encountered in liquid crystals.\cite{ball2017mathematics,ball2010nematic} Indeed, in such models, the alignment dynamics (ignoring the noise) is written: 
\begin{equation} 
\frac{d\omega_i}{dt} = \nu \, P_{\omega_i^\perp} (Q_i \omega_i). 
\label{eq:liqcryst}
\end{equation}
Definition \eqref{eq:liqcryst} is more straightforward to handle than \eqref{eq:liqcrystsimpl} as it does not impose the leading eigenvector to be simple. By a manipulation involving \eqref{eq:def_Q}, it is also easy to show that the interaction \eqref{eq:liqcryst} is additive: i.e. the total contribution to $\frac{d\omega_i}{dt}$ of all the particles is a sum of the contributions of every individual particle. Expression \eqref{eq:liqcrystsimpl} does not enjoy this additivity property. However, in most self-organization systems, interactions are not additive so it might happen that \eqref{eq:liqcrystsimpl} is more accurate to model them than \eqref{eq:liqcryst}. Furthermore, \eqref{eq:liqcrystsimpl} has the advantage to rule out any phase transition which are present with \eqref{eq:liqcryst} and are associated with a change in the multiplicity of the leading eigenvalue (we refer to Ref.~\refcite{ball2017mathematics,ball2010nematic,han2015microscopic,wanghoffman08,wangzhou11} for literature on phase transitions in liquid crystals and to Ref.~ \refcite{degond2019phase,degond2013macroscopic,degond2015phase,frouvelle2012dynamics} for the corresponding mathematical literature on the Vicsek model). The techniques developed in the present paper and notably, the GCI technique (see Section \ref{subsec:GCI}) are not yet ready to handle \eqref{eq:liqcryst} and their elaboration is still in progress. The fact that \eqref{eq:liqcrystsimpl} does not exhibit phase transitions is not a problem when one wants to focus on the dynamics of the ordered phase, which is our case here. In this case, from a phenomenological viewpoint, both models encompass the same effects and can be used to investigate them qualitatively. 
 
Here, we stress that although subject to nematic alignment, the particles are polar in their movement i.e. two particles having orientations $\omega$ and $-\omega$ move in opposite directions. Hence, system \eqref{eq:IBM} is not invariant by the reversal of the orientations $\omega_i$ of the particles. However, one may think that if there are many particles, the nematic interaction will contribute to quickly relax the distribution of $\omega$'s to a symmetric distribution, invariant by the change of $\omega$ to $-\omega$. This is indeed what we will observe in the macroscopic regime.

\subsection{Mean-field limit}

In this section, we formally establish the mean-field limit as the number of agents $N \to \infty$ of System \eqref{eq:IBM}. We construct the empirical measure $f^N(t)$ of the particles, given by
\begin{equation}
f^N(t)(x,\omega) = f^N(t,x,\omega)= \frac{1}{N}\sum_{i=1}^N \delta_{(X_i(t), \omega_i(t))}(x,\omega), 
\label{eq:empmeas}
\end{equation}
where $\delta_{(X_0, \omega_0)}(x,\omega)$ stands for the Dirac delta distribution on ${\mathbb R}^d \times {\mathbb S}^{d-1}$ located at $(X_0, \omega_0)~\in~{\mathbb R}^d~\times~{\mathbb S}^{d-1}$. We also introduce the initial measure $f^N_0$:
$$
f^N_0(x,\omega)= \frac{1}{N}\sum_{i=1}^N \delta_{(X_{i0}, \omega_{i0})}(x,\omega),
$$
such that $f^N(0) = f^N_0$. Then $f^N(t)$ is a random measure on ${\mathbb R}^d \times {\mathbb S}^{d-1}$ for all $t \geq 0$. For many kinds of particle systems, it can be shown that, as $N \to \infty$, $f^N$ converges to a deterministic measure which satisfies a partial differential equation.\cite{bolley2012mean,hauray2007n,jabin2014review} In the present case, the same result is conjectured, although the proof might be delicate due to the necessity to avoid configurations where the leading eigenvalue is multiple. So, the following result is purely formal.

\begin{proposition}[Formal mean-field limit]
The empirical distribution \eqref{eq:empmeas} converges to a function $f=f(t,x,\omega)$ which satisfies the following kinetic equation:
\begin{equation} \label{eq:kinetic_eq}
\partial_t f + \nabla_x \cdot (\omega f)  = \nabla_\omega \cdot \left[ -\nu \, (\omega \cdot \bar \omega_{R,f}) \, P_{\omega^\perp} \bar \omega_{R,f} f + D \, \nabla_\omega f \right]:=C_R(f),
\end{equation}
where $\nabla_\omega$ and $\nabla_\omega \cdot$ denote the gradient and divergence operators on $\mathbb{S}^{d-1}$, respectively, and where $\bar \omega_{R,f}$ is the unitary leading eigenvector (up to a sign) of 
\begin{equation}
Q_{R,f}(t,x):= \int_{{\mathbb R}^d}\int_{\mathbb{S}^{d-1}} \frac{1}{R^d} K\left( \frac{|x-y|}{R}\right) \left(\omega \otimes \omega - \frac{1}{d}\mbox{Id} \right)\, f\, d\omega\, dy.
\label{eq:QKf}
\end{equation}
The initial condition to \eqref{eq:kinetic_eq} is $f(0,x,\omega)=f_0(x,\omega)$.
\end{proposition}

In the language of kinetic theory, the left-hand side of \eqref{eq:kinetic_eq} is called the transport operator, and its right-hand side, namely $C_R(f)$, is the collision operator. 
In \eqref{eq:kinetic_eq}, the time-derivative is balanced by a space-derivative term which corresponds to \eqref{eq:IBM_x} and the collision operator which corresponds to \eqref{eq:IBM_omega}. In the latter, the first term is the contribution of alignment while the second one is that of the noise. The alignment term depends on the leading eigenvector of the Q-tensor $Q_{R,f}$ whose expression \eqref{eq:QKf} is a continuous version of the expression \eqref{eq:def_Q} of the discrete Q-tensor $Q_i$. 

We note that the space derivative term is antisymmetric in the transformation $\omega \to - \omega$ while the collision term is invariant under this transformation. Again, this reflects the fact that the motion of the particles is polar (i.e. depends on the orientation of $\omega$), while the alignment is nematic (i.e. depends on the direction of $\omega$ but not on its orientation). Again, we expect that, if the latter dominates, the limit model will be purely nematic. This is what we observe in the macroscopic below.

\section{The main result: macroscopic equations}
\label{sec:main_result}

\subsection{Parabolic rescaling}

The goal of this paper is to investigate the behavior of \eqref{eq:kinetic_eq} at macroscopic scales. This means that we must simultaneously dilate the space and time units so as to be able to observe the system on large regions and on large times. The dilation factor for space and time are not independent and their relation depends on the problem studied. Here, we will see that the convenient one is the so-called parabolic or diffusive rescaling, whereby the time dilation factor is quadratic in terms of the spatial dilation factor. 

More precisely, we rescale space and time in the kinetic equation \eqref{eq:kinetic_eq} by introducing a small parameter $\varepsilon \ll 1$ which corresponds to a spatial dilation factor of $1/\varepsilon$. We then introduce new time and space variables $t'$ and $x'$ by letting
$$x'=\varepsilon x, \qquad t'=\varepsilon^2 t,$$
and a new kinetic distribution function $f'(t',x',\omega)$ by
$$  f'(t',x',\omega) \, dx' \, d\omega = f(t,x,\omega) \, dx \, d\omega, $$
i.e.
$$ f'(t',x',\omega) = \frac{1}{\varepsilon^d} \, f \big( \frac{t'}{\varepsilon^2}, \frac{x'}{\varepsilon}, \omega \big). $$
This choice allows us to keep the number of particles in a given volume in phase space unchanged through the scaling. Note that we do not rescale the orientation $\omega$. Similarly, we define a rescaled Q-tensor as follows:
$$ Q'(t',x') = \frac{1}{\varepsilon^d} \, Q_{R,f} \big( \frac{t'}{\varepsilon^2}, \frac{x'}{\varepsilon} \big). $$
We easily verify that $Q'(t',x') = Q_{\varepsilon R, f'}(t',x')$. So, after removing the primes and renaming $f'$ into $f^\varepsilon$, we obtain
\begin{equation} \label{eq:rescaled_kinetic_eq_0}
\varepsilon^2 \partial_t f^\varepsilon + \varepsilon \nabla_x \cdot (\omega f^\varepsilon) = C_{\varepsilon R} (f^\varepsilon).
\end{equation}
Now, we make the key assumption that $R$ is independent of $\varepsilon$. This means that the sensing region does not change in the scaling. Note that different assumptions could be made, leading to different results.\cite{degond2013hydrodynamic} Now, we expand $Q_{\varepsilon R}$ in powers of $\varepsilon$. 

\begin{lemma} 
When $\varepsilon \to 0$, we have:
\begin{eqnarray} 
Q_{\varepsilon R,f}(t,x) &=& Q_f + \mathcal{O}(\varepsilon^2), \label{eq:expanQKf} \\
\bar \omega_{\varepsilon R ,f} &=& u_f + \mathcal{O}(\varepsilon^2), \label{eq:expanbaromKf} \\
C_{\varepsilon R} (f)  &=& \Gamma (f)  + \mathcal{O}(\varepsilon^2),  \label{eq:expanCRf}
\end{eqnarray}
where
\begin{equation} \label{eq:def_Q_f}
Q_f := \int_{\mathbb{S}^{d-1}}  \left(\omega \otimes \omega - \frac{1}{d}\mbox{Id} \right)\, f\, d\omega,  
\end{equation}
$u_f$ is one of the two normalized leading eigenvector of $Q_f$ (here too, we assume that the leading eigenvalue of $Q_f$ is simple) and
\begin{equation}
\Gamma(f) = \nabla_\omega \cdot \left[ -\nu (\omega \cdot u_f) P_{\omega^\perp}u_f\, f + D\nabla_\omega f\right] . \label{eq:defGamm}
\end{equation}
\label{lem:expan}
\end{lemma}

\begin{proof} Introduce the change of variables $y=x + \varepsilon \xi$, $\xi \in {\mathbb R}^d$ into \eqref{eq:QKf} with $R$ replaced by $\varepsilon R$ and Taylor expand with respect to $\varepsilon$. Because the kernel $K(|x|)$ is rotationally invariant, the odd powers in $\varepsilon$ vanish by antisymmetry. So, the first non-zero term following the leading order term appears with the power $\varepsilon^2$, hence the formula \eqref{eq:expanQKf}. Then \eqref{eq:expanbaromKf} follows from the Taylor expansion of a simple eigenvector of a matrix with respect to its coeficients (see also Ref.~\refcite[Prop 4.3]{degond2018quaternions}), and \eqref{eq:expanCRf} is a straightforward consequence of \eqref{eq:expanbaromKf}. 
\end{proof}

Now, inserting \eqref{eq:expanCRf} into \eqref{eq:rescaled_kinetic_eq_0} and neglecting powers of $\varepsilon$ larger than $2$ (because they will have no influence on the results) leads to the following problem: 
\begin{equation} 
\label{eq:rescaled_kinetic_eq}
\varepsilon^2 \partial_t f^\varepsilon + \varepsilon \nabla_x \cdot (\omega f^\varepsilon) = \Gamma (f^\varepsilon).
\end{equation}
This paper investigates the formal limit $\varepsilon \to 0$ in this equation.

We define 
$$\kappa:= \frac{\nu}{D}. $$
For any $u\in \mathbb{S}^{d-1}$, we introduce the probability distribution on ${\mathbb{S}^{d-1}}$ defined by 
\begin{equation} \label{eq:equilibria}
M_u(\omega)= \frac{1}{Z}\exp\left( \frac{\kappa}{2} (\omega \cdot u)^2 \right), \quad Z := \int_{\mathbb{S}^{d-1}}\exp\left( \frac{\kappa}{2} (\omega \cdot u)^2\right)\, d\omega. 
\end{equation}
We note that, using the change of variables \eqref{eq:int_spheriq} defined below, we can write $Z$ as
$$ Z = \frac{1}{W_{d-2}} \int_0^\pi \exp \Big( \frac{\kappa}{2} \cos^2 \theta \Big) \, \sin^{d-2} \theta \, d \theta, $$
(with $W_{d-2}$ a constant given by \eqref{eq:jacobspher}), which shows that $Z$ is independent of $u$ and only depends on $\kappa$. A simple computation following the remark that 
$$ \nabla_\omega \big( (\omega \cdot u)^2 \big) = 2 \, (\omega \cdot u) \, P_{\omega^\perp} u, $$
(see e.g. Ref.~\refcite{degond2018quaternions}) shows that $\Gamma$ can be written as follows: 
\begin{equation}
\Gamma(f) = D\nabla_\omega \cdot \left[M_{u_f}\nabla_\omega \left( \frac{f}{M_{u_f}}\right) \right].   
\label{eq:def_gamma}
\end{equation}
We note that $\Gamma$ can be defined as an operator on functions of $\omega$ only.

\subsection{Statement of the main result}
\label{subsec:main_result}

Before stating the main result, we need to introduce some notations. For two real numbers $\mu_1$ and $\mu_2$, we define the Hilbert space ${\mathcal H}_{\mu_1,\mu_2}$ by:
\begin{equation}
\begin{split} \label{eq:def_space_H}
{\mathcal H}_{\mu_1,\mu_2}:= \Big\{h:(-1,1)\longrightarrow {\mathbb R}, \mbox {such that } \int_{-1}^1 (1-r^2)^{\mu_1} h^2(r)\,  dr <\infty \mbox{ and }\\
\int_{-1}^1 (1-r^2)^{\mu_2}  \big(h'(r) \big)^2\, dr<\infty \Big\}.
\end{split}
\end{equation}

We define the following functions whose existence and uniqueness will be proved further:

\noindent {\tiny$\blacksquare$} $h$: $[-1,1] \to {\mathbb R}$, $r \mapsto h(r)$, is the unique solution in ${\mathcal H}_{\frac{d-1}{2},\frac{d+1}{2}}$ of the problem
\begin{equation}
\begin{split} \label{eq:ode_h} -(1-r^2)^{(d-1)/2} \exp\left( \frac{\kappa r^2}{2}\right) \left( \kappa r^2 + (d-1) \right) h(r) + \frac{d}{d r} \left[  (1-r^2)^{(d+1)/2} \exp\left( \frac{\kappa r^2}{2} \right) h'(r) \right]\\
=r\, (1-r^2)^{(d-1)/2}  \exp\left( \frac{\kappa r^2}{2}\right). 
\end{split}
\end{equation}
$h$ is an odd function of $r$ and $h(r) \leq 0$ for $r \geq 0$. 

\noindent {\tiny$\blacksquare$} $a$: $[-1,1] \to {\mathbb R}$, $r \mapsto a(r)$, is the unique solution in ${\mathcal H}_{\frac{d-1}{2},\frac{d+1}{2}}$  to Eq. \eqref{eq:ode_h} with right-hand side $(1-r^2)^{(d-1)/2}  \exp( \frac{\kappa r^2}{2} )$ (note a factor $r$ has been dropped compared to the right-hand side that defines $h$). $a$ is even and $a(r) \leq 0$, for all $r \in [-1,1]$.
 
\noindent {\tiny$\blacksquare$} $b$: $[-1,1] \to {\mathbb R}$, $r \mapsto b(r)$,  is the unique solution in ${\mathcal H}_{\frac{d-1}{2},\frac{d+1}{2}}$ to Eq. \eqref{eq:ode_h} with right-hand side $r^2 \, (1-r^2)^{(d-1)/2}  \exp ( \frac{\kappa r^2}{2} )$ (note the factor $r$ appears with exponent $2$ compared to exponent $1$ at the right-hand side of the equation that defines $h$). $b$ is even and $b(r) \leq 0$, for all $r \in [-1,1]$.

\noindent {\tiny$\blacksquare$} $c$: $[-1,1] \to {\mathbb R}$, $r \mapsto c(r)$,  is the unique solution in $\dot {\mathcal H}_{0,\frac{d-1}{2}}$ to the equation
\begin{equation}
 \frac{d}{d r} \left[  (1-r^2)^{(d-1)/2} \exp\left( \frac{\kappa r^2}{2} \right) c'(r) \right]
=r\, (1-r^2)^{(d-2)/2}  \exp\left( \frac{\kappa r^2}{2}\right), 
\label{eq:def_c}
\end{equation}
where 
\begin{equation} 
\dot {\mathcal H}_{0,\frac{d-1}{2}} = \Big\{ \varphi \in {\mathcal H}_{0,\frac{d-1}{2}} \, \, | \, \, \int_{-1}^1 \varphi (r) \, d r = 0 \Big\},
\label{eq:dotH0}
\end{equation}
$c$ is odd and $c(r) \leq 0$ for $r \geq 0$. 
 
\noindent{\tiny$\blacksquare$} $e$ is the unique solution in ${\mathcal H}_{\frac{d+1}{2},\frac{d+3}{2}}$ to the equation:
\begin{eqnarray}  
-2(1-r^2)^{(d+1)/2} \exp\left( \frac{\kappa r^2}{2}\right) \left( \kappa r^2 + d \right) e(r) + \frac{d}{d r} \left[  (1-r^2)^{(d+3)/2} \exp\left( \frac{\kappa r^2}{2} \right) e'(r) \right] \nonumber \\
=r\, (1-r^2)^{(d+1)/2}  \exp\left( \frac{\kappa r^2}{2}\right),\nonumber\\
\label{eq:def_e}
\end{eqnarray}
$e$ is odd and $e(r) \leq 0$ for $r \geq 0$.  

\noindent{\tiny$\blacksquare$} $k$
is the unique solution in $\dot {\mathcal H}_{0,\frac{d-1}{2}}$ to Eq. \eqref{eq:def_c} with right-hand side $- 2 e(r) \, (1-r^2)^{(d-2)/2}  \exp( \frac{\kappa r^2}{2})$. $k$ is odd and $k(r) \leq 0$ for $r \geq 0$. 

\medskip
For two functions $f$, $g$: $[0,\pi] \to {\mathbb R}$, with $g \geq 0$ and $\int_0^\pi g(\theta) \, d \theta >0$, we denote by $\langle f \rangle_{g}$ the average of $f$ with respect to the probability density $g(\theta ) d\theta / \int_0^\pi g(\theta) \, d \theta$, i.e.
$$ \langle f \rangle_g = \frac{\int_0^\pi f(\theta) \, g(\theta) \, d \theta}{\int_0^\pi g(\theta) \, d \theta}. $$

We now state the main result:
\begin{theorem}[Formal macroscopic limit]
\label{th:macro}
Suppose that $f^\varepsilon$ converges to $f$ as $\varepsilon\to 0$. Then, it holds that
 $$ f^\varepsilon \to \rho M_u, \quad \mbox{ with } \, \, \rho=\rho(t,x) \in [0,\infty), \quad u =u(t,x) \in {\mathbb S}^{d-1},$$ 
 where $M_u$ is given in Eq. \eqref{eq:equilibria}. If the convergence is strong enough and $\rho, \ u$ are smooth enough, then they satisfy the following system:
\begin{subequations}\label{eq:macro_equations}
\begin{empheq}[left=\empheqlbrace]{align}
\partial_t \rho &+ \nabla_x \cdot \big( C_1 \, (u \cdot \nabla_x \rho) \, u + C_2 \, P_{u^\perp} \nabla_x \rho  + C_3 \, \rho \, (u\cdot\nabla_x) u \nonumber \\
&+ C_4 \, (\nabla_x \cdot u) \, \rho u \big) = 0, \label{eq:macro_equations_rho} \\
\rho \partial_t u & + E_1 \, P_{u^\bot} \nabla_x \big( (u \cdot \nabla_x) \rho \big)
\nonumber \\
& + F_1 \, \rho \, P_{u^\bot} \big[ (u \cdot \nabla_x) \big( (u \cdot \nabla_x) u \big)\big] + F_2 \, \rho \, P_{u^\bot} \big( \nabla_x \cdot (P_{u^\bot} \nabla_x u) \big) \nonumber\\
&+ F_3 \, \rho \, P_{u^\bot} \nabla_x (\nabla_x \cdot u) \nonumber \\
& + G_1 \, ( u \cdot \nabla_x \rho ) \, (u \cdot \nabla_x) u  + G_2 \, (P_{u^\bot} \nabla_x u) (P_{u^\bot} \nabla_x \rho)   \nonumber \\
& + G_3 \, \big( (P_{u^\bot} \nabla_x \rho) \cdot P_{u^\bot} \nabla_x \big) u + G_4 \, (\nabla_x \cdot u) \, P_{u^\bot} \nabla_x \rho  \nonumber \\
& + H_1  \, (u  \cdot \nabla_x \log \rho) \, (P_{u^\bot} \nabla_x \rho)
 + H_2 \,  \rho \, (P_{u^\perp}  \nabla_x u) \big( (u \cdot \nabla_x) u \big)  \nonumber \\
& + H_3 \, \rho \big[ \big( (u \cdot \nabla_x) u \big) \cdot P_{u^\bot} \nabla_x \big] u + H_4 \, \rho \, (\nabla_x \cdot u) \, (u \cdot \nabla_x) u = 0,
\label{eq:macro_equations_u} \\
|u|&=1. \label{eq:macro_equations_|u|}
\end{empheq}
\end{subequations}
The constants $C_i$, $E_i$, $F_i$, $G_i$, $H_i$ are given by (where all functions $h$, $a$, $b$, $c$, $e$ and $k$ have argument $\cos \theta$): 
\begin{eqnarray}
C_1 &=& \big \langle c \, \cos \theta \big \rangle_q \, , \label{eq:C1_bis}\\
C_2 &=& \Big \langle \frac{1}{d-1} \, a \, \sin^2 \theta  \Big \rangle_q \, , \label{eq:C2_bis} \\
C_3 &=& \Big  \langle \frac{\kappa}{d-1}  \, b \, \sin^2 \theta \Big \rangle_q \, ,\label{eq:C3_bis}\\
C_4 &=& \Big \langle \kappa  \, \cos \theta \, \Big( e \,  \frac{\sin^2 \theta}{d-1} + k \Big) \Big \rangle_q \, , \label{eq:C4_bis} \\
E_1 &=& \Big\langle \frac{1}{\kappa} \Big( a + \frac{c}{\cos \theta} \big) \Big\rangle_s, \label{eq:E1} \\
F_1 &=& \big \langle b \big \rangle_s, \label{eq:F1} \\
F_2 &=&  \Big\langle \frac{e \, \sin^2 \theta}{(d+1) \, \cos \theta} \Big\rangle_s, \label{eq:F2} \\
F_3 &=&  \Big\langle \frac{1}{\cos \theta} \Big( \frac{2}{d+1} e \, \sin^2 \theta + k \Big) \Big\rangle_s = 2 F_2 + \Big\langle \frac{k}{\cos \theta} \Big\rangle_s , \label{eq:F3} \\
G_1 &=& \Big\langle  c \, \cos \theta + b + \frac{1}{\kappa} (c' - a) \Big\rangle_s, \label{eq:G1} \\
G_2 &=& \Big\langle - 2 \frac{a}{\kappa}  + \frac{\sin^2 \theta}{d+1} \Big( \frac{e}{\cos \theta} + a  + \frac{a'}{\kappa \, \cos \theta} \Big) \Big\rangle_s = G_3 - 2 \Big\langle \frac{a}{\kappa} \Big\rangle_s, \label{eq:G2} \\
G_3 &=& \Big\langle  \frac{\sin^2 \theta}{d+1} \Big( \frac{e}{\cos \theta} + a  + \frac{a'}{\kappa \, \cos \theta} \Big) \Big\rangle_{s}, \label{eq:G3} \\
G_4 &=& \Big\langle  \frac{k}{\cos \theta} + \frac{\sin^2 \theta}{d+1} \Big( \frac{e}{\cos \theta} + a  + \frac{a'}{\kappa \, \cos \theta} \Big) \Big\rangle_s = G_3 + F_3 - 2 F_2, \label{eq:G4} \\
H_1 &=&  \Big\langle  \frac{1}{\kappa} \Big( a + \frac{c}{\cos \theta} \Big) \Big\rangle_s = E_1, \label{eq:H1} \\
H_2 &=& \Big\langle - b - e\,\cos\theta + \frac{\sin^2 \theta}{d+1} \Big(\kappa \, e \, \cos \theta +  \kappa b + \frac{b'}{\cos \theta} + e' + \frac{e}{\cos \theta} \Big)  \Big\rangle_s \nonumber \\
&=& -F_1 + F_2 + \Big\langle- e\,\cos\theta+ \frac{\sin^2 \theta}{d+1} \Big(\kappa \, e \, \cos \theta +   \kappa b + \frac{b'}{\cos \theta} + e' \Big)  \Big\rangle_s, \label{eq:H2} \\
H_3 &=& \Big\langle -  e \, \cos \theta  +  \frac{\sin^2 \theta}{d+1} \Big( \kappa \, e \, \cos \theta + \kappa b + \frac{b'}{\cos \theta} + e' \Big) \Big\rangle_s  \nonumber \\
&=& H_2 + F_1 - F_2, \label{eq:H3} \\
H_4 &=& \Big\langle \kappa \, k \, \cos \theta + k' +  \frac{\sin^2 \theta}{d+1} \Big( \kappa \, e \, \cos \theta + \kappa b + \frac{b'}{\cos \theta} + e' \Big) \Big\rangle_s \nonumber \\
&=& H_3 + \Big\langle (\kappa \, k +e) \cos \theta + k' \Big\rangle_s, \label{eq:H4}
\end{eqnarray}
where
$$
q(\theta) = \exp(\frac{\kappa}{2} \cos^2 \theta) \sin^{d-2} \theta, \quad s(\theta) = \exp(\frac{\kappa}{2} \cos^2 \theta) \, |h(\cos \theta) \cos \theta| \, \sin^d \theta, $$
for all $ \theta \in [0,\pi]$.

\end{theorem}

\begin{remark}
The apparent singularity in the expression of some of the coefficients is only fictitious as, indeed, the probability distribution $s$ involves the same factors at the numerator and these cancel the singular factors.
\end{remark}

\subsection{Comments on System \eqref{eq:macro_equations}}
\label{subsec:comments}

System \eqref{eq:macro_equations} is a system of diffusion equations. We will leave the check of the ellipticity of the second-order differential operator for future work. However, given the sign conditions on $a$, $b$, $c$, $e$ and $k$, all $C$, $E$ and $F$ coefficients, which correspond to the second order operators involved, are positive.  
Although this is not a sufficient condition of ellipticity, this a good start, because at least, each equation on $\rho$ and $u$ separately is elliptic. Even in the case where the system is not elliptic, we might be able to fix it by incorporating additional effects, such as a different scaling of the interaction radius $R$,\cite{degond2013hydrodynamic} which may introduce stabilizing terms. Also, some instability is needed for the generation of patterns (which have been observed in simulations of the IBM.\cite{chate2008modeling,ginelli2010large}) So, a weak breakup of the ellipticity condition might just be the manifestation of the patterning capabilities of the model. We will address these points in future work. 

The model has strong structural properties. First, the normalization constraint \eqref{eq:macro_equations_|u|} is preserved for all times as soon as the initial condition satisfies it. Indeed, it is readily seen that the spatial gradient terms in \eqref{eq:macro_equations_u} are all vectors normal to $u$, so that $u$ satisfies the conservation relation $\partial_t |u|^2 = 0$. 
The system is also invariant under the change $u \to -u$. So,  if $(\rho,u)$ is a solution of the system, $(\rho,-u)$ is another one. Indeed, in \eqref{eq:macro_equations_rho}, each term involves an even number of copies of $u$, while in \eqref{eq:macro_equations_u}, each term involves an odd number of such copies. In both cases, the change $u \to -u$ leaves the equations unchanged. Thus, the orientation of $u$ is unimportant, only its direction matters. This means that we should consider $u$ as belonging to the projective space ${\mathbb P}^{d-1}$, (i.e. the quotient of the sphere ${\mathbb S}^{d-1}$ by the symmetry $u \to -u$) rather than to the sphere ${\mathbb S}^{d-1}$ itself. Since the macroscopic equations are derived in a regime where the collision operator is large, the system retains the nematic symmetry of the collision operator and ignores the disruption of this symmetry caused by the polar transport operator. 

We now comment on the structure of these equations and justify their apparent complexity. First, we note that the density equation \eqref{eq:macro_equations_rho} is in divergence (or conservative) form, i.e. it has the following structure: 
\begin{equation} 
\partial_t \rho + \nabla_x \cdot {\mathcal J} = 0, 
\label{eq:rhodiv}
\end{equation}
where ${\mathcal J}$ is the particle flux, given by the quantity inside the bracket in \eqref{eq:macro_equations_rho}. This divergence form is a consequence of the fact that the particle interactions are conservative, i.e. there is no creation or destruction of particle during an interaction. Therefore, the rate of change of the particle number in a small volume is exactly balanced by the net flux of entering particles in this volume (this flux can take negative values if there are more particles leaving that volume that entering it). This balance is what is expressed by the conservative form \eqref{eq:rhodiv} of \eqref{eq:macro_equations_rho}. On the other hand, the interactions do not conserve momentum and consequently, the equation \eqref{eq:macro_equations_u} for $u$ is in non-conservative form, and presumably cannot be put in divergence form. So, the number of terms is higher than for the $\rho$ equation (indeed, when developed, each conservative term in the $\rho$ equation would give rise to several non-conservative terms, so, the conservative form is more 'compact'). We also note that \eqref{eq:macro_equations_u} involves two kinds of terms: (i) terms which are linear in the second order derivatives (these are all terms in factor of an $E$ or $F$ coefficient) and (ii) terms that are quadratic in first order derivatives (these are all terms in factor of a $G$ or $H$ coefficient). 

Now, we comment on the structure of these terms. Due to the special role taken by self-propulsion, which occurs macroscopically in the direction of $u$, this direction is an anisotropy direction for the system. On the other hand, the system is isotropic in any direction belonging to $\{u\}^\bot$, which means that two directions belonging to $\{u\}^\bot$ should be equivalent. Therefore, we expect that the system's response to gradients in the macroscopic quantities $\rho$ and $u$ will be different for gradients along $u$ and gradients normal to $u$ but responses to gradients in directions that are normal to $u$ will be the same.  This is why all gradients have been decomposed into gradients along $u$, namely $(u \cdot \nabla_x \ldots ) \, u$ and gradients in the normal direction, namely $P_{u^\bot} \nabla_x \ldots$, where the $\ldots$ stand for any quantity that needs to be differentiated. But for second order derivatives, these terms are operated twice: these are: 
\begin{equation}
(u \cdot \nabla_x) \big( (u \cdot \nabla_x \ldots) \big), \quad  P_{u^\bot} \nabla_x \,  (u \cdot \nabla_x \ldots), \quad  P_{u^\bot} \nabla_x \,( P_{u^\bot} \nabla_x \ldots), 
\label{eq:symbols} 
\end{equation}
where these notations are purely symbolic. In each case, the exact form taken by the operator must take into account the nature of the objects to which they are applied and which they produce (scalars, vectors or tensors). Note that $(u \cdot \nabla_x) (P_{u^\bot} \nabla_x    \ldots)$ can be written as $P_{u^\bot} \nabla_x \,  (u \cdot \nabla_x \ldots)$ up to first order terms so, these two operators are not independent and we have chosen to express the cross-derivatives in terms of the latter as it makes it clear that the result is a vector normal to $u$. Indeed, the term in factor of the $E$ coefficient corresponds to \eqref{eq:symbols} applied to $\rho$ with the following correspondence 
\begin{subequations} \label{eq:decomprho}
\begin{empheq}[left=\empheqlbrace]{align}
(u \cdot \nabla_x) \big( (u \cdot \nabla_x  \ldots) \big) & \, \, \longrightarrow \, \, 
\emptyset, \label{eq:decomprho1} \\
P_{u^\bot} \nabla_x \,  (u \cdot \nabla_x  \ldots) & \, \, \longrightarrow \, \, 
E_1 \, P_{u^\bot} \nabla_x \big( (u \cdot \nabla_x) \rho \big),  \label{eq:decomprho2} \\
P_{u^\bot} \nabla_x \,( P_{u^\bot} \nabla_x  \ldots) & \, \, \longrightarrow  \, \, \emptyset. \label{eq:decomprho3}
\end{empheq}
\end{subequations}
Similarly, the terms in factor of the $F$ coefficients correspond to \eqref{eq:symbols} applied to $u$ as follows
\begin{subequations} \label{eq:decompD2u}
\begin{empheq}[left=\empheqlbrace]{align}
(u \cdot \nabla_x) \big( (u \cdot \nabla_x \ldots) \big) & \, \, \longrightarrow \, \, 
F_1 \, \rho \, P_{u^\bot} \big[ (u \cdot \nabla_x) \big( (u \cdot \nabla_x) u \big) \big], \label{eq:decompD2u1}\\
P_{u^\bot} \nabla_x \,  (u \cdot \nabla_x \ldots) & \, \, \longrightarrow \, \,  \emptyset,  \label{eq:decompD2u2}\\
P_{u^\bot} \nabla_x \,( P_{u^\bot} \nabla_x \ldots) & \, \, \longrightarrow \, \,  F_2 \, \rho \, P_{u^\bot} \big( \nabla_x \cdot (P_{u^\bot} \nabla_x u) \big)\nonumber\\
& \qquad\quad + F_3 \, \rho \, P_{u^\bot} \nabla_x (\nabla_x \cdot u). \label{eq:decompD2u3}
\end{empheq}
\end{subequations}
The last line \eqref{eq:decompD2u3} corresponds to the third term in \eqref{eq:symbols} in which a contraction or trace operation has been intercalated. Indeed, we can easily check that 
\begin{subequations} \nonumber 
\begin{empheq}[left=\empheqlbrace]{align}
P_{u^\bot} \big( \nabla_x \cdot (P_{u^\bot} \nabla_x u) \big) &= \mbox{Tr}_{[12]} \Big( (P_{u^\bot} \nabla_x) ( P_{u^\bot} \nabla_x u) \Big) \nonumber\\
&\qquad+ (((u\cdot \nabla_x)u)\cdot \nabla_x)u\nonumber \\
& \qquad  - u((P_{u^\perp}\nabla_x u):(P_{u^\perp}\nabla_x u)), \label{eq:aux_decompPDPD}\\
P_{u^\bot} \nabla_x (\nabla_x \cdot u) & = P_{u^\bot} \nabla_x \big( \mbox{Tr} (P_{u^\bot} \nabla_x u) \big). \nonumber
\end{empheq}
\end{subequations}
In the first line $(P_{u^\bot} \nabla_x) ( P_{u^\bot} \nabla_x u)$ is a tensor of order $3$ and, up to terms which involve first order derivatives only, its contraction with respect to the first two indices (hence the notation $\mbox{Tr}_{[12]}$) is equal to the left-hand side  of \eqref{eq:aux_decompPDPD}. Terms involving first order derivatives are those in factor of the $G$ and $H$ coefficients and will be discussed below. The proof of \eqref{eq:aux_decompPDPD} can be found in  \ref{sec:proof_aux_decomp}. In the second line, $(P_{u^\bot} \nabla_x u)$ is a tensor of order $2$ and we take its trace in the usual way (we will prove further that $\nabla_x~\cdot u = P_{u^\perp}~:~(\nabla_x u) = \mbox{Tr}~(P_{u^\bot}~\nabla_x~u)$, see \eqref{eq:divu}). In fact, it corresponds to contracting the third order tensor $(P_{u^\bot} \nabla_x) ( P_{u^\bot} \nabla_x u)$ with respect to the last two indices. Since this tensor is symmetric with respect to the first two indices, there is no other way to contract two of its indices. 

Now, we can explain why there are missing terms in the series \eqref{eq:decomprho} and \eqref{eq:decompD2u}. This corresponds to the fact that no operator constructed with these operators would respect the symmetries of the system. Indeed, let us analyze \eqref{eq:decomprho3} for instance. The tensor $P_{u^\bot} \nabla_x \,( P_{u^\bot} \nabla_x \rho)$ is of order 2. So it cannot be used as it is because we need a vector. The only two operations compatible with the symmetries which would give rise to a vector are presumably multiplication by $u$ (either to the right or to the left) or contraction with respect to its two indices (which would give a scalar) followed by multiplication by $u$. In the former case the result is either $0$ or a first order operator. In the second case, it leads to a vector proportional to $u$ which is not allowed since we need a vector normal to $u$ to preserve $|u|=1$. Therefore, there is no possibility to construct a genuinely second order operator from $P_{u^\bot} \nabla_x \,( P_{u^\bot} \nabla_x \rho)$ which respects the symmetries of the system. Similar considerations can be developed for the other missing lines in \eqref{eq:decomprho} and \eqref{eq:decompD2u}. To make these arguments rigorous, we need representation theory.\cite{faraut2008analysis}  This will be explored in forthcoming works. 

We now turn towards the structure of the second series of terms in \eqref{eq:macro_equations_u}, those which are quadratic in gradients of $\rho$ and $u$. Again, the gradients are decomposed along $u$ and normal to $u$, which leads to the following combination of terms: 
\begin{eqnarray*}
&& \hspace{-0.8cm} 
\big((u \cdot \nabla_x) \rho \big) \, \big((u \cdot \nabla_x) u \big) , \,   \big((u \cdot \nabla_x) \rho \big) \, (P_{u^\bot} \nabla_x u), \,  \big((u \cdot \nabla_x) u \big) \, (P_{u^\bot} \nabla_x \rho), \,   (P_{u^\bot} \nabla_x \rho) \, (P_{u^\bot} \nabla_x u), \\
&& \hspace{-0.8cm} 
\big((u \cdot \nabla_x) \rho \big)^2, \quad \big((u \cdot \nabla_x) \rho \big) \, (P_{u^\bot} \nabla_x \rho), \quad (P_{u^\bot} \nabla_x \rho)^2, \\
&& \hspace{-0.8cm} 
\big((u \cdot \nabla_x) u \big)^2, \quad ((u \cdot \nabla_x) u \big) \, (P_{u^\bot} \nabla_x u), \quad (P_{u^\bot} \nabla_x u)^2.
\end{eqnarray*}
The first line corresponds to cross-product terms of one gradient in $\rho$ and one gradient in $u$ ; the second line corresponds to quadratic terms in $\nabla_x \rho$ ; the third line to quadratic terms in $\nabla_x u$. Again, the products are taken symbolically. The exact form of the result depends on the nature of the objects involved (scalars, vectors or tensors):  
The terms in factor of the $G$ coefficients correspond to cross-product terms of one gradient in $\rho$ and one gradient in $u$ as follows: 
\begin{subequations} \nonumber 
\begin{empheq}[left=\empheqlbrace]{align*}
\big((u \cdot \nabla_x) \rho \big) \, \big((u \cdot \nabla_x) u \big) & \, \, \longrightarrow \, \, G_1 \, \big( (u \cdot \nabla_x) \rho \big) \, (u \cdot \nabla_x) u, \\
\big((u \cdot \nabla_x) \rho \big) \, (P_{u^\bot} \nabla_x u) & \, \, \longrightarrow \, \, \emptyset,   \\
\big((u \cdot \nabla_x) u \big) \, (P_{u^\bot} \nabla_x \rho) & \, \, \longrightarrow \, \,  \emptyset,   \\
(P_{u^\bot} \nabla_x \rho) \, (P_{u^\bot} \nabla_x u) & \, \, \longrightarrow \, \,   G_2 \, (P_{u^\bot} \nabla_x u) (P_{u^\bot} \nabla_x \rho) \nonumber \\
&\qquad\quad + G_3 \, \big( (P_{u^\bot} \nabla_x \rho) \cdot P_{u^\bot} \nabla_x \big) u \nonumber \\ 
& \qquad\quad + G_4 \, (\nabla_x \cdot u) \, P_{u^\bot} \nabla_x \rho.
\end{empheq}
\end{subequations}
Indeed, 
the terms in factor of $G_2$, $G_3$ and $G_4$ can be respectively written $(P_{u^\bot} \nabla_x u) (P_{u^\bot} \nabla_x \rho)$, $(P_{u^\bot} \nabla_x u)^T (P_{u^\bot} \nabla_x \rho)$ and $\mbox{Tr} (P_{u^\bot} \nabla_x u) \, P_{u^\bot} \nabla_x \rho$ and correspond to three ways to realize the symbolic operation $(P_{u^\bot} \nabla_x \rho) \, (P_{u^\bot} \nabla_x u)$ while respecting the symmetries of the system. The terms in factor of the $H$ coefficients correspond to quadratic terms in either gradients of $\rho$ or gradients of $u$ as follows: 
\begin{subequations} \nonumber 
\begin{empheq}[left=\empheqlbrace]{align*}
\big((u \cdot \nabla_x) \rho \big)^2 & \, \, \longrightarrow \, \, \emptyset, \\
\big((u \cdot \nabla_x) \rho \big) \, (P_{u^\bot} \nabla_x \rho) & \, \,  \longrightarrow \, \, H_1  \, (u  \cdot \nabla_x \log \rho) \, (P_{u^\bot} \nabla_x \rho),  \\
(P_{u^\bot} \nabla_x \rho)^2 & \, \, \longrightarrow \, \,  \emptyset,  \\
\big((u \cdot \nabla_x) u \big)^2 & \, \, \longrightarrow \, \,  \emptyset,  \\
((u \cdot \nabla_x) u \big) \, (P_{u^\bot} \nabla_x u) & \, \, \longrightarrow \, \,  H_2 \,  \rho \, (P_{u^\perp}  \nabla_x u) \big( (u \cdot \nabla_x) u \big)\nonumber \\
& \qquad\quad+ H_3 \, \rho \Big( \big( (u \cdot \nabla_x) u \big) \cdot P_{u^\bot} \nabla_x \Big) u \nonumber\\
&\qquad\quad + H_4 \, \rho \, (\nabla_x \cdot u) \, (u \cdot \nabla_x) u, \\
(P_{u^\bot} \nabla_x u)^2 & \, \, \longrightarrow \, \, \emptyset.
\end{empheq}
\end{subequations}
The terms in factor of $H_2$, $H_3$ and $H_4$ involve respectively $(P_{u^\perp}  \nabla_x u) \big( (u \cdot \nabla_x) u \big)$, $(P_{u^\perp}~\nabla_x~u)^T~\big( (u~\cdot~\nabla_x)~u \big)$, $\mbox{Tr} (P_{u^\perp}  \nabla_x u) \, \big( (u \cdot \nabla_x) u \big)$. They correspond to three ways we can multiply $(u \cdot \nabla_x) u $ and $P_{u^\perp}  \nabla_x u$, while respecting the symmetries of the system. Again, we conjecture that for the missing lines (those indicated by $\emptyset$) there is an obstruction to construct a non-trivial operator with the requirements imposed by the symmetries of the system. 

Comparatively, the structure of the $\rho$ equation \eqref{eq:macro_equations_rho} is simpler: inside the divergence, the four different gradients allowed by the symmetries of the system appear according to the following correspondence: 
\begin{subequations} \nonumber 
\begin{empheq}[left=\empheqlbrace]{align*}
(u \cdot \nabla_x) \rho & \, \, \longrightarrow \, \, C_1 \, (u \cdot \nabla_x \rho) \, u, \\
P_{u^\bot} \nabla_x \rho & \, \, \longrightarrow \, \, C_2 \, P_{u^\perp} \nabla_x \rho,  \\
(u \cdot \nabla_x) u & \, \, \longrightarrow \, \,  C_3 \, \rho \, (u\cdot\nabla_x) u, \\
P_{u^\bot} \nabla_x u & \, \, \longrightarrow \, \,  C_4 \, (\nabla_x \cdot u) \, \rho u.  
\end{empheq}
\end{subequations}
Indeed, the term in factor of $C_4$ can be written $\big( \mbox{Tr} (P_{u^\bot} \nabla_x u) \big) \, \rho u$. 

Physically, this system describes the anisotropic diffusion of a mass density $\rho$, which has different diffusivities in the direction along $u$ and normal to $u$ as the first two terms in \eqref{eq:macro_equations_rho} show. If the anisotropy direction $u$ was given and did not evolve with time, the last two terms of \eqref{eq:macro_equations_rho} would appear as convection terms for $\rho$ powered by gradients of $u$. However, the anisotropy direction $u$  is subject to a diffusion equation and these last two terms of \eqref{eq:macro_equations_rho} must be seen as cross-diffusion terms. Now, the $u$ equation is itself an anisotropic diffusion equation where the anisotropic diffusion terms in $u$ are seen in factor of the $F$ coefficients. In this equation, the cross-diffusivities, i.e. how second derivatives in $\rho$ affect $u$ are seen in factor of the $E$-coefficient. The terms in factor of $G$ and $H$ coefficients can be seen as convection terms drifting $u$ in directions depending on the various gradients in the system. 

\medskip
The rest of this article is devoted to the proof of Th. \ref{th:macro}.

\section{Proof of the main result (Th. \ref{th:macro}).}
\label{sec:main_proof}

\subsection{Preliminaries: decomposition of ${\mathbb S}^{d-1}$}
\label{sec:preliminaries_notation}

Let $u \in {\mathbb S}^{d-1}$ be given. For all $\omega \in {\mathbb S}^{d-1}$, we will use the decomposition
\begin{equation} \label{eq:space_decomposition}
 \omega= (\omega\cdot u) u + \omega_\perp, \qquad \omega_\perp := P_{u^\perp}(\omega).
\end{equation}
We note that $\omega_\perp$ depends on $u$ although not explicitly stated. The vector $u$ with respect to which the decomposition \eqref{eq:space_decomposition} is considered will be clear from the context. We will denote by $\sigma_{e,o}$ the set of functions $f=f(\omega) = f((\omega\cdot u) u + \omega_\perp)$ that are even in $(\omega\cdot u)$ and odd in $\omega_\perp$. Analogously we will define $\sigma_{o,e}$, $\sigma_{e,e}$, $\sigma_{o,o}$. Any function of $\omega$ can be decomposed uniquely into 
$$ \omega = \omega_{e,o} + \omega_{o,o} +\omega_{o,e} +\omega_{e,e}, \, \, \mbox{ with } \,  \omega_{e,o} \in \sigma_{e,o}, \, \,  \omega_{o,o} \in \sigma_{o,o} \, \, \mbox{ and so on.} $$

\medskip
Using \eqref{eq:space_decomposition}, we define the following change of variables: ${\mathbb S}^{d-1} \setminus \{ \pm u \} \to (0,\pi) \times {\mathbb S}^{d-2}$, $\omega \mapsto (\theta,z)$ such that
\begin{equation} \label{eq:change_variable_projection}
\omega \cdot u = \cos \theta, \quad \omega_\perp = \sin \theta \, z, \, \mbox{ or equivalently } \, \,  \omega = \cos \theta \, u + \sin \theta \, z ,
\end{equation}
where $\mathbb{S}^{d-2}$ is identified with $\mathbb{S}^{d-1}\cap u^\perp$. We endow unit spheres of all dimensions with their associated Lebesgue measure normalized such that the total measure of the sphere is equal to $1$. With this convention, we have 
\begin{equation}
d\omega = \frac{\sin^{d-2} \theta \, d \theta}{W_{d-2}} \, dz \, \, \mbox{ with } \, \,  W_{d-2} = \int_0^\pi \sin^{d-2} \theta \, d \theta.
\label{eq:jacobspher}
\end{equation}
We note that $W_0=\pi, \, W_1=2$ and that $W_d$ is twice the Wallis integral for integer $d$.  For any function $f=f(\omega)$, we get:
\begin{equation}
\int_{{\mathbb S}^{d-1}} f(\omega) \ d\omega = \frac{1}{W_{d-2}}\int^\pi_0 \int_{{\mathbb S}^{d-2}} f (\cos \theta \, u + \sin \theta \, z ) \ \sin^{d-2} \theta\ dz \, d\theta.
\label{eq:int_spheriq}
\end{equation}
For $d=2$, the convention is that ${\mathbb S}^{d-2}$ is just the pair of points which intersect ${\mathbb S}^1$ and the line $u^\bot$ endowed with half the counting measure. We will also use the variable $r = \cos \theta$, in which case, the change of variable formula \eqref{eq:int_spheriq} takes the form 
\begin{equation}
\int_{{\mathbb S}^{d-1}} f(\omega) \ d\omega = \frac{1}{W_{d-2}}\int_{-1}^1 \int_{{\mathbb S}^{d-2}} f (r \, u + \sqrt{1-r^2} \, z )\ (1-r^2)^{\frac{d-3}{2}} \ dz \, dr.
\label{eq:int_spheriq2}
\end{equation}

For a vector $\xi = (\xi_1, \ldots, \xi_d) \in {\mathbb R}^d$ and an integer $p \in {\mathbb N}$, we denote by $\xi^{\otimes p}$ the $p$-th tensor power of $\xi$ i.e. $\xi^{\otimes p}$ is the order-$p$ tensor defined by $(\xi^{\otimes p})_{i_1, \ldots , i_p} = \xi_{i_1} \ldots \xi_{i_p}$, \, $\forall (i_1, \ldots , i_p) \in \{1, \ldots , d\}^p$. Similarly for two order-2 tensors ${\mathcal A}  = ({\mathcal A}_{ij})_{(i,j) \in \{1, \ldots, d\}^2}$ and ${\mathcal B}  = ({\mathcal B}_{ij})_{(i,j)}$, the tensor ${\mathcal A} \otimes {\mathcal B}$ is the order-$4$ tensor defined by $({\mathcal A} \otimes {\mathcal B})_{ijk \ell} = {\mathcal A}_{ij} {\mathcal B}_{k \ell}$. Finally, if ${\mathcal T}$ is an order-$p$ tensor, $\mbox{Sym}({\mathcal T})$ is the symmetric order-$p$ tensor generated by ${\mathcal T}$ i.e. $(\mbox{Sym}({\mathcal T}))_{i_1, \ldots , i_p} = \frac{1}{p!} \sum_{\tau \in {\mathfrak S}_p} {\mathcal T}_{i_{\tau(1)}, \ldots , i_{\tau(p)}}$, with ${\mathfrak S}_p$ being the group of permutations of $p$ elements. Then, we have the following identities:  

\begin{lemma}
Let $d \geq 2$. For any function $a$: $[-1,1] \to {\mathbb R}$, $r \mapsto a(r)$, we have:  
\begin{eqnarray}
&& \hspace{-1cm}
\int_{{\mathbb S}^{d-1}} a(\omega \cdot u) \, \omega_\perp^{\otimes (2k+1)} \ d\omega =0, \quad \forall k \in {\mathbb N}, 
\label{eq:moments1} \\
&& \hspace{-1cm}
\int_{{\mathbb S}^{d-1}} a(\omega \cdot u) \, \omega_\perp \otimes \omega_\perp d\omega = \frac{1}{d-1} \int_{{\mathbb S}^{d-1}} a(\omega \cdot u) \, (1-(\omega \cdot u)^2) \, d\omega \, \,  P_{u^\perp},
\label{eq:moments2} \\
&& \hspace{-1cm}
\int_{{\mathbb S}^{d-1}} a(\omega \cdot u) \, \omega_\perp^{\otimes 4} d\omega = \frac{1}{(d-1)(d+1)} \int_{{\mathbb S}^{d-1}} a(\omega \cdot u) \, (1-(\omega \cdot u)^2)^2 \, d\omega \,  \, \Sigma,
\label{eq:moments4}
\end{eqnarray}
where $\Sigma$ is the symmetric order-$4$  tensor defined by:
$$ \Sigma = 3 \, \mbox{Sym} (P_{u^\perp} \otimes P_{u^\perp}). $$
In cartesian coordinates, $\Sigma$ is given by: 
\begin{equation}
\Sigma_{ijk\ell} = (P_{u^\perp})_{ij} (P_{u^\perp})_{k \ell} + (P_{u^\perp})_{ik} (P_{u^\perp})_{j \ell} + (P_{u^\perp})_{i \ell} (P_{u^\perp})_{jk}. 
\label{eq:defSigma}
\end{equation}
\label{lem:moments}
\end{lemma}

\begin{proof}
\eqref{eq:moments1} follows from antisymmetry. To prove \eqref{eq:moments2}, let ${\mathcal A}$ denote the matrix appearing at the left-hand side of \eqref{eq:moments2}, $(e_1, \ldots, e_{d})$ be an orthonormal basis of ${\mathbb R}^d$ with $e_d = u$ and $\omega_j = \omega \cdot e_j$ the $j$-th coordinate of $\omega$ in this basis. Then, $\omega_\perp$ has coordinates $(\omega_\perp)_j$ such that $(\omega_\perp)_j = \omega_j$, $\forall j \in \{ 1, \ldots , d-1\}$ and $(\omega_\perp)_d = 0$. In this basis, 
$$ {\mathcal A}_{ij} = \int_{{\mathbb S}^{d-1}} a(\omega \cdot u) \,  (\omega_\perp)_i \, (\omega_\perp)_j \, d\omega.$$
Since $(\omega_\perp)_d=0$, we have ${\mathcal A}_{dj} = {\mathcal A}_{id} = 0$, $\forall i,\, j \in \{1, \ldots , d \}$. We also have ${\mathcal A}_{ij} = 0$, $\forall i,j \in \{1, \ldots , d \}$, $i \not = j$ by antisymmetry through the change of variables corresponding to the exchange of the basis vectors $e_i$ and $e_j$. Finally, for $i=1, \ldots , d-1$, we have ${\mathcal A}_{ii} = {\mathcal A}_{jj}$ by rotational symmetry around $u$. Thus, for $i=1, \ldots , d-1$, we have
\begin{eqnarray*}
{\mathcal A}_{ii} &=& \int_{{\mathbb S}^{d-1}} a(\omega \cdot u) \,  \frac{1}{d-1} \sum_{j=1}^{d-1} \omega_j^2  \, d\omega  = \int_{{\mathbb S}^{d-1}} a(\omega \cdot u) \,  \frac{1}{d-1} |\omega_\perp|^2  \, d\omega \\
&=&  \frac{1}{d-1}  \int_{{\mathbb S}^{d-1}} a(\omega \cdot u) \, (1-(\omega \cdot u)^2) \,   d\omega. \\
\end{eqnarray*}
Since in the basis $(e_1, \ldots, e_d)$, the matrix $P_{u^\perp}$ has entries:
\begin{eqnarray*}
(P_{u^\perp})_{ij} &=& 0, \quad \forall  i,j \in \{1, \ldots , d \}, \quad i \not = j, \\
(P_{u^\perp})_{dd} &=& 0, \\
(P_{u^\perp})_{ii} &=& 1, \quad \forall  i \in \{1, \ldots , d-1 \}, 
\end{eqnarray*}
Eq. \eqref{eq:moments2} follows. 

We now prove \eqref{eq:moments4}. We denote by $S$ the order-$4$ symmetric tensor at the left-hand side of \eqref{eq:moments4}. In the basis $(e_1, \ldots , e_d)$, we have 
$$ 
S_{ijk\ell} = \int_{{\mathbb S}^{d-1}} a(\omega \cdot u) \, (\omega_\perp)_i \, (\omega_\perp)_j \, (\omega_\perp)_k \, (\omega_\perp)_\ell \, d\omega. 
$$
Using the same arguments as for ${\mathcal A}$, we get that $S_{ijk\ell} = 0$ when $d \in \{i,j,k,\ell\}$, or when one of the values $1$, $\ldots$, $d-1$ of the four integers $i$, $j$, $k$, $\ell$ is taken an odd number of times. So, there are two cases where $S_{ijk\ell} \not = 0$: either one of the values $1$, $\ldots$, $d-1$ is taken four times, corresponding to a term of the form $S_{iiii}$ with $i \in \{1, \ldots , d-1 \}$, or two values $1$, $\ldots$, $d-1$ are taken twice each, corresponding to a term of the form $S_{iijj}$, $S_{ijij}$ or $S_{ijji}$ with $i, j \in \{1, \ldots , d-1 \}$, $i \not = j$. Furthermore by rotational symmetry,  
$$ 
S_{iiii} = S_{1111} =  \int_{{\mathbb S}^{d-1}} a(\omega \cdot u) \, \omega_1^4 \,  \, d\omega, \quad \forall i \in \{1, \ldots , d-1 \}, 
$$
and 
$$ 
S_{iijj} = S_{ijij} =S_{ijji} = S_{1122} = \int_{{\mathbb S}^{d-1}} a(\omega \cdot u) \, \omega_1^2 \, \omega_2^2  \, d\omega, \quad \forall i ,\, j \in \{1, \ldots , d-1 \}, \quad i \not = j. 
$$
If $d \geq 3$, there is a relation between $S_{1111}$ and $S_{1122}$ because, again by rotational symmetry
\begin{eqnarray}
S_{1122} &=& \int_{{\mathbb S}^{d-1}} a(\omega \cdot u) \, \omega_1^2 \, \Big( \frac{1}{d-2} \sum_{j=2}^{d-1} \omega_j^2 \Big) \, d\omega \nonumber \\
&=& \frac{1}{d-2} \int_{{\mathbb S}^{d-1}} a(\omega \cdot u) \, \omega_1^2 \, (|\omega_\perp|^2 - \omega_1^2)  \, d\omega \nonumber \\
&=& \frac{1}{d-2} \int_{{\mathbb S}^{d-1}} a(\omega \cdot u) \, \Big( \frac{1}{d-1} \sum_{j=1}^{d-1} \omega_j^2 \Big) \, |\omega_\perp|^2 \, d \omega - \frac{1}{d-2} \int_{{\mathbb S}^{d-1}} a(\omega \cdot u) \, \omega_1^4 \, d \omega \nonumber \\
&=& \frac{1}{(d-2)(d-1)} \int_{{\mathbb S}^{d-1}} a(\omega \cdot u) \, |\omega_\perp|^4 \, d \omega - \frac{1}{d-2} S_{1111} \nonumber \\
&=& \frac{1}{(d-2)(d-1)} \int_{{\mathbb S}^{d-1}} a(\omega \cdot u) \, (1 - (\omega \cdot u)^2)^2 \, d \omega - \frac{1}{d-2} S_{1111}. \label{eq:relS}
\end{eqnarray} 
Now, we compute $S_{1111}$ using the change of variables \eqref{eq:int_spheriq}. We have 
\begin{eqnarray*}
S_{1111} &=& \int_0^\pi \int_{{\mathbb S}^{d-2}} a(\cos\theta) \, (\sin \theta \, z_1)^4 \,  \frac{\sin^{d-2} \theta \, d \theta}{W_{d-2}} \, dz, 
\end{eqnarray*} 
where $z_i$ are the coordinates of $z$ in the basis $(e_1, \ldots , e_d)$ (with $z_d=0$). Using again the change of variables \eqref{eq:int_spheriq} but on ${\mathbb S}^{d-2}$ this time, using $e_1$ as the polar vector, we have, in dimension $d \geq 4$: 
\begin{eqnarray*}
\int_{{\mathbb S}^{d-2}}  z_1^4 \, dz &=& \int_0^\pi \cos^4 \theta' \, \frac{\sin^{d-3} \theta' \, d \theta'}{W_{d-3}} , 
\end{eqnarray*} 
and after two rounds of integrations by parts, we get 
\begin{eqnarray*}
\int_0^\pi \cos^4 \theta' \, \sin^{d-3} \theta' \, d \theta' = \frac{3}{d(d-2)} W_{d+1}. 
\end{eqnarray*} 
Thus 
\begin{eqnarray*}
S_{1111} &=& \frac{3}{d(d-2)} \frac{W_{d+1}}{W_{d-3}} \, \int_0^\pi a(\cos \theta) \, \sin^4 \theta \,  \frac{\sin^{d-2} \theta \,d \theta}{W_{d-2}} \, dz, 
\end{eqnarray*} 
Using the usual recursion for Wallis's integrals: $W_{d+1} = \frac{d}{d+1} W_{d-1}$, we get 
\begin{eqnarray}
S_{1111} &=& \frac{3}{(d-1)(d+1)} \, \int_0^\pi a(\cos \theta) \, (1 - \cos^2 \theta)^2 \,  \frac{\sin^{d-2} \theta \,d \theta}{W_{d-2}} \nonumber \\
&=&
\frac{3}{(d-1)(d+1)} \, \int_{{\mathbb S}^{d-1}} a(\omega \cdot u) \, (1 - (\omega \cdot u)^2)^2 \,d \omega. \label{eq:S1111}
\end{eqnarray} 
Now, using \eqref{eq:relS}, we get 
\begin{eqnarray}
S_{1122} &=&
\frac{1}{(d-1)(d+1)} \, \int_{{\mathbb S}^{d-1}} a(\omega \cdot u) \, (1 - (\omega \cdot u)^2)^2 \,d \omega, \label{eq:S1122}
\end{eqnarray}
which, with \eqref{eq:S1111}, yields $ S_{1111} = 3 S_{1122}$. Now, a careful inspection shows that $\Sigma$ has the same zero terms as $S$ and that its non-zero terms satisfy 
\begin{eqnarray*}
 &&\Sigma_{iijj} = \Sigma_{ijij}  =  \Sigma_{ijji} = 1, \quad \forall i, \, j \in \{1, \ldots , d-1 \}, \quad i \not = j, \\
&&\Sigma_{iiii} = 3, \quad \forall i \in \{1, \ldots , d-1 \}.
\end{eqnarray*}
Therefore, $S$ and $\Sigma$ are proportional and the proportionality coefficient is $S_{1122}$ given by \eqref{eq:S1122}, which yields \eqref{eq:moments4}. A straightforward inspection of the cases $d=2$ and $d=3$ shows that \eqref{eq:moments4} is still valid in these cases. 
\end{proof}

\subsection{Properties of the operator $\Gamma$}

\begin{proposition}[Properties of the operator $\Gamma$]
\label{prop:properties_Gamma}
We have the following properties:
\begin{itemize}
\item[(i)] Entropy dissipation: the following inequality holds: 
\begin{equation}
H(f):= \int_{{\mathbb S}^{d-1}} \Gamma (f) \, \frac{f}{M_{u_f}} \, d \omega = - D \int_{{\mathbb S}^{d-1}} \Big| \nabla_\omega \Big( \frac{f}{M_{u_f}} \Big) \Big|^2 \, M_{u_f } \, d \omega \leq 0. 
\label{eq:entropy}
\end{equation}
\item[(ii)] Consistency relation: $u$ is the leading eigenvector (up to a sign) of 
$$Q_{M_u}= \int_{\mathbb{S}^{d-1}} M_u(\omega) \, \left( \omega \otimes \omega - \frac{1}{d}\mbox{Id} \right) \, d\omega.$$
\item[(iii)] Equilibria: the set ${\mathcal E}$ of functions $f=f(\omega) \geq 0$ such that $\Gamma(f) = 0$ are given by  
\begin{equation}
{\mathcal E}=\{ \rho M_u\, |\, \rho \in [0,\infty),\, u\in \mathbb{S}^{d-1} \}.
\label{eq:ker_Gamma}
\end{equation}
\end{itemize}
\end{proposition}

The proof of this Proposition can be found in Ref.~\refcite[Prop 4.4]{degond2018quaternions} in the case $d=4$. We summarize the proof for a generic $d$ below for the reader's convenience. 

\begin{proof} (i): \eqref{eq:entropy} follows upon multiplying \eqref{eq:def_gamma} by $f/M_{u_f}$, integrating with respect to $\omega$ and using Stokes formula. 

(ii): using \eqref{eq:space_decomposition} we have
$$
Q_{M_u} u = \int_{\mathbb{S}^{d-1}} M_u \, (\omega \cdot u) \, [(\omega \cdot u) \, u + \omega_\perp] \, d\omega - \frac{u}{d}.
$$
But the term proportional to $\omega_\perp$ in the integral vanishes by antisymmetry. So, it only remains 
\begin{equation}
Q_{M_u} u = \lambda_\parallel \, u, \qquad \lambda_\parallel := \int_{\mathbb{S}^{d-1}} M_u \, (\omega \cdot u)^2 \, d\omega - \frac{1}{d}.
\label{eq:lambdapar}
\end{equation}
Now, taking $\xi \in {\mathbb R}^d$ such that $\xi \cdot u = 0$, we have 
$$
Q_{M_u} \xi = \int_{\mathbb{S}^{d-1}} M_u \, (\omega_\perp \cdot \xi) \, [(\omega \cdot u) \, u + \omega_\perp] \, d\omega - \frac{\xi}{d}, 
$$
and now the first term in the integral vanishes by antisymmetry. Then using \eqref{eq:moments2}, we get
$$Q_{M_u} \xi = \Big( \int_{\mathbb{S}^{d-1}} M_u \, (\omega_\perp \otimes \omega_\perp) \, d\omega \Big) \, \xi- \frac{\xi}{d} = \lambda_\bot \, \xi, 
$$
with  
\begin{equation}
\lambda_\bot := \frac{1}{d-1} \int_{\mathbb{S}^{d-1}} M_u \, (1-(\omega \cdot u)^2) \, d\omega - \frac{1}{d} = \frac{1}{d-1} \Big(1 - \big(\lambda_\parallel + \frac{1}{d})\Big) - \frac{1}{d} = - \frac{\lambda_\parallel}{d-1}. 
\label{eq:lambdabot}
\end{equation}
Therefore, $\lambda_\parallel$ is a simple eigenvalue associated with eigenvector $u$ while $\lambda_\bot$ is an eigenvalue of multiplicity $d-1$ associated with any vector orthogonal to $u$. To show that $u$ is the leading eigenvalue, it suffices to show that $\lambda_\parallel >0$. Using \eqref{eq:int_spheriq} and integrating by parts, we have 
\begin{eqnarray*} 
\int_{\mathbb{S}^{d-1}} e^{\frac{\kappa}{2} (\omega \cdot u)^2} \, (\omega \cdot u)^2 \, d\omega
&=& \int_0^\pi  e^{\frac{\kappa}{2} \cos^2 \theta} \, \cos^2 \theta \, \frac{\sin^{d-2} \theta \, d \theta}{W_{d-2}} \\
&=& \frac{1}{d-1} \int_0^\pi  e^{\frac{\kappa}{2} \cos^2 \theta} \, (1- \cos^2 \theta + \kappa \sin^2 \theta \cos^2 \theta ) \, \frac{\sin^{d-2} \theta \, d \theta}{W_{d-2}} \\
&>& \frac{1}{d-1} \int_0^\pi  e^{\frac{\kappa}{2} \cos^2 \theta} \, (1- \cos^2 \theta) \, \frac{\sin^{d-2} \theta \, d \theta}{W_{d-2}}.
\end{eqnarray*}
It follows that 
$$
\int_{\mathbb{S}^{d-1}} M_u \, (\omega \cdot u)^2 \, d\omega
> \frac{1}{d-1} \, (1- \int_{\mathbb{S}^{d-1}} M_u \, (\omega \cdot u)^2 \, d\omega) , 
$$
which is equivalent to 
$$
\int_{\mathbb{S}^{d-1}} M_u \, (\omega \cdot u)^2 \, d\omega
> \frac{1}{d}, 
$$
i.e. $\lambda_\parallel >0$. 

(iii): suppose that $f \in {\mathcal E}$. Then by \eqref{eq:entropy}, it follows that $f/M_{u_f}$ is a constant, which shows that there exist $\rho >0$ and $u \in \mathbb{S}^{d-1}$ such that $f = \rho M_u$. Conversely, suppose that $f = \rho M_u$. Then it obviously satisfies 
\begin{equation}
D\nabla_\omega \cdot \left[M_u \nabla_\omega \left( \frac{f}{M_u} \right) \right] = 0. 
\label{eq:gammarhoMu}
\end{equation}
But since $u$ is the leading eigenvalue of $Q_f$, we have $u = u_f$ and, upon substituting $u_f$ for $u$ into \eqref{eq:gammarhoMu}, we get $\Gamma(f) = \Gamma(\rho M_u) = 0$, showing \eqref{eq:ker_Gamma}. \end{proof}

\subsection{The Generalised Collision Invariant}
\label{subsec:GCI}

In Ref. \refcite{degond2008continuum} a new methodology was introduced through the concept of the Generalised Collision Invariant. This method was develop to coarse-grain non-conserved quantities, like the mean orientation in the Vicsek model. In this section we introduce this concept and main properties which will be used in the sequel. This section extends Ref.~\refcite[Sect. 4.3]{degond2018quaternions} to a generic dimension $d$ (Ref.~\refcite{degond2018quaternions} was restricted to the case $d=4$.)

\medskip
Collision invariants are fundamental in the derivation of macroscopic equations. They are defined as the scalar  functions $\psi = \psi(\omega)$ such that
\begin{equation} 
\label{eq:collision_invariant}
\int_{{\mathbb S}^{d-1}} \Gamma(f) \, \psi \, d\omega=0, \qquad \forall \mbox{ functions } f. 
\end{equation}
In the present case, $\psi=$constant clearly satisfies this relation. This is a consequence of Stokes' formula (in mathematical terms) or of the conservation of mass during the interactions between agents (in physical terms). It can be shown that there are no other collision invariants. This implies, particularly, that the dimension of the space of collision invariants is smaller than the dimension of the set of equilibria ${\mathcal E}$ (from \eqref{eq:ker_Gamma}, it follows that ${\mathcal E}$ is a nonlinear manifold of dimension $d$). Classical methods require the dimension of the space of collision invariants (they obviously form a vector space) to be the same as the dimension of the manifold of equilibria in order to enable the derivation of a closed system of macroscopic equations. The collision invariants corresponding to the constants will allow us to derive the equation for the spatial density $\rho=\int f d\omega$, but it will not be enough to determine the equation for the mean direction $u$. To sort out this problem, the concept of Generalised Collision Invariant (GCI) has been introduced in Ref.~\refcite{degond2008continuum}. 

To define the GCI, we first need to define a new operator $\bar \Gamma$ as follows:  

\begin{definition}
Let $u \in {\mathbb S}^{d-1}$ be given. The operator $\bar \Gamma (f,u)$ is defined by 
\begin{equation} \label{eq:def_bar_gamma}
\bar \Gamma(f,u) := D\,\nabla_\omega \cdot \left[ M_{u} \nabla_\omega\left(\frac{f}{M_{u}} \right)\right].
\end{equation}
\label{def:barGam}
\end{definition}

With this definition, we have
\begin{equation}
\Gamma(f) = \bar \Gamma(f, u_f).
\label{eq:relgambargam}
\end{equation}
Note that $\bar \Gamma(f, u_f)$ is not the linearization of $\Gamma$. It is rather the action of $\Gamma$ when one 'freezes' the parameter $u_f$ to the value $u$. Below, we will elaborate more on the relation between $\bar \Gamma$ and the linearization of $\Gamma$. Now, we can define the GCI: 

\begin{definition}
Let $u \in {\mathbb S}^{d-1}$ be given. A function $\psi$: ${\mathbb S}^{d-1} \to {\mathbb R}$ is called a `Generalised Collision Invariant (GCI)' associated to $u$ if and only if
\begin{equation}
\int_{{\mathbb S}^{d-1}}\bar \Gamma(f, u) \, \psi \, d\omega = 0,\quad \mbox{for all } f \mbox{ such that } P_{u^\perp} (Q_f \, u)=0.
\label{eq:GCI_def}
\end{equation}
\label{def:GCI}
\end{definition}

The condition on $f$ in \eqref{eq:GCI_def} means that $u$ is an eigenvector of $Q_f$. Since $u_f$ is the leading eigenvector of $Q_f$, we have $P_{u_f^\perp} (Q_f \, u_f)=0$ and consequently if $\psi$ is a GCI associated with $u_f$, we have 
\begin{equation}
\int_{{\mathbb S}^{d-1}} \Gamma(f) \, \psi \, d\omega =\int_{{\mathbb S}^{d-1}}\bar \Gamma(f, u_f) \,  \psi \, d\omega = 0. 
\label{eq:GCICI}
\end{equation}
Therefore, $\psi$ is 'like' a collision invariant except that it depends on $f$ through its dependence on~$u_f$. In the next proposition, we characterize the GCI. First, we introduce the formal $L^2$ adjoint of $\bar \Gamma (\cdot, u)$. For $\psi = \psi(\omega)$,  $\bar \Gamma (\psi, u)$ is defined as follows:  
\begin{equation}
\bar \Gamma^* (\psi,u) := \frac{1}{M_u} \nabla_\omega \cdot [ M_u \nabla_\omega \psi].
\label{eq:bargGamstar}
\end{equation}
We will also denote by $\{u\}^\bot$ the orthogonal space to $u$ in ${\mathbb R}^d$.

\begin{proposition}[Generalised Collision Invariant] 
\label{lem:GCI}
(i) Given $u \in {\mathbb S}^{d-1}$, introduce the function $\vec\psi_u$: ${\mathbb S}^{d-1} \ni \omega \mapsto \vec\psi_u(\omega) \in {\mathbb R}^d$, defined as the unique (componentwise) solution of
\begin{equation} 
\bar \Gamma^* (\vec\psi_u,u) (\omega) = P_{u^\perp}\omega \ (\omega\cdot u),
\label{eq:vecGCI}
\end{equation}
in the Hilbert space 
\begin{equation}
H^1_0 (\mathbb{S}^{d-1}) = \Big\{ \varphi\in H^1(\mathbb{S}^{d-1}) \mbox{ such that } \int_{\mathbb{S}^{d-1}}\varphi(\omega)\, d\omega=0 \Big\}. 
\label{eq:def_H10}
\end{equation}
$\vec\psi_u$ is called the vector GCI. The set ${\mathcal G}_u$ of GCIs associated to $u$ is given by 
\begin{equation}
{\mathcal G}_u = \big\{ B \cdot \vec\psi_{u}+C \, \, | \, \, B \in \{u\}^\bot, \, \, C \in {\mathbb R} \big\}. 
\label{eq:setofGCI}
\end{equation}
$\vec\psi_u$ is odd in both $(\omega \cdot u)$ and $\omega_\perp$, so, $\vec\psi_u \in \sigma_{o,o}$ in the sense of Section \ref{sec:preliminaries_notation}. \\
(ii) The vector GCI $\vec\psi_{u}$ is written:
\begin{equation}
\label{eq:GCI_explicit}
\vec{\psi}_{u}(\omega)=P_{u^\perp}\omega\  h(\omega\cdot u),
\end{equation}
where the function $h$ is the unique solution in ${\mathcal H}_{\frac{d-1}{2},\frac{d+1}{2}}$ of the equation \eqref{eq:ode_h} (with ${\mathcal H}_{\mu_1,\mu_2}$ defined at \eqref{eq:def_space_H}, see Section \ref{subsec:main_result}). We recall that $h$ is an odd function of $r$ and $h(r) \leq 0$ for $r \geq 0$. \\
(iii) For a given function $f: {\mathbb S}^{d-1} \to {\mathbb R}$, we consider 
\begin{equation} \label{eq:GCI}
\vec\psi_{u_f}(\omega)= P_{u_f^\perp} \omega\ h(\omega \cdot u_f ),
\end{equation}
then $\vec \psi_{u_f}$ satisfies
\begin{equation}
\label{eq:keyCGI}
\int_{{\mathbb S}^{d-1}} \Gamma(f) (\omega) \,  \vec \psi_{u_f} (\omega) \, d\omega =0.
\end{equation}
\end{proposition}

\begin{proof} This statement is the generalization to an arbitrary dimension $d$ of Ref.~\refcite[Sect. 4.3]{degond2018quaternions}. We summarize it here for the sake of completeness. 

We first show that $\psi$ is a GCI associated with $u$ if and only if there exists $B \in \{u\}^\bot$ such that  
\begin{equation}
\bar \Gamma^*(\psi,u) (\omega) = (B \cdot \omega) \, (\omega \cdot u), \, \, \forall \omega \in {\mathbb S}^{d-1}. 
\label{eq:GCIchar1} 
\end{equation}
Indeed, using the formal adjoint $\bar \Gamma^*$ of $\bar \Gamma$ given by \eqref{eq:bargGamstar}, and Eq. \eqref{eq:def_Q_f} to develop the condition $P_{u^\perp} (Q_f \, u)=0$, the definition \eqref{eq:GCI_def} for $\psi$  can be written:
\begin{eqnarray*}
 \int_{{\mathbb S}^{d-1}} f \, \bar \Gamma^*(\psi,u) \, d \omega = 0, && \mbox{ for all } \, \, f \, \, \mbox{ such that } \\
&& \int_{{\mathbb S}^{d-1}} f \, (B \cdot \omega) \, (\omega \cdot u) \, d \omega = 0, \, \, \forall B \in \{u\}^\bot. 
\end{eqnarray*}
This leads to 
$$ 
\bar \Gamma^*(\psi,u) \in \{ (B \cdot \omega) \, (\omega \cdot u) \, \, | \, \, B \in \{u\}^\bot \}, 
$$
because this set being finite dimensional, it is closed and so, equal to its bi-orthogonal. This proves the claim.  

We now determine the solutions of \eqref{eq:GCIchar1}. We interpret this equation in the weak sense through the classical variational formulation: find $\psi \in H^1 (\mathbb{S}^{d-1})$ such that 
\begin{equation}
\int_{{\mathbb S}^{d-1}} M_u \, \nabla_\omega \psi \cdot \nabla_\omega \phi \, d\omega = - \int_{{\mathbb S}^{d-1}} M_u \, (B \cdot \omega) \, (\omega \cdot u) \, \phi \, d \omega, \, \, \forall \phi \in H^1 (\mathbb{S}^{d-1}). 
\label{eq:varfor}
\end{equation}
We first show that we can restrict the set of test functions $\phi$ to $H^1_0 (\mathbb{S}^{d-1})$. Indeed, suppose $\psi~\in~H^1 (\mathbb{S}^{d-1})$ is a solution of \eqref{eq:varfor} for all test functions $\phi \in H^1_0 (\mathbb{S}^{d-1})$. Now, take $\phi \in H^1 (\mathbb{S}^{d-1})$ and construct $\tilde \phi = \phi - \int_{{\mathbb S}^{d-1}} \phi \, d \omega$. Then $\tilde \phi \in H^1_0 (\mathbb{S}^{d-1})$ and we can use it as a test function. But since $\int_{{\mathbb S}^{d-1}} \phi \, d \omega$ is a constant and $ \int_{{\mathbb S}^{d-1}} M_u \, (B \cdot \omega) \, (u \cdot \omega) \, d \omega = 0$ 
by antisymmetry, the contribution of $\int_{{\mathbb S}^{d-1}} \phi \, d \omega$ vanishes in both sides of \eqref{eq:varfor} and we obtain that \eqref{eq:varfor} is also valid when tested against $\phi$. We now look for $\psi$ in $H^1_0 (\mathbb{S}^{d-1})$ such that \eqref{eq:varfor} holds for all $\phi \in H^1_0 (\mathbb{S}^{d-1})$. By the Poincar\'e Wirtinger inequality, the bilinear form at the left-hand side of \eqref{eq:varfor} is coercive on $H^1_0 (\mathbb{S}^{d-1})$. So, Lax-Milgram's theorem applies and shows that there exists a unique solution to this variational problem in $H^1_0 ({\mathbb S}^{d-1})$. Now, if we have two solutions $\psi_1$ and $\psi_2$ in $H^1 ({\mathbb S}^{d-1})$ of \eqref{eq:varfor}, the difference $\psi_1 - \psi_2$ satisfies \eqref{eq:varfor} with right-hand side equal to $0$. Using $\psi_1 - \psi_2$ as a test function, we deduce that $\int_{{\mathbb S}^{d-1}} M_u \, |\nabla_\omega (\psi_1-\psi_2)|^2 \, d\omega = 0$, which implies that $\psi_1-\psi_2$ is a constant. It follows that any solution of \eqref{eq:varfor} is equal to the unique solution of \eqref{eq:varfor} in $H^1_0 (\mathbb{S}^{d-1})$ up to an additive constant. Now, denote by $\psi_B$ the unique solution of \eqref{eq:varfor} in $H^1_0 ({\mathbb S}^{d-1})$. For a fixed $\omega$ the map $\{u\}^\bot \to {\mathbb R}$, $B \mapsto \psi_B(\omega)$ is a continuous linear form. So, by Riesz's theorem, there exists a vector in $\{u\}^\bot$ denoted by $\vec \psi_u(\omega)$ such that $\psi_B(\omega) = \vec \psi_u(\omega) \cdot B$. From what precedes, it follows that $\vec \psi_u(\omega)$ is the unique componentwise solution of \eqref{eq:vecGCI} in $H^1_0 ({\mathbb S}^{d-1})$. Finally, any GCI assocated to $u$ is of the form $\psi_B + C$ with $B \in \{u\}^\bot$ and $C$ in ${\mathbb R}$, which leads to \eqref{eq:setofGCI}.  

Now, for any $B \in \{u\}^\bot$, we show that 
\begin{equation} 
\psi_B(\omega) = (B \cdot \omega) \, h(\omega \cdot u), 
\label{eq:AnsatzpsiB}
\end{equation} 
with $h$ the unique solution in ${\mathcal H}_{\frac{d-1}{2},\frac{d+1}{2}}$ of \eqref{eq:ode_h}, which will prove \eqref{eq:GCI_explicit}. We note that \eqref{eq:GCIchar1} can be written
\begin{eqnarray}
\bar \Gamma^*(\psi,u):= \frac{\kappa}{2} \nabla_\omega \big( (\omega \cdot u)^2 \big) \cdot \nabla_\omega \psi + \Delta_\omega \psi = (B \cdot \omega) \, (\omega \cdot u). 
\label{eq:pbmpsiB}
\end{eqnarray}
To insert Ansatz \eqref{eq:AnsatzpsiB} into \eqref{eq:pbmpsiB}, we note the following identities\cite{degond2018quaternions}:
\begin{eqnarray*}
&& \nabla_\omega (\omega \cdot u) \cdot \nabla_\omega (\omega \cdot B) = - (\omega \cdot u) \, (\omega \cdot B), \\
&& |\nabla_\omega (\omega \cdot u)|^2 = 1 - (\omega \cdot u)^2, \\
&& \Delta_\omega (\omega \cdot u) = - (d-1) (\omega \cdot u), \qquad \Delta_\omega (\omega \cdot B) = - (d-1) (\omega \cdot B). 
\end{eqnarray*}
The last equalities come from the fact that $(\omega \cdot u)$ is a spherical harmonics of degree 1.\cite{faraut2008analysis} After some tedious but straightforward computations, we end up with 
\begin{eqnarray*}
\bar \Gamma^*\big((B \cdot \omega) \, h(\omega \cdot u),u\big) &=& (B \cdot \omega) \Big\{ h''(\omega \cdot u) \,  \big(1 - (\omega \cdot u)^2\big) \\
&&+ h'(\omega \cdot u) \, (\omega \cdot u) \, \big[ \kappa  \,  \big(1 - (\omega \cdot u)^2\big) - (d+1) \big] \\  
&&+ h(\omega \cdot u) \big[ - \kappa \, (\omega \cdot u)^2 - (d-1) \big] \Big\} \\
&=& (\omega \cdot B) \, (\omega \cdot u). 
\end{eqnarray*}
Therefore, $(\omega \cdot B)$ can be simplified and introducing $r = (\omega \cdot u) \in [-1,1]$, we obtain the equation for $h$: 
\begin{eqnarray*}
(1-r^2) h'' + \big( \kappa \, (1-r^2) - (d+1) \big) \, r h' - \big( \kappa \, r^2 + (d-1) \big) \, h = r. 
\end{eqnarray*}
A straightforward integration factor technique leads to \eqref{eq:ode_h}. 

Eq. \eqref{eq:ode_h} has variational formulation given by: find $h \in {\mathcal H}_{\frac{d-1}{2},\frac{d+1}{2}}$ such that 
\begin{eqnarray}
\int_{-1}^1 (1-r^2)^{\frac{d+1}{2}} \, e^{\frac{\kappa r^2}{2}} \, h'(r) \, \ell'(r) \, dr &+& 
\int_{-1}^1 (1-r^2)^{\frac{d-1}{2}} \, (\kappa r^2+(d-1))\,  e^{\frac{\kappa r^2}{2}} \, h(r) \, \ell(r) \, dr \nonumber \\
&=& - \int_{-1}^1 r \, (1-r^2)^{\frac{d-1}{2}} \, e^{\frac{\kappa r^2}{2}} \, \ell(r) \, dr, \quad \forall \ell \in {\mathcal H}_{\frac{d-1}{2},\frac{d+1}{2}}, \nonumber\\
\label{eq:varforh}
\end{eqnarray}
and since the functions $\exp (\kappa \, r^2/2)$ and $(\kappa r^2 + (d-1)) \, \exp (\kappa \, r^2/2)$ are bounded from above and below, the bilinear form at the left hand side of \eqref{eq:varforh} is coercive on ${\mathcal H}_{\frac{d-1}{2},\frac{d+1}{2}}$. Therefore, Lax-Milgram's theorem applies and gives a unique solution $h \in {\mathcal H}_{\frac{d-1}{2},\frac{d+1}{2}}$ to \eqref{eq:varforh}. Furthermore, since the operator at the left-hand side of \eqref{eq:ode_h} is invariant by the change $r \to -r$ and the right-hand side of \eqref{eq:ode_h} is an odd function, by the uniqueness of the solution, it follows that $h$ is odd. Finally, since the right-hand side of \eqref{eq:ode_h} is nonnegative on $[0,1]$ and thanks to the maximum principle applied on $[0,1]$, $h$ itself is nonpositive on $[0,1]$. 

It remains to show that, with $h$ in ${\mathcal H}_{\frac{d-1}{2},\frac{d+1}{2}}$, $\tilde \psi_B$ given by \eqref{eq:AnsatzpsiB} belongs to $H^1_0 ({\mathbb S}^{d-1})$. Indeed, by the uniqueness of the solution of \eqref{eq:GCIchar1} in $H^1_0 ({\mathbb S}^{d-1})$, it will follow that $\tilde \psi_B = \psi_B$, hence finishing to show the validity of \eqref{eq:AnsatzpsiB}. Using \eqref{eq:int_spheriq2}, we have 
\begin{eqnarray*}
\int_{{\mathbb S}^{d-1}} |\tilde \psi_B(\omega)|^2 \, d \omega &=& \int_{{\mathbb S}^{d-1}} |
\omega \cdot B|^2 \, |h(\omega \cdot u)|^2 \, d \omega \\
&=& \frac{|B|^2}{W_{d-2}} \int_{-1}^1 (1-r^2) \, |h(r)|^2 \, (1-r^2)^{\frac{d-3}{2}} \, dr\\
&=& \frac{|B|^2}{W_{d-2}} \int_{-1}^1 \, |h(r)|^2 \, (1-r^2)^{\frac{d-1}{2}} \, dr <\infty, 
\end{eqnarray*}
because $h \in {\mathcal H}_{\frac{d-1}{2},\frac{d+1}{2}}$. 
Then, 
$$ \nabla_\omega \tilde \psi_B(\omega) = P_{\omega^\bot} B \, h(\omega \cdot u) + (B \cdot \omega) \, h'(\omega \cdot u) \, P_{\omega^\bot}u := \Xi_1 + \Xi_2. $$
We have: 
\begin{eqnarray*}
\int_{{\mathbb S}^{d-1}} |\Xi_2|^2 \, d \omega &=& \int_{{\mathbb S}^{d-1}} |
\omega \cdot B|^2 \, |P_{\omega^\bot}u|^2 \, |h'(\omega \cdot u)|^2 \, d \omega \\
&=& \frac{|B|^2}{W_{d-2}} \int_{-1}^1 (1-r^2) \, (1-r^2) \,  |h'(r)|^2 \, (1-r^2)^{\frac{d-3}{2}} \, dr\\
&=& \frac{|B|^2}{W_{d-2}} \int_{-1}^1 \, |h'(r)|^2 \, (1-r^2)^{\frac{d+1}{2}} \, dr <\infty, 
\end{eqnarray*}
again because $h \in {\mathcal H}_{\frac{d-1}{2},\frac{d+1}{2}}$. Now, 
\begin{eqnarray*}
\int_{{\mathbb S}^{d-1}} |\Xi_1|^2 \, d \omega &=& \int_{{\mathbb S}^{d-1}} \, |P_{\omega^\bot}B|^2 \, |h(\omega \cdot u)|^2 \, d \omega \\
&=& \frac{|B|^2}{W_{d-2}} \int_{-1}^1 r^2 \,  |h(r)|^2 \, (1-r^2)^{\frac{d-3}{2}} \, dr.  
\end{eqnarray*}
Integrating by parts, we compute: 
\begin{eqnarray*}
J &=& \int_{-1}^1 r^2 \,  |h(r)|^2 \, (1-r^2)^{\frac{d-3}{2}} \, dr \\
&=& \big[ - \frac{r \, h^2(r) \,  (1-r^2)^{\frac{d-1}{2}}}{d-1} \Big]_{-1}^1 + \frac{1}{d-1} \int_{-1}^1  \big(rh(r)^2\big)'  \, (1-r^2)^{\frac{d-1}{2}} \, dr \\
&\leq & \frac{1}{d-1} \int_{-1}^1  \big(rh(r)^2\big)'  \, (1-r^2)^{\frac{d-1}{2}} \, dr , 
\end{eqnarray*}
where we have used that $[ - \frac{r \, h^2(r) \,  (1-r^2)^{\frac{d-1}{2}}}{d-1} ]_{-1}^1 \leq 0$. In fact, it is not clear that this term is finite. So, for complete rigour, we should consider the integral on $[-1+\delta, 1-\delta]$ and let $\delta  \to 0$ in the end. We skip this step and refer to Ref.~\refcite[Sect. 4.3]{degond2018quaternions} for details. Then, using Cauchy-Schwarz and Young's inequality $2ab \leq a^2/\eta + \eta b^2$ with $\eta>1$, we get: 
\begin{eqnarray*}
J &\leq & \frac{1}{d-1} \Big\{ \int_{-1}^1 h(r)^2\  \, (1-r^2)^{\frac{d-1}{2}} \, dr 
+ 2 \int_{-1}^1 r\, hh'(r)\  \, (1-r^2)^{\frac{d-1}{2}} \, dr \Big\}\\
&\leq& \frac{1}{d-1} \Big\{ \int_{-1}^1 h(r)^2\  \, (1-r^2)^{\frac{d-1}{2}} \, dr 
+ \eta \int_{-1}^1 h'(r)^2 \, (1-r^2)^{\frac{d+1}{2}} \, dr + \frac{1}{\eta} J \Big\}, 
\end{eqnarray*}
hence, 
\begin{eqnarray*}
J &\leq & \frac{\eta^2}{\eta(d-1)-1} \Big\{ \int_{-1}^1 h(r)^2\  \, (1-r^2)^{\frac{d-1}{2}} \, dr 
+ \int_{-1}^1 h'(r)^2 \, (1-r^2)^{\frac{d+1}{2}} \, dr \Big\} <\infty,
\end{eqnarray*}
and $\eta^2/(\eta(d-1)-1)>0$ since $\eta>1$.
Besides, since $\tilde \psi_B$ is odd with respect to $\omega_\perp$, its integral over ${\mathbb S}^{d-1}$ vanishes. Thus, we can conclude that $\tilde \psi_B  \in H^1_0 ({\mathbb S}^{d-1})$ and consequently, that $\tilde \psi_B = \psi_B$. 

Finally, since $h$ is odd with respect to $\omega \cdot u$, we get that $\vec \psi_{u}$ is odd with repect to both $\omega_\perp$ and $\omega \cdot u$ and thus belongs to $\sigma_{o,o}$. All these considerations complete the proof of (i) and (ii). Finally,  (iii) is just rephrasing \eqref{eq:GCICI}.  \end{proof}

\subsection{Hilbert expansion and inversion of the linearized collision operator}

We introduce the Hilbert expansion for $f^\varepsilon$:
\begin{equation} \label{eq:hilbert_expansion}
f^\varepsilon= f_0+\varepsilon f_1+\varepsilon^2 f_2 +\mathcal{O}(\varepsilon^3),
\end{equation}
where $f_i=f_i(t,x,\omega)$, $i=1,2,3$, are independent of $\varepsilon$. 
 Inserting the expansion for $f^\varepsilon$ in \eqref{eq:rescaled_kinetic_eq} we obtain:
\begin{eqnarray} \nonumber
&& \hspace{-1cm}
\varepsilon^2 \left( \partial_t f_0 + \mathcal{O}(\varepsilon) \right) + \varepsilon ( \omega\cdot \nabla_x) \left( f_0 + \varepsilon f_1 + \mathcal{O}(\varepsilon^2) \right)= \Gamma\Big(f_0 + \varepsilon f_1+\varepsilon^2 f_2 + \mathcal{O}(\varepsilon^3)\Big). \label{eq:expansion_equation}
\end{eqnarray}
Now, we can Taylor expand the operator $\Gamma$ about $f_0$ as follows:
\begin{equation} \label{eq:expansion_Gamma}
\Gamma(f_0 +\varepsilon f_1 + \varepsilon^2 f_2 + \mathcal{O}(\varepsilon^3)) = \Gamma(f_0) + \varepsilon D_{f_0}\Gamma (f_1)+ \varepsilon^2 (D_{f_0}\Gamma(f_2)+ \frac{1}{2} D^2_{f_0}\Gamma(f_1, f_1) ) + \mathcal{O}(\varepsilon^3),
\end{equation}
where $D_{f_0}\Gamma (f_1)$ denotes the first derivative of $\Gamma$ at $f_0$ acting on $f_1$ and $D^2_{f_0}\Gamma(f_1, f_1)$, the second derivative of $\Gamma$ at $f_0$ acting on the pair $(f_1,f_1)$. 
Using this expansion and identifying equal powers of $\varepsilon$ in \eqref{eq:expansion_equation} we have at each order the following equations:
\begin{eqnarray}
\mathcal{O}(\varepsilon^0):\quad && \Gamma(f_0)=0, \label{eq:order0}\\
\mathcal{O}(\varepsilon^1):\quad && D_{f_0}\Gamma(f_1) = (\omega\cdot \nabla_x) f_0, \label{eq:order1}\\
\mathcal{O}(\varepsilon^2): \quad && D_{f_0}\Gamma(f_2) = \partial_t f_0 + (\omega\cdot \nabla_x) f_1 - \frac{1}{2} D^2_{f_0} \Gamma(f_1, f_1).\label{eq:order2}
\end{eqnarray}
From equation \eqref{eq:order0}, using Prop. \ref{prop:properties_Gamma} we conclude that there exists $\rho_0=\rho_0(t,x)$, $u_0=u_0(t,x)$ such that
\begin{equation} \label{eq:f_0}
f_0(t,x,\omega) = \rho_0(t,x) M_{u_0(t,x)}(\omega), \quad \forall (t,x,\omega) \in [0,\infty) \times {\mathbb R}^d \times {\mathbb S}^{d-1}.
\end{equation}
Now, to investigate equations \eqref{eq:order1} and \eqref{eq:order2} we need to study the solvability of the equation
\begin{equation} \label{eq:generic}
D_{\rho_0 M_{u_0}}\Gamma(f) =g,
\end{equation}
where $g$ is a given function. We note that, like $\Gamma$, $D_{\rho_0 M_{u_0}}\Gamma$ operates on functions depending on $\omega$ only. So, we will determine under which conditions on a function $g=g(\omega)$, there exists a solution $f=f(\omega)$ of \eqref{eq:generic}. This is answered in the following:

\begin{theorem}[Inversion of the linearized operator $D_{\rho_0 M_{u_0}}\Gamma$]
\label{th:solvability condition}
(i) Let $(\rho_0,u_0) \in [0,\infty) \times {\mathbb S}^{d-1}$ and $g\in L^2(\mathbb{S}^{d-1})$. There exists $f \in  H^{1}(\mathbb{S}^{d-1})$
 such that Eq. \eqref{eq:generic} holds if and only if $g$ satisfies the solvability conditions:
\begin{equation} \label{eq:solvability condition} 
\int_{\mathbb{S}^{d-1}} g(\omega)\, d\omega =0, \quad \int_{\mathbb{S}^{d-1}} g(\omega)\, \vec{\psi}_{u_0}(\omega)\, d\omega=0.
\end{equation}
(ii) If condition \eqref{eq:solvability condition} is satisfied, Eq. \eqref{eq:generic} has a unique solution $f$ satisfying the two properties 
\begin{subequations} \label{eq:condforunique} 
\begin{empheq}[left=\empheqlbrace]{align}
&f \in \dot H^1_0({\mathbb S}^{d-1}),  \label{eq:condforunique_rho} \\
& P_{u_0^\bot} (Q_f u_0) = 0, \label{eq:intfomu0omper}
\end{empheq}
\end{subequations}
where 
$$ \dot H^1_0 (\mathbb{S}^{d-1}) = \Big\{ \varphi \in H^1 (\mathbb{S}^{d-1}) \, \, \Big| \, \, \int_{\mathbb{S}^{d-1}} \frac{\varphi}{M_{u_0}} \, d \omega = 0 \Big\}. $$
This solution is also the unique solution to the problem
\begin{equation}
\bar \Gamma(f, u_0) =g,
\label{eq:barGamf=g}
\end{equation}
in $\dot H^1_0 (\mathbb{S}^{d-1})$ (where $\bar \Gamma$ is defined in \eqref{eq:def_bar_gamma}) and conversely, the unique solution to \eqref{eq:barGamf=g} in $\dot H^1_0 (\mathbb{S}^{d-1})$ is also the unique solution to \eqref{eq:generic} satisfying the two conditions \eqref{eq:condforunique}. \\
(iii) If $f$ is the above solution, the set ${\mathcal S}_{u_0}$ of all solutions of \eqref{eq:generic} in $H^1 (\mathbb{S}^{d-1})$  is given by
\begin{equation}
{\mathcal S}_{u_0} = \big\{ f+ M_{u_0} \left(\hat \rho + (\omega \cdot u_0) (\omega \cdot \hat u) \right) \, \, | \, \, 
\hat \rho \in {\mathbb R}, \, \, \hat u \in \{u_0\}^\bot \big\}.
\label{eq:express_Su0}
\end{equation}
\end{theorem}

The proof of this theorem will be done through a succession of Lemmas. We start with:

\begin{lemma}[\eqref{eq:solvability condition} is a necessary condition]
\label{lem:Df_against_GCI}
Let $f_0= \rho_0 M_{u_0}$ with $\rho_0 >0$ and $u_0 \in {\mathbb S}^{d-1}$. For all functions $f_1 = f_1(\omega)$ it holds that
\begin{equation} \label{eq:condition_D} 
\int_{{\mathbb S}^{d-1}} D_{f_0}\Gamma(f_1)\, d\omega=0 \quad \mbox{and} \quad 
\int_{{\mathbb S}^{d-1}} D_{f_0}\Gamma(f_1)\, \vec \psi_{u_0}\, d\omega=0,
\end{equation}
where $\vec \psi_{u_0}$ is the vector GCI defined in Proposition \ref{lem:GCI}. As a consequence, conditions \eqref{eq:solvability condition} are necessary conditions for the solvability of \eqref{eq:generic}.
\end{lemma}

\begin{proof} Let $f^\varepsilon= f_0 + \varepsilon f_1$ be a variation of $f_0$ along $f_1$ (here $\varepsilon$ stands for an arbitrary small parameter, not necessarily the parameter involved in the parabolic rescaling). Let $u^\varepsilon:= u_{f^\varepsilon}$ be the unit leading eigenvector (up to a sign) of $Q_{f^\varepsilon}$. Assume that the choice of the sign of $u^\varepsilon$ is made continuously with $\varepsilon$. Thanks to the divergence form of $\Gamma$ and to \eqref{eq:keyCGI}, we have for all $\varepsilon$: 
$$\int_{{\mathbb S}^{d-1}} \Gamma(f^\varepsilon) \, d\omega =  0 \quad \mbox{and} \quad 
\int_{{\mathbb S}^{d-1}} \Gamma(f^\varepsilon) \, \vec \psi_{u^\varepsilon} d\omega = \int_{{\mathbb S}^{d-1}} \bar \Gamma(f^\varepsilon, u^\varepsilon) \, \vec \psi_{u^\varepsilon} d\omega = 0.$$
Expanding these expressions with respect to $\varepsilon$ and using that $\Gamma(f_0) = 0$, 
we get \eqref{eq:condition_D}. \end{proof}

We now show that conditions \eqref{eq:solvability condition} are also sufficient conditions for the solvability of \eqref{eq:generic}. We first prove the 

\begin{lemma}[Equation for $u_1$]
\label{lem:equation_u_1}
Consider $f_0 = \rho_0 M_{u_0}$ with $\rho_0 >0$ and $u_0 \in {\mathbb S}^{d-1}$. Let $f^\varepsilon= f_0 + \varepsilon f_1$ be a variation of $f_0$ with an arbitrary first order variation $f_1=f_1(\omega)$. Let $u^\varepsilon:= u_{f^\varepsilon}$ be the unit leading eigenvector (up to a sign) of $Q_{f^\varepsilon}$. Assume that the choice of the sign of $u^\varepsilon$ is made continuously with $\varepsilon$. Then $u^\varepsilon$ has the following expansion
\begin{equation}
u^\varepsilon = u_0 + \varepsilon u_1 + {\mathcal O}(\varepsilon^2), 
\label{eq:expansion_u}
\end{equation}
where 
\begin{equation}
u_1 = \frac{d-1}{d \, \lambda_\parallel \rho_0} \, P_{u_0^\perp}(Q_{f_1}u_0),
\label{eq:system_u1}
\end{equation}
where $\lambda_\parallel$ is the leading eigenvalue of $Q_{f_0}$ given by \eqref{eq:lambdapar} (see proof of Proposition \ref{prop:properties_Gamma} (ii)). 
\end{lemma}

\begin{proof}
The vector $u^\varepsilon$ is a unit eigenvector of $ Q_{f^\varepsilon}$. Any unit eigenvector fulfills: 
$$ |u^\varepsilon|^2 = 1 \quad \mbox{ and } \quad P_{(u^\varepsilon)^\perp} Q_{f^\varepsilon} u^\varepsilon =0.$$ 
Inserting \eqref{eq:expansion_u} into these equations, we obtain:
\begin{subequations} \nonumber
\begin{numcases}{}
|u_0|^2 +2\varepsilon u_0\cdot u_1 +  \mathcal{O}(\varepsilon^2) = 1, \nonumber \\
 P_{(u_0+ \varepsilon u_1+ \mathcal{O}(\varepsilon^2))^{\perp}} \, Q_{f_0+\varepsilon f_1}  (u_0+\varepsilon u_1 + \mathcal{O}(\varepsilon^2)) =0. \nonumber
\end{numcases}
\end{subequations}
We note that $Q_f$ is linear with respect to $f$ so that $ Q_{f_0+\varepsilon f_1} = Q_{f_0} + \varepsilon Q_{f_1}$. An easy computation shows that 
$$ P_{(u_0+ \varepsilon u_1+ \mathcal{O}(\varepsilon^2))^{\perp}} = P_{u_0^{\perp}} - \varepsilon (u_0 \otimes u_1 + u_1 \otimes u_0) +\mathcal{O}(\varepsilon^2). $$
Now, using that $ |u_0|=1$ and $P_{u_0^\perp} Q_{f_0} u_0=0$, we obtain :
\begin{subequations} \label{eq:aux_system_u1}
\begin{numcases}{}
u_0\cdot u_1 =0,  \label{eq:aux_system_u1-1}\\
- (u_0 \otimes u_1 + u_1 \otimes u_0) Q_{f_0} u_0 +  P_{u_0^\perp}(Q_{f_1}u_0) + P_{u_0^\perp}(Q_{f_0}u_1) =0,  \label{eq:aux_system_u1-2}
\end{numcases}
\end{subequations}
Since $u_0$ is a normalized eigenvector of $Q_{f_0}$ associated with the eigenvalue $\rho_0\lambda_\parallel$ and $u_1 \in \{u_0\}^\bot$ by \eqref{eq:aux_system_u1-1}, we have 
$$(u_0 \otimes u_1 + u_1 \otimes u_0) Q_{f_0} u_0 = \rho_0\lambda_\parallel \big( (u_1 \cdot u_0) \, u_0 + (u_0 \cdot u_0)  \, u_1 \big) = \rho_0\lambda_\parallel u_1. $$
Besides, $\{u_0\}^\bot$ is the eigenspace of $Q_{f_0}$ associated to the eigenvalue $\rho_0\lambda_\bot$ given by \eqref{eq:lambdabot}. Therefore: 
$$ P_{u_0^\perp}(Q_{f_0}u_1) = Q_{f_0}u_1 = \rho_0\lambda_\bot u_1. $$  
Thus, \eqref{eq:aux_system_u1-2} gives 
\begin{equation} \label{eq:aux_u1_def}
\rho_0(\lambda_\parallel - \lambda_\bot) u_1 = P_{u_0^\perp}(Q_{f_1}u_0).
\end{equation}
 With \eqref{eq:lambdabot}, this leads to \eqref{eq:system_u1}.\end{proof}

\begin{lemma}[Linearised operator]
\label{lem:linearisation}
Let $f_0= \rho_0 M_{u_0}$ with $\rho_0 >0$ and $u_0 \in {\mathbb S}^{d-1}$. For all functions $f_1 = f_1(\omega)$ it holds that
\begin{equation} \label{eq:linearGam} 
D_{f_0} \Gamma (f_1) = \bar \Gamma (f_1,u_0) - \kappa \nabla_\omega \cdot \big[ f_0 \nabla_\omega \big( (\omega \cdot u_0) \, (\omega \cdot u_1) \big) \big], 
\end{equation}
where $u_1$ is related to $f_1$ through \eqref{eq:system_u1}. 
\end{lemma}

\begin{proof}
Again, let $f^\varepsilon= f_0 + \varepsilon f_1$ be a variation of $f_0$ along $f_1$ with $u^\varepsilon$ associated with $f^\varepsilon$ like in Lemma \ref{lem:equation_u_1}. 
We have:
\begin{equation} \label{eq:expansion D_f lemma}
D_{f_0} \Gamma(f_1) = \left. \frac{\partial \bar \Gamma}{\partial f}\right|_{(f_0,u_0)}(f_1) + \left. \frac{\partial \bar \Gamma}{\partial u}\right|_{(f_0, u_0)} (u_1).
\end{equation}
Indeed, \eqref{eq:expansion D_f lemma} follows from identifying the terms of order $\varepsilon$ in the expansion of 
$$
\Gamma(\rho_0 M_{u_0}+\varepsilon f_1) = \bar \Gamma(f_0 + \varepsilon f_1, u_0 + \varepsilon u_1 + {\mathcal O}(\varepsilon^2)),
$$
with respect to $\varepsilon$. Next, since $\bar \Gamma(f,u)$ is linear with respect to $f$, we have 
$$ \left. \frac{\partial \bar \Gamma}{\partial f}\right|_{(f_0,u_0)}(f_1) = \bar \Gamma (f_1, u_0). $$
Now, since
$$  \bar \Gamma(f,u) = \Delta_\omega f - \frac{\kappa}{2} \nabla_\omega \big[ f \, \nabla_\omega \big( (\omega \cdot u )^2\big) \big], $$
we get
$$ \left. \frac{\partial \bar \Gamma}{\partial u}\right|_{(f_0, u_0)} (u_1) = - \kappa \nabla_\omega \cdot \big[ f_0 \nabla_\omega \big( (\omega \cdot u_0) \, (\omega \cdot u_1) \big) \big]. $$
which leads to the result. \end{proof}

\begin{lemma}[Existence of solutions to \eqref{eq:barGamf=g}]
\label{lem:solution_Gamma_bar}
Let $u_0 \in {\mathbb S}^{d-1}$ be given. Assume that the function $g~\in~L^2({\mathbb S}^{d-1})$ satisfies
$$\int_{{\mathbb S}^{d-1}} g\, d\omega=0,$$
then, there exists a unique solution $f\in \dot H^1_0({\mathbb S}^{d-1})$, of \eqref{eq:barGamf=g}. The set of solutions of \eqref{eq:barGamf=g} in $H^1({\mathbb S}^{d-1})$ is $\{ f + C M_{u_0} \, \, | \, \, C \in {\mathbb R} \}$. 
\end{lemma}

\begin{proof}
By the change of functions $f = M_{u_0} \tilde f$, $g = M_{u_0} \tilde g$, we are led to an equation of the form \eqref{eq:GCIchar1} (with $\tilde g$ replacing the right-hand side of \eqref{eq:GCIchar1}). The existence theory for Eq. \eqref{eq:GCIchar1} developed in the proof of Proposition \ref{lem:GCI} directly gives the result. \end{proof}

\begin{lemma}[\eqref{eq:solvability condition} is a sufficient condition]
\label{lem:solvcond_sufficient}
Let $(\rho_0,u_0) \in [0,\infty) \times {\mathbb S}^{d-1}$ and $g\in L^2(\mathbb{S}^{d-1})$. If $g$ satisfies the solvability conditions \eqref{eq:solvability condition}, there exists $f \in  \dot H^1_0 (\mathbb{S}^{d-1})$ such that Eq. \eqref{eq:generic} holds. Furthermore, this solution is also a solution to the problem \eqref{eq:barGamf=g} and it additionally satisfies \eqref{eq:intfomu0omper}. 
\end{lemma}

\begin{proof}
Since $g$ satisfies the first condition \eqref{eq:solvability condition}, we can apply Lemma \ref{lem:solution_Gamma_bar} and define $f_1$, the unique solution in $\dot H^1_0({\mathbb S}^{d-1})$ of \eqref{eq:barGamf=g}. Now, by Lemma \ref{lem:linearisation}, $D_{f_0} \Gamma (f_1)$ is given by \eqref{eq:linearGam}  where $u_1$ is related to $f_1$ by \eqref{eq:system_u1}. We now show that $u_1$ = 0, which will show that $f_1$ is also a solution to \eqref{eq:generic} and prove the Lemma. Using the second condition \eqref{eq:solvability condition} together with \eqref{eq:vecGCI}, we get
\begin{eqnarray*}
0&=& \int_{{\mathbb S}^{d-1}} g \, \vec \psi_{u_0} \, d\omega = \int_{{\mathbb S}^{d-1}} \bar \Gamma(f_1,u_0) \, \vec \psi_{u_0} \, d\omega  \\
&=& \int_{{\mathbb S}^{d-1}} f_1 \, \bar \Gamma^*(\vec \psi_{u_0},u_0) \, d\omega = \int_{{\mathbb S}^{d-1}} f_1 \, P_{u_0^\perp} \, \omega\,  (\omega \cdot u_0) \,  d\omega  \\
&=& P_{u_0^\perp} \Big( \int_{{\mathbb S}^{d-1}} f_1 \, (\omega \otimes \omega) \, d\omega \Big) u_0. 
\end{eqnarray*}
It follows that $P_{u_0^\perp} (Q_{f_1} u_0) = 0$. Then, by \eqref{eq:system_u1}, $u_1=0$.   
\end{proof}

\begin{lemma}[Solutions to the homogeneous problem]
\label{lem:solutions_homogeneous_equation}
Let $(\rho_0,u_0) \in [0,\infty) \times {\mathbb S}^{d-1}$ and $f_0 = \rho_0 M_{u_0}$. All solutions $f_1 \in H^1({\mathbb S}^{d-1})$ to the homogeneous equation
\begin{equation} \label{eq:homogeneous equation}
D_{f_0} \Gamma (f_1)=0,
\end{equation}
are of the form
\begin{equation} \label{eq:solution homogeneous equation}
f_1=M_{u_0}\big( \hat \rho + (\omega\cdot u_0) (\omega \cdot \hat u) \big),
\end{equation}
for some $\hat \rho\in {\mathbb R}$ and $\hat u\in \{u_0\}^\bot$. 
\end{lemma}

\begin{proof}

First we prove that if $f_1$ is of the form \eqref{eq:solution homogeneous equation}, then $u_1$ given by \eqref{eq:system_u1} has the expression
\be \label{eq:hat_u}
\hat u= \kappa \rho_0 u_1.
\ee
Indeed, it is a straightforward computation thanks to \eqref{eq:moments2} to check that
\beqarl
P_{u_0^\perp}(Q_{f_1}u_0) &=&  \lp \int_{\mathbb{S}^{d-1}} (\omega\cdot u_0)^2\,  \omega_\bot \otimes \omega_\bot \, M_{u_0}\, d\omega \rp \hat u + \hat \rho P_{u_0^\perp}Q_{M_{u_0}}u_0 \nonumber\\
&=& \frac{1}{(d-1)}\lp\int_0^\pi\cos^2\theta \sin^d\theta\, M_{u_0}(\theta)\, \frac{d\theta}{W_{d-2}}  \rp\, \hat u := \tilde \lambda\, \hat u, \label{eq:Pu0Qhatf1u0_aux}
\eeqarl
where we have used that $P_{u_0^\perp}Q_{M_{u_0}}u_0= P_{u_0^\perp}(\lambda_\parallel u_0) =0$. Now, we express the value of $\tilde \lambda$ in terms of $\lambda_\parallel$ and $\lambda_\bot$, which are given in \eqref{eq:lambdapar}, \eqref{eq:lambdabot}, respectively. Integrating by parts, we have
\beqar
\lambda_\parallel + \frac{1}{d} &=& \int_{\mathbb{S}^{d-1}} M_{u_0}\, (\omega\cdot u_0)^2\, d\omega = \int^\pi_0 M_{u_0}(\theta) \cos^2\theta\ \sin^{d-2}\theta\  \, \frac{d\theta}{W_{d-2}}\\
&=& \frac{1}{d-1}\int^{\pi}_0 \sin\theta \ M_{u_0}(\theta)(\kappa \cos^2\theta +1) \sin^{d-1}\theta\, \frac{d\theta}{W_{d-2}}\\
&=& \frac{\kappa}{d-1}\int^\pi_0 \cos^2\theta\,\sin^d\theta\, M_{u_0}(\theta)\, \frac{d\theta}{W_{d-2}} + \frac{1}{d-1}\int^\pi_0 \sin^d\theta\, M_{u_0}(\theta) \, \frac{d\theta}{W_{d-2}}\\
&=& \frac{\kappa}{d-1}\int^\pi_0 \cos^2\theta\,\sin^d\theta\, M_{u_0}(\theta)\, \frac{d\theta}{W_{d-2}} + \frac{1}{d-1}\int_{\mathbb{S}^{d-1}}M_{u_0}\, (1-(\omega\cdot u_0)^2) \, d\omega\\
&=& \kappa\,  \tilde \lambda + \lambda_{\bot}+\frac{1}{d}.
\eeqar
Therefore, we have that $\tilde \lambda = (\lambda_\parallel-\lambda_\bot)/\kappa$. Substituting the expression for $\tilde \lambda$ in \eqref{eq:Pu0Qhatf1u0_aux}, we conclude that
\be
P_{u_0^\perp}(Q_{f_1}u_0)= \frac{1}{\kappa}(\lambda_\parallel - \lambda_\bot) \hat u. \label{eq:Pu0Qhatf1u0}
\ee
 Finally, thanks to expression \eqref{eq:aux_u1_def} we get \eqref{eq:hat_u}.

Next we show that
 if $f_1$ is solution to \eqref{eq:homogeneous equation}, then it is of the form \eqref{eq:solution homogeneous equation}. Let $f_1$ be solution to \eqref{eq:homogeneous equation}. Then, by Lemma \ref{lem:linearisation} we have 
\be \label{eq:alternative_homogeneous}
D_{f_0}\Gamma(f_1)=\bar \Gamma( f_1, u_0) - \kappa \nabla_\omega\cdot[\rho_0 M_{u_0}\nabla_\omega\lp (\omega \cdot u_0) (\omega\cdot u_1) \rp]=0,
\ee
where the relation between $u_1$ and $f_1$ is given by \eqref{eq:system_u1}. This last expression can be recast into
\be \label{eq:aux_recast_homogeneous}
\nabla_\omega\cdot \lp M_{u_0}\nabla_\omega \lp\frac{ f_1 -\kappa\rho_0 M_{u_0}\lp (\omega\cdot u_0) (\omega\cdot u_1) \rp}{M_{u_0}} \rp \rp =0.
\ee
Then, $\tilde f_1 =  f_1 -\kappa\rho_0 M_{u_0}\lp (\omega\cdot u_0) (\omega\cdot u_1) \rp$ is a solution of
\be \label{eq:aux_homogeneous_gamma}
\bar \Gamma(\tilde f_1, u_0) =0.
\ee
By Lemma \ref{lem:solution_Gamma_bar}, all solutions $\tilde f_1\in H^1(\mathbb{S}^{d-1})$ of \eqref{eq:aux_homogeneous_gamma} are of the form
$\tilde f_1~ =~ c M_{u_0},\, c~\in~\R$,
 (since obviously zero is the unique solution in $\dot{H}^1_0(\mathbb{S}^{d-1})$).  From this we conclude that $f_1$ is of the form \eqref{eq:solution homogeneous equation}.\\
\medskip
 Next we prove that if $f_1$ is of the form \eqref{eq:solution homogeneous equation}, then it is a solution to \eqref{eq:homogeneous equation}. Again by Lem. \ref{lem:linearisation}  and using the same transformation as in \eqref{eq:aux_recast_homogeneous} we have 
  \be \label{eq:alternative_homogeneous2}
D_{f_0}\Gamma(f_1)=\nabla_\omega\cdot \lp M_{u_0}\nabla_\omega \lp\frac{ (\omega\cdot u_0) (\omega\cdot (\hat u-\kappa \rho_0 u_1))  M_{u_0}}{M_{u_0}} \rp \rp,
\ee
(notice that the $\hat \rho M_{u_0}$ term in $f_1$ does not have any contribution to the right-hand side of \eqref{eq:alternative_homogeneous2} as it belongs to the kernel of $\bar \Gamma$). With \eqref{eq:hat_u}, the right-hand side of \eqref{eq:alternative_homogeneous2} vanishes and we conclude that $f_1$ is a solution of \eqref{eq:homogeneous equation}.
\end{proof}

\begin{lemma}[Uniqueness]
\label{lem:uniqsolDf0gamf=g}
There exists a unique solution $f$ to \eqref{eq:generic} such that the two properties \eqref{eq:condforunique} hold. Moreover, if $f$ is this solution, the set ${\mathcal S}_{u_0}$ of all solutions to \eqref{eq:generic} in $H^1 (\mathbb{S}^{d-1})$  is given by~\eqref{eq:express_Su0}. 
\end{lemma}

\begin{proof}
Let $f$ be the solution to \eqref{eq:generic} found in Lemma \ref{lem:solvcond_sufficient}. It satisfies \eqref{eq:condforunique}. Now, from Lemma \ref{lem:solutions_homogeneous_equation}, the set ${\mathcal S}_{u_0}$ of solutions of \eqref{eq:generic} in $H^1 (\mathbb{S}^{d-1})$  is given by \eqref{eq:express_Su0}. We show that none of the other solutions than $f$ satisfies \eqref{eq:condforunique}. This amounts to showing that for $f_1$ given by \eqref{eq:solution homogeneous equation}, to satisfy \eqref{eq:condforunique}, we need $\hat\rho = 0$ and $\hat u = 0$. Indeed, with \eqref{eq:moments1}, one can check that
$$\int_{\mathbb{S}^{d-1}}\frac{ f_1}{M_{u_0}} \, d\omega= \int_{\mathbb{S}^{d-1}}[\hat\rho+(\omega\cdot u_0)(\omega\cdot \hat u)]\, d\omega = \hat \rho.$$
This with \eqref{eq:condforunique_rho} implies $\hat \rho = 0$.
On the other hand, with \eqref{eq:Pu0Qhatf1u0}, \eqref{eq:intfomu0omper} implies $\hat u=0$, which proves the claim. \end{proof}

\begin{proof}[Proof of Th. \ref{th:solvability condition} (\nameref{th:solvability condition})]
Collect Lemmas \ref{lem:Df_against_GCI} to \ref{lem:uniqsolDf0gamf=g}. The converse property in (ii) is straightforward and is left to the reader. \end{proof}

\smallskip
Thanks to Eqs. \eqref{eq:hilbert_expansion} and \eqref{eq:f_0} we know that, as $\varepsilon \to 0$, $f_\varepsilon \to f_0=\rho_0 M_{u_0}$, with $\rho_0= \rho_0(t,x)~\geq~0$, $u_0=u_0(t,x)\in\mathbb{S}^{d-1}$. We are therefore left with determining the equations for $\rho_0$ and $u_0$. To obtain these equations, we apply Theorem \ref{th:solvability condition} to Eqs. \eqref{eq:order1} and \eqref{eq:order2}. We will see that the solvability conditions for Eq. \eqref{eq:order1} are satisfied and we will determine its solution $f_1$. Then, the solvability conditions for Eq. \eqref{eq:order2} will give us the equations for $\rho_0$ and $u_0$. This is performed in the forthcoming sections.

\subsection{Resolution of Eq. \eqref{eq:order1}}

For Eq. \eqref{eq:order1} to have a solution, by Th. \ref{th:solvability condition}, the solvability conditions \eqref{eq:solvability condition} with $g= (\omega \cdot \nabla_x) (\rho_0 M_{u_0})$ must hold. These conditions are made explicit in the following lemma:

\begin{lemma}[Eq. \eqref{eq:order1} satisfies the solvability condition] 
(i) Let $f_0 = \rho_0 M_{u_0}$. The function $g= (\omega \cdot \nabla_x) f_0$ satisfies the solvability conditions \eqref{eq:solvability condition}. \\
(ii) There exists a unique solution $f_{10}\in H^1(\mathbb{S}^{d-1})$ to Eq. \eqref{eq:order1} satisfying the two conditions \eqref{eq:condforunique}. Equivalently, $f_{10}$ is the unique solution in $\dot H^1_0(\mathbb{S}^{d-1})$ to 
\begin{equation}
\bar\Gamma(f_{10}, u_0)=(\omega\cdot \nabla_x)(\rho_0 M_{u_0}).
\label{eq:barGamf1}
\end{equation}
The general solution to Eq. \eqref{eq:order1} in $H^1(\mathbb{S}^{d-1})$ is given by $ f_1= f_{10} + \hat f_1$ with 
\begin{equation} \label{eq:expressbarf1}
\hat f_1 =  M_{u_0} \left(\hat \rho_1 + (\omega \cdot u_0) (\omega \cdot \hat u_1) \right) ,
\end{equation}
where $\hat \rho_1 \in {\mathbb R}$ and $\hat u_1 \in \{u_0\}^\bot$ are arbitrary. 
\label{lem:existence_of_f1}
\end{lemma}

\begin{proof}
(i) the solvability conditions \eqref{eq:solvability condition} take the following form (in the sequel we skip the sub-indices '0'):
\begin{equation} \label{eq:solvability_order1}
\int_{{\mathbb S}^{d-1}} (\omega\cdot \nabla_x) (\rho M_{u})(\omega)\, d\omega =0 \quad \mbox{ and } \int_{{\mathbb S}^{d-1}} (\omega\cdot \nabla_x) (\rho M_u)(\omega)\, \vec{\psi}_u (\omega)\, d\omega=0.
\end{equation}
Now, we note that
$$\frac{\partial \log M_u}{\partial u}= \kappa (\omega \cdot u) \omega_\perp,$$
where $\omega_\perp$ is defined in \eqref{eq:space_decomposition}. Then, for any linear first-order differential operator with respect to $(t,x)$, we have
\begin{equation}
D(\rho M_u) = \rho M_u \left[ D(\log \rho) + \kappa (\omega \cdot u) \omega_\perp \cdot Du \right].
\label{eq:formul_deriv}
\end{equation}
In particular, for $D=(\omega\cdot \nabla_x)$ and using the decomposition \eqref{eq:space_decomposition} again, we have:
\begin{eqnarray}
 	(\omega \cdot \nabla_x) (\rho M_u) &=& \rho M_u \Big[(\omega\cdot u) (u\cdot \nabla_x) \log \rho + (\omega_\perp \cdot \nabla_x) \log \rho \nonumber\\
 	&& + \kappa (\omega \cdot u)^2 \omega_\perp \cdot [(u\cdot \nabla_x) u]+\kappa (\omega \cdot u) \omega_\perp \cdot [(\omega_\perp \cdot \nabla_x) u]
 	\Big] \nonumber\\
 	&=&\rho M_u \,  (S_{e,o} + S_{o,e}), \label{eq:transport_equilibrium_parity}
 \end{eqnarray}
 where 
 \begin{eqnarray}
 S_{e,o}  &:=& (\omega_\perp \cdot \nabla_x) \log \rho + \kappa (\omega\cdot u)^2 \omega_\perp \cdot[(u\cdot \nabla_x) u]\in \sigma_{e,o}, \label{eq:def_Seo}\\
 S_{o,e} &:=& (\omega \cdot u) (u\cdot \nabla_x) \log \rho + \kappa (\omega \cdot u) \omega_\perp \cdot [(\omega_\perp\cdot \nabla_x) u]\in \sigma_{o,e}, \label{eq:def_Soe}
 \end{eqnarray}
(see the definitions of the spaces $\sigma_{e,o}$ and $\sigma_{o,e}$ in Sec. \ref{sec:preliminaries_notation}). Indeed, $S_{e,o}$ is even in $(\omega\cdot u)$ and odd in $\omega_\perp$ and the opposite holds for $S_{o,e}$. Since $\rho M_u$ is even both in $(\omega\cdot u)$ and $\omega_\perp$, its product with $S_{e,o}$ and $S_{o,e}$ does not change the parity. From this we conclude that the first integral of the solvability condition \eqref{eq:solvability_order1} is indeed zero by antisymmetry. Now, since $\vec{\psi}_u\in \sigma_{o,o}$, we have $\vec{\psi}_u \, \rho M_u \, (S_{e,o} + S_{o,e}) \in \sigma_{o,e} +\sigma_{e,o}$ and again, the integral with respect to $\omega$ is zero, which shows the second compatibility condition \eqref{eq:solvability_order1} and ends the proof of (i). (ii) is a direct application of Theorem \ref{th:solvability condition}. \end{proof}

Next we compute the explicit form of the function $f_{10}$, as it will be useful in the sequel. We remind the following notations: $\nabla_x u$ denotes the gradient tensor of the vector field $u$, i.e. $(\nabla_x u)_{ij} = \partial_{x_i} u_j$, $\forall i, \, j \in \{1, \ldots, d \}$ and for two order-two tensors ${\mathcal A} = ({\mathcal A}_{ij})_{(i,j) \in \{1, \ldots, d \}^2}$ and ${\mathcal B} = ({\mathcal B}_{ij})_{(i,j)}$, we denote by ${\mathcal A}:{\mathcal B}$ their contracted product, i.e. ${\mathcal A}:{\mathcal B} = \sum_{i,j=1}^d {\mathcal A}_{ij} {\mathcal B}_{ij}$. 
 
\begin{lemma}
\label{lem:shape_of_f1}[Determination of $f_{10}$]
The unique solution $f_{10}$ to Eq. \eqref{eq:order1} in $H^1(\mathbb{S}^{d-1})$ satisfying the two conditions \eqref{eq:condforunique} is given by 
\begin{equation} \label{eq:shape_of_f1}
f_{10}= \rho M_u (T_{e,o}+T_{o,e}),
\end{equation}
where $T_{e,o} \in \sigma_{e,o}$ and $T_{o,e} \in \sigma_{o,e}$ (see the definitions of the spaces $\sigma_{e,o}$ and $\sigma_{o,e}$ in Sec. \ref{sec:preliminaries_notation}) are defined by
\begin{eqnarray}
 T_{e,o}&:=& \alpha(\omega) \cdot (\nabla_x \log \rho) + \kappa\, \beta(\omega) \, \cdot [(u\cdot \nabla_x) u], \label{eq:Soo}\\
 T_{o,e}&:=& \gamma(\omega) \,  (u\cdot \nabla_x) \log \rho + \kappa \,  \zeta(\omega) : (\nabla_x u), \label{eq:See}
\end{eqnarray}
with $\alpha$, $\beta$, $\gamma$ and $\zeta$ defined by: 
\begin{itemize}
\item $\alpha$: ${\mathbb S}^{d-1} \to {\mathbb R}^d$, $\omega \mapsto \alpha(\omega)$ is the unique (componentwise) solution in $H^1_0({\mathbb S}^{d-1})^d$ of
\begin{equation}
\bar \Gamma^*(\alpha,u) = \omega_\perp, 
\label{eq:barGam*al}
\end{equation}
where $\bar \Gamma^*$ is defined in \eqref{eq:vecGCI}. 
\item $\beta$: ${\mathbb S}^{d-1} \to {\mathbb R}^d$, $\omega \mapsto \beta(\omega)$ is the unique (componentwise) solution in $H^1_0({\mathbb S}^{d-1})^d$ of
\begin{equation}
\bar \Gamma^*(\beta,u) = (\omega \cdot u)^2 \, \omega_\perp. 
\label{eq:barGam*bet}
\end{equation}
\item $\gamma$: ${\mathbb S}^{d-1} \to {\mathbb R}$, $\omega \mapsto \gamma(\omega)$ is the unique solution in $H^1_0({\mathbb S}^{d-1})$ of
\begin{equation}
\bar \Gamma^*(\gamma,u) = (\omega \cdot u). 
\label{eq:barGam*gam}
\end{equation}
\item $\zeta$: ${\mathbb S}^{d-1} \to {\mathcal S}^d$, $\omega \mapsto \zeta(\omega)$, where ${\mathcal S}^d$ is the space of $d \times d$ symmetric matrices with coefficients in ${\mathbb R}$, is the unique solution in $H^1_0({\mathbb S}^{d-1})^{d(d+1)/2}$ of
\begin{equation}
\bar \Gamma^*(\zeta,u) = (\omega \cdot u) \, \omega_\perp \otimes \omega_\perp.  
\label{eq:barGam*zet}
\end{equation}
\end{itemize}
Furthermore, $\alpha$ to $\zeta$ take the following forms: 
\begin{eqnarray}
\alpha (\omega) &=& a(\omega \cdot u) \, \omega_\perp, \label{eq:espressal}\\
\beta (\omega) &=& b(\omega \cdot u) \, \omega_\perp, \label{eq:espressbet}\\
\gamma (\omega) &=& c(\omega \cdot u) , \label{eq:espressgam}\\
\zeta (\omega) &=& e(\omega \cdot u) \, \omega_\perp \otimes \omega_\perp + k(\omega \cdot u) \,  \, P_{u^\perp}, \label{eq:espresszet}
\end{eqnarray}
where $a$, $b$, $c$, $e$ and $k$ have been defined in Section \ref{subsec:main_result}. With the parities of $a$, $b$, $c$, $e$, $k$ stated in Section \ref{subsec:main_result}, we have $ \alpha, \, \beta \in \sigma_{e,o}$, $\gamma, \, \zeta \in \sigma_{o,e}$. The expressions of $T_{e,o}$ and $T_{o,e}$ are consequently given by 
\begin{eqnarray}
 T_{e,o}&=& a(\omega \cdot u) \, \omega_\perp \cdot (\nabla_x \log \rho) + \kappa\, b(\omega \cdot u) \, \omega_\perp \cdot [(u\cdot \nabla_x) u], \label{eq:Teo}\\
 T_{o,e}&=& c(\omega \cdot u) \,  (u\cdot \nabla_x) \log \rho + \kappa \, e(\omega \cdot u) \, (\omega_\perp \otimes \omega_\perp) : (\nabla_x u) + \kappa \, k(\omega \cdot u) \, (\nabla_x \cdot u).  \label{eq:Toe}
\end{eqnarray}
The general solution $f_1$ to Eq. \eqref{eq:order1} in $H^1(\mathbb{S}^{d-1})$ is given by 
\begin{equation} \label{eq:shape_of_f1_generic}
f_1= f_{10} + \hat f_1 = \rho M_u (T_{e,o}+T_{o,e}) + \hat f_1, 
\end{equation}
with $\hat f_1$ given by \eqref{eq:expressbarf1}. 
\end{lemma}

\medskip
\begin{proof}
We perform the explicit resolution of \eqref{eq:barGamf1}, which we formally write as $\bar \Gamma (f_1,u) = g_1$, with $g_1 = (\omega \cdot \nabla_x) (\rho M_u)$ (we write $f_1$ for $f_{10}$ to keep the notations simple). First, we note that 
this equation is equivalent, through the change of functions $f_1 = M_u \, \tilde f_1$, $g_1 = M_u \, \tilde  g_1$ to the equation 
\begin{equation}
\bar \Gamma^* (\tilde f_1,u) = \tilde g_1. 
\label{eq:barGam*tilf1}
\end{equation}
From the proof of Lemma \ref{lem:existence_of_f1}, we get that 
$$ \tilde g_1 = \rho \, (S_{e,o} + S_{o,e}), $$ 
with $S_{e,o}$ and $S_{o,e}$ given by \eqref{eq:def_Seo}, \eqref{eq:def_Soe}. Furthermore, we know from the proof of Proposition \ref{lem:GCI} that \eqref{eq:barGam*tilf1} is uniquely solvable in $H^1_0(\mathbb{S}^{d-1})$ provided that 
\begin{equation}
\int_{{\mathbb S}^{d-1}} \tilde g_1 \, M_u \,  d \omega = 0.
\label{eq:inttilg1=0}
\end{equation} 
This condition is satisfied thanks to Lemma \ref{lem:existence_of_f1}. Furthermore, each of the terms involved in the equations \eqref{eq:def_Seo} and \eqref{eq:def_Soe} satisfies this  condition separately. In other words, the following functions of $\omega$: $\omega_\perp$, $(\omega \cdot u)^2 \omega_\perp$, $(\omega \cdot u)$ and $(\omega \cdot u) \, \omega_\perp \otimes \omega_\perp$, satisfy the solvability condition \eqref{eq:inttilg1=0} (in the case of vector or tensor quantities, this is meant componentwise). This provides the existence and uniqueness in $H^1_0(\mathbb{S}^{d-1})$ of the solutions $\alpha$, $\beta$, $\gamma$, $\zeta$ to Eqs. \eqref{eq:barGam*al} to \eqref{eq:barGam*zet}. By linearity, we have 
$$
\tilde f_1 = \rho (T_{e,o}+T_{o,e}),
$$
with  $T_{e,o}$ and $T_{o,e}$ given by \eqref{eq:Soo}, \eqref{eq:See}. 

Now, each of the Eqs. \eqref{eq:barGam*al} to \eqref{eq:barGam*zet} is of the same form as Eq. \eqref{eq:vecGCI} for the GCI and we can use the same methodology as in the proof of Proposition \ref{lem:GCI} to derive the expressions  \eqref{eq:espressal}-\eqref{eq:espresszet} of $\alpha$ to $\zeta$. More precisely, the computations leading to the expressions \eqref{eq:espressal} and \eqref{eq:espressbet} of $\alpha$ and $\beta$ are identical to those leading to the expression \eqref{eq:AnsatzpsiB} of $\psi_B$ in the proof of Proposition \ref{lem:GCI} and are omitted. 

The same methodology leads to the expression \eqref{eq:espressgam} of $\gamma$ with $c$ a solution to \eqref{eq:def_c}. It is readily shown that this equation has a unique solution in $\dot {\mathcal H}_{0,\frac{d-1}{2}}$ (defined in \eqref{eq:dotH0}). Indeed, \eqref{eq:def_c} has the following variational formulation: find $c \in \dot {\mathcal H}_{0,\frac{d-1}{2}}$ such that 
\begin{equation}
\int_{-1}^1 (1-r^2)^{\frac{d-1}{2}} \, e^{\frac{\kappa r^2}{2}} \,  c'(r) \, \ell'(r) \, dr = - \int_{-1}^1 r \, (1-r^2)^{\frac{d-2}{2}}  \, e^{\frac{\kappa r^2}{2}} \, \ell(r) \, dr, \quad \forall \ell \in \dot {\mathcal H}_{0,\frac{d-1}{2}}.  
\label{eq:varforc}
\end{equation}
The fact that it is equivalent to take test functions in ${\mathcal H}_{0,\frac{d-1}{2}}$  or in $\dot {\mathcal H}_{0,\frac{d-1}{2}}$ follows from the fact that $\int_{-1}^1 r \, (1-r^2)^{\frac{d-1}{2}} \, e^{\frac{\kappa r^2}{2}} \, dr =0$ by a similar reasoning as that done in the proof of Proposition \ref{lem:GCI} to solve \eqref{eq:GCIchar1}. It is easy to prove a Poincare-Wirtinger inequality for $\dot {\mathcal H}_{0,\frac{d-1}{2}}$, namely, there exists a constant $C>0$ such that 
$$ \int_{-1}^1 \ell^2 \, dr \leq C \int_{-1}^1 (\ell')^2 \, (1-r^2)^{\frac{d-1}{2}} \, dr, \qquad \forall \ell \in \dot {\mathcal H}_{0,\frac{d-1}{2}}, $$
(the proof is left to the reader). So, the bilinear form at the left hand side of \eqref{eq:varforc} is continuous and coercive on $\dot {\mathcal H}_{0,\frac{d-1}{2}}$, while the linear form at the right-hand side of \eqref{eq:varforc} is continuous on the same space. Therefore, Lax-Milgram's theorem applies and provides the unique solvability of \eqref{eq:def_c} in $\dot {\mathcal H}_{0,\frac{d-1}{2}}$. Thanks to uniqueness, we also get that $c$ is odd. 
It is also straightforward to show that $\tilde\gamma$ constructed from the so-defined $c$ through \eqref{eq:espressgam} belongs to $H^1_0({\mathbb S}^{d-1})$. By the uniqueness to the solution of \eqref{eq:barGam*tilf1} in $H^1_0({\mathbb S}^{d-1})$, we deduce that $\tilde \gamma = \gamma$ and consequently that \eqref{eq:espressgam} holds true.

The methodology requires a small adaptation in the case of $\zeta$ because of the need for two functions of $(\omega \cdot u)$ namely $e$ and $k$. Let $\tilde \zeta$ be constructed from $e$ and $k$ through the Ansatz \eqref{eq:espresszet} and let $B$ be a vector in ${\mathbb R}^d$. Then, thanks to the polarization identity, we only need to show that 
\begin{equation}
\bar \Gamma^* \big(B \cdot (\tilde \zeta B), u \big) = (\omega \cdot u) \, (\omega_\perp \cdot B)^2, \quad \forall B \in {\mathbb R}^d, 
\label{eq:barGam*CzetB}
\end{equation}
and that $B \cdot (\tilde \zeta B) \in H^1_0({\mathbb S}^{d-1})$. We only need to show \eqref{eq:barGam*CzetB} for an orthonormal basis of ${\mathbb R}^d$. We first start with $B=u$ and notice that, in this case, both $u \cdot (\tilde \zeta u)$ and the right-hand side of \eqref{eq:barGam*CzetB} are zero, which shows that \eqref{eq:barGam*CzetB} is satisfied in this case. Now, we take any unit vector $B$ such that $(u \cdot B)=0$. Then $ B \cdot (\tilde \zeta B) = e(\omega \cdot u) \, (\omega_\perp \cdot B)^2 + k(\omega \cdot u)$ and computations similar to those of the proof of Proposition \ref{lem:GCI} give:
\begin{eqnarray*}
\bar \Gamma^*\big(B \cdot (\tilde \zeta B),u\big) &=& (\omega_\perp \cdot B)^2 \Big\{ e''(\omega \cdot u) \,  \big(1 - (\omega \cdot u)^2\big) \\
&&+ e'(\omega \cdot u) \, (\omega \cdot u) \, \big[ \kappa  \,  \big(1 - (\omega \cdot u)^2\big) - (d+3) \big] \\  
&&+ e(\omega \cdot u) \big[ - 2 \, \kappa \, (\omega \cdot u)^2 - 2d \big] \Big\} \\
&+& 2 e(\omega \cdot u)  + \bar \Gamma^*\big(k(\omega \cdot u) ,u\big) \\
&=& (\omega \cdot u) \, (\omega_\perp \cdot B)^2. 
\end{eqnarray*}
Since this equation must be satisfied for all values of $(\omega_\perp \cdot B)^2$, this requires
the following two equations to be satisfied: 
\begin{eqnarray*}
&& (1-r^2) e'' + \big( \kappa \, (1-r^2) - (d+3) \big) \, r e' - \big( 2 \kappa \, r^2 + 2d \big) \, e = r,  \\
&& \bar \Gamma^*\big(k(\omega \cdot u) ,u\big) = - 2 e(\omega \cdot u). 
\end{eqnarray*}
The first equation, after rearrangement, gives \eqref{eq:def_e}, while the second one (which is similar but for the right-hand side to the equation for $c$) shows that $k$ satisfies Eq. \eqref{eq:def_c} with right-hand side $- 2 e(r) \, (1-r^2)^{(d-2)/2}  \exp( \frac{\kappa r^2}{2})$. By the same arguments as for \eqref{eq:ode_h}, Eq. \eqref{eq:def_e} is uniquely solvable in  ${\mathcal H}_{\frac{d+1}{2},\frac{d+3}{2}}$ and $e$ is odd and nonpositive for $r \geq 0$. Similarly, by the same arguments as for $c$ above, the equation for $k$ is uniquely solvable in $\dot {\mathcal H}_{0,\frac{d-1}{2}}$. We also have that $k(\omega \cdot u) \in H^1_0({\mathbb S}^{d-1})$. So, it remains to show that $ e(\omega \cdot u) \, (\omega_\perp \cdot B)^2 \in  H^1_0({\mathbb S}^{d-1})$. This is done using the same arguments as for $h$ in the proof of Proposition \ref{lem:GCI} and is omitted. 

Inserting \eqref{eq:espressal} and \eqref{eq:espressbet} into \eqref{eq:Soo} immediately gives \eqref{eq:Teo}. To show that the insertion of \eqref{eq:espressgam} and \eqref{eq:espresszet} into \eqref{eq:See} gives \eqref{eq:Toe}, we need to show that 
\begin{equation}
P_{u^\perp}:(\nabla_x u) = \nabla_x \cdot u.
\label{eq:divu}
\end{equation}
Indeed, in a cartesian coordinate system, the matrix $\nabla_x u$ has entries $(\nabla_x u)_{ij} = \partial_{x_i} u_j$. Then, the vector $(\nabla_x u) u$ has components in this basis:  
$$ ((\nabla_x u) u)_i =  \sum_{j=1}^d (\partial_{x_i} u_j) u_j = \frac{1}{2} \sum_{j=1}^d (\partial_{x_i} u_j^2) = \frac{1}{2} \partial_{x_i} (|u|^2) = 0, $$
since $|u|=1$, so that 
\begin{equation} 
(\nabla_x u) u = 0. 
\label{eq:nauu=0} 
\end{equation}
Therefore, we have 
\begin{equation}
(u \otimes u) : \nabla_x u = u \cdot ((\nabla_x u) u) = 0. 
\label{eq:otonu}
\end{equation}
Now, $\mbox{Id} = P_{u^\perp} + u \otimes u$. So, we have 
$$ \nabla_x \cdot u = \mbox{Id} : \nabla_x u = (P_{u^\perp} + u \otimes u) : \nabla_x u = P_{u^\perp} : \nabla_x u, $$
which yields \eqref{eq:divu}.

Finally, going back to $f_1$ from $\tilde f_1$, we note that $\rho M_u (T_{e,o}+T_{o,e})$ satisfies the two conditions \eqref{eq:condforunique} (the second one because $g_1$ satisfies the second solvability condition \eqref{eq:solvability condition}). So, returning to the notation $f_{10}$, we have proven that  the unique solution $f_{10}$ to Eq. \eqref{eq:order1} satisfying the two conditions \eqref{eq:condforunique} is given by \eqref{eq:shape_of_f1}. Consequently, the generic solution to \eqref{eq:order1} is given by \eqref{eq:shape_of_f1_generic}. This ends the proof.  \end{proof}

\subsection{Solvability conditions for Eq. \eqref{eq:order2}}

The solvability conditions \eqref{eq:solvability condition} applied to \eqref{eq:order2} will yield the evolution equations for $\rho=\rho(t,x)$ and $u=u(t,x)$. These conditions are written: 
\begin{subequations} \label{eq:solvability_equations}
\begin{numcases}{}
\int_{{\mathbb S}^{d-1}} \Big[\partial_t f_0 + (\omega\cdot \nabla_x) (f_{10} + \hat f_1) - \frac{1}{2}D^2_{f_0}\Gamma(f_1, f_1) \Big]\, d\omega =0, \label{eq:solvability_equations1}\\
\int_{{\mathbb S}^{d-1}} \Big[\partial_t f_0 + (\omega\cdot \nabla_x) (f_{10} + \hat f_1) - \frac{1}{2}D^2_{f_0}\Gamma(f_1, f_1) \Big]\, \vec \psi_{u}\, d\omega=0, \label{eq:solvability_equations2}
\end{numcases}
\end{subequations}
where $f_0$, $f_1$, $f_{10}$ and $\hat f_1$ are given by \eqref{eq:f_0} \eqref{eq:shape_of_f1_generic}, \eqref{eq:shape_of_f1} and \eqref{eq:expressbarf1} respectively. In the first two forthcoming subsections, we show that the terms involving $D^2_{f_0}\Gamma$ and $\hat f_1$ do not contribute to the result. Then, we will consider sequentially \eqref{eq:solvability_equations1} and \eqref{eq:solvability_equations2} and show that they yield the mass conservation equation \eqref{eq:macro_equations_rho} and the evolution equation for $u$ 
\eqref{eq:macro_equations_u} respectively.

\subsubsection{The terms involving $D^2_{f_0}\Gamma$}

\begin{lemma}
\label{lem:second_order_terms}
Let $f_0 = \rho M_u$. For any function $f_1 = f_1(\omega)$, we have
\begin{equation}
\int_{{\mathbb S}^{d-1}} D^2_{f_0}\Gamma(f_1, f_1) \ d\omega =0, \qquad \int_{{\mathbb S}^{d-1}}  D^2_{f_0}\Gamma(f_1, f_1) \, \vec\psi_u \ d\omega =0,
\label{eq:intD2Gamma}
\end{equation}
\end{lemma}

\begin{proof}
The proof is an extension to the second order in $\varepsilon$ of the proof of Lemma \ref{lem:Df_against_GCI}. The first identity \eqref{eq:intD2Gamma} is straightforward and its proof is omitted. The second identity \eqref{eq:intD2Gamma} requires a bit more care due to the need to expand $\vec \psi_{u^\varepsilon}$ and is developed below. Using the same notations as in the proof of Lemma \ref{lem:Df_against_GCI}, the second order term in the expansion of $\int_{{\mathbb S}^{d-1}}  \Gamma(f^\varepsilon ) \, \vec\psi_{u^\varepsilon}\ d\omega =0$ leads to 
\begin{eqnarray*} 0 &=& \int_{{\mathbb S}^{d-1}} \big[ \Gamma (f_0) \, D^2_{f_0} \vec \psi_f (f_1, f_1) + 2 D_{f_0} \Gamma (f_1) \, D_{f_0} \vec \psi_f (f_1) + D^2_{f_0}\Gamma(f_1, f_1)  \,  \vec \psi_f(f_0) \big] \, d \omega \\
&:=& J_1 + J_2 + J_3, 
\end{eqnarray*}
where we have introduced the map $f \mapsto  \vec \psi_f$ with $\vec \psi_f = \vec \psi_{u_f}$ and its derivatives up to the second order. To prove the second identity in \eqref{eq:intD2Gamma} it is enough to show that $J_1=J_2=0$. But $\Gamma(f_0) = 0$, so $J_1 = 0$. Now, we compute $J_2$ with $f_1$  given by \eqref{eq:shape_of_f1_generic}. Since  $D_{f_0} \Gamma (\hat f_1) = 0$ and  
$$ D_{f_0} \vec \psi_f (f_1) = D_{u_0} \vec \psi_u \big( D_{f_0} u_f (f_1) \big), $$
(where we now consider the map $u \mapsto \vec \psi_{u}$), we have 
$$ J_2 = 2 \int_{{\mathbb S}^{d-1}} D_{f_0} \Gamma (f_{10}) \, D_{u_0} \vec \psi_u \big( D_{f_0} u_f (f_{10} + \hat f_1) \big) \, d \omega.$$ 
But $u_{10}:= D_{f_0} u_f (f_{10})$ has already been computed at Lemma \ref{lem:equation_u_1} and is given by \eqref{eq:system_u1}. Since $f_{10}$ satisfies \eqref{eq:intfomu0omper}, we get $u_{10} = 0$. Similarly, thanks to \eqref{eq:Pu0Qhatf1u0}, we have $P_{u_0^\bot} (Q_{\hat f_1} u_0) = c \, \hat u_1 $ where $c$ denote generic constants in ${\mathbb R}$. Then, by \eqref{eq:system_u1}, we get $D_{f_0} u_f (\hat f_1) = c \hat u_1$ as well. Since $\hat u_1 \in \{u_0\}^\bot$ and with \eqref{eq:GCI}, we get:
\be \label{eq:expansion_GCI}
 D_{u_0} \vec \psi_{u} (\hat u_1) = h'(\omega \cdot u_0) \, (\omega_\perp \cdot \hat u_1) \, \omega_\perp - h(\omega \cdot u_0) \, (\omega_\perp \cdot \hat u_1) \, u_0 - h(\omega \cdot u_0) \,  (\omega \cdot u_0) \hat u_1. 
 \ee
We note that the first and third terms belong to $\sigma_{e,e}$ while the second one belongs to $\sigma_{o,o}$. Since $D_{f_0} \Gamma (f_{10}) \in \sigma_{o,e} + \sigma_{e,o}$ by the proof of Lemma \ref{lem:existence_of_f1}, we get that $J_2 = 0$ by antisymmetry. \end{proof}

\subsubsection{The terms involving $\hat f_1$}

\begin{lemma}
\label{lem:terms_hatf1}
Let $\hat  f_1$ be given by \eqref{eq:expressbarf1}. Then, we have 
\begin{equation}
\int_{{\mathbb S}^{d-1}} (\omega \cdot \nabla_x) \hat f_1 \ d\omega =0, \qquad \int_{{\mathbb S}^{d-1}}  (\omega \cdot \nabla_x) \hat f_1 \, \vec\psi_u \ d\omega =0,
\label{eq:intomnabhatf1}
\end{equation}
\end{lemma}

\begin{proof}
The fact that $(\omega \cdot \nabla_x) (\hat \rho_1 \, M_u)$ satisfies \eqref{eq:intomnabhatf1} is a reproduction of the proof of Lemma \ref{lem:existence_of_f1} with $\rho_0$ replaced by $\hat \rho_1$.  Consider $f_0= M_{u_0}$ for some $u_0 \in \mathbb{S}^{d-1}$ (notice that here we assume $\rho_0=1$ contrary to previous lemmas, this is because $\rho_0$ does not play a role in the proof). Consider $\hat f_1 = M_{u_0}(\omega\cdot u_0)(\omega\cdot \hat u)$ with $\hat u \in \{u_0\}^\perp$. Define $u_1= \hat u/\kappa$. Define 
$$u^\eps = \frac{u_0 +\eps u_1}{|u_0 + \eps u_1|}.$$
Note that $u^\eps = u_0 + \eps u_1+ \mathcal{O}(\eps^2)$.

We consider $M_{u^\eps}$ and $\vec\psi_{u^\eps}$. We have that\be \label{eq:non_expansion}
\int_{\mathbb{S}^{d-1}}[(\omega\cdot \nabla_x)M_{u^\eps}]\,
\lp
\begin{array}{c}
1\\
\vec \psi_{u^\eps}
\end{array}
\rp\, d\omega =0.
\ee
This was proven in Lemma \ref{lem:existence_of_f1}, see expression \eqref{eq:solvability_order1} (notice that $\rho=1$ does not change the result of the lemma). Now we expand $M_{u^\eps}$ and $\vec \psi_{u^\eps}$ in terms of $\eps$. Firstly, we have that
\be \label{eq:aux_M_expansion}
 M_{u^\varepsilon} = M_{u_0 + \varepsilon u_1+\mathcal{O}(\eps^2)} = M_{u_0} \, \Big( 1 + \varepsilon \left. \frac{\partial}{\partial u} (\log M_{u}) \right|_{u_0} (u_1) + {\mathcal O}(\varepsilon^2) \Big),
\ee
with 
$$ \left. \frac{\partial}{\partial u} (\log M_{u}) \right|_{u_0} (u_1) = \kappa \, (\omega \cdot u_0) \, (\omega \cdot u_1). $$
Now, using \eqref{eq:expansion_GCI} it also holds that for any $\beta \in \mathbb{S}^{d-1}$:
\be \label{eq:aux_inner_product_GCI_expansion}
 \beta\cdot \vec \psi_{u^\eps} =
\beta\cdot \lp \vec \psi_{u_0} + \eps A(u_0, u_1) \rp + \mathcal{O}(\eps^2),
\ee
where
$$A(u_0, u_1) := [(\omega\cdot u_0) u_1 + (\omega_\bot \cdot u_1)u_0] h(\omega\cdot u_0) + h'(\omega\cdot u_0)\, (\omega_\bot\cdot u_1)\omega_\bot. $$
Inserting \eqref{eq:aux_M_expansion} in the first line of  \eqref{eq:non_expansion} and considering the terms of order $\eps$ only (since the leading order term is zero) gives 
$$\int_{\mathbb{S}^{d-1}}(\omega\cdot \nabla_x)[\kappa (\omega\cdot u_0) (\omega\cdot u_1) M_{u_0})]\, d\omega=0,$$
and this corresponds, precisely, to the first identity in \eqref{eq:intomnabhatf1}.
 For the second line of  in \eqref{eq:non_expansion} using \eqref{eq:aux_M_expansion} and \eqref{eq:aux_inner_product_GCI_expansion}  and again keeping the order $\eps$ terms only gives
$$
\int_{\mathbb{S}^{d-1}}(\omega\cdot\nabla_x) [\kappa (\omega\cdot u_0)(\omega\cdot u_1) M_{u_0}]\,\vec \psi_{u_0} \, d\omega 
+ \int_{\mathbb{S}^{d-1}} [(\omega\cdot \nabla_x) M_{u_0}]\, A(u_0, u_1)  \, d\omega =0. 
$$
One can check that $(\omega\cdot \nabla_x)M_{u_0}\in \sigma_{e,o}+\sigma_{o,e}$ and that
$A(u_0, u_1)  \in \sigma_{e,e}+\sigma_{o,o}$. So, by parity, the second integral in the previous expression vanishes, and we get
$$
\int_{\mathbb{S}^{d-1}}(\omega\cdot\nabla_x) [\kappa (\omega\cdot u_0)(\omega\cdot u_1) M_{u_0}]\, \vec \psi_{u_0} \, d\omega =0,$$
which corresponds to the second identity in \eqref{eq:intomnabhatf1}.

 \end{proof}

\subsubsection{Equation for the density: explicit form of Eq. \eqref{eq:solvability_equations1}}
\label{subsubsec:density}

In this section, we prove Eq. \eqref{eq:macro_equations_rho}. We compute the various terms in \eqref{eq:solvability_equations1}. Since $\hat f_1$ does not have any contribution and there is no possible confusion, we denote $f_{10}$ by $f_1$ to simplify the notations. For the term involving $(\omega \cdot \nabla_x) f_1$, we notice that 
\begin{equation}
\int_{{\mathbb S}^{d-1}} (\omega\cdot \nabla_x) f_1 \, d\omega = \nabla_x \cdot \int_{{\mathbb S}^{d-1}} \omega \, f_1 \, d\omega = \nabla_x \cdot I, \qquad I= \int_{{\mathbb S}^{d-1}} \omega \, f_1 \, d\omega.
\label{eq:2ndterm}
\end{equation}
We compute, using Lem. \ref{lem:shape_of_f1}: 
\begin{eqnarray}
I =   \rho \int_{{\mathbb S}^{d-1}} \omega M_u (T_{o,e}+ T_{e,o}) \, d\omega =: I_1 + I_2. 
\label{eq:I}
\end{eqnarray}
We first compute $I_1$ using \eqref{eq:space_decomposition}, \eqref{eq:moments2}, \eqref{eq:divu},  and that $T_{o,e}\in \sigma_{o,e}$ and $\rho (u\cdot \nabla_x \log \rho)\, u = (u\cdot \nabla_x \rho)\, u$:
\begin{eqnarray}
I_1 &=& \rho \int_{{\mathbb S}^{d-1}} (\omega\cdot u) u \  T_{o,e} M_u \ d\omega \nonumber\\
&=& \rho u \, \int_{{\mathbb S}^{d-1}} M_u \, \, (\omega\cdot u) \Big[ c \,  (u\cdot\nabla_x)\log \rho  + \kappa \, e \,  \, (\omega_\perp \otimes \omega_\perp) :\nabla_x u + \kappa \, k \,  \, (\nabla_x \cdot u) \Big] \ d\omega \nonumber\\
&=& C_1 \, \big( (u \cdot \nabla_x) \rho \big) \, u + C_4 \, (\nabla_x \cdot u) \, \rho u  , 
\label{eq:I1_bis}
\end{eqnarray} 
where
\begin{eqnarray}
C_1 &=& \int_{{\mathbb S}^{d-1}} M_u \, c \, \, (\omega \cdot u) \, d \omega,  \label{eq:C1} \\
C_4 &=& \kappa \int_{{\mathbb S}^{d-1}} M_u \, (\omega \cdot u) \, \Big( e \, \,  \frac{1- (\omega \cdot u)^2}{d-1} + k \,  \Big) \,  d \omega. \label{eq:C4} 
\end{eqnarray}

Next we compute the integral $I_2$, again using the decomposition \eqref{eq:space_decomposition}, the identity \eqref{eq:moments2} and that  $T_{e,o}\in \sigma_{e,o}$:
\begin{eqnarray}
I_2 &=& \rho \int_{{\mathbb S}^{d-1}} \omega_\perp T_{e,o} M_u \ d\omega \nonumber \\
&=&\rho \int_{{\mathbb S}^{d-1}} M_u \, \omega_\perp \big[a \, \, \omega_\perp \cdot \nabla_x\log \rho + \kappa \, b \, \,  \omega_\perp \cdot \big((u\cdot \nabla_x) u\big) \big] \ d\omega \nonumber \\
&=& C_2 \, P_{u^\perp}\nabla_x \rho + C_3 \, \rho \, (u\cdot\nabla_x) u,
\label{eq:I2}
\end{eqnarray}
with 
\begin{eqnarray}
C_2 &=& \frac{1}{d-1} \int_{{\mathbb S}^{d-1}} M_u \, a \, \, (1- (\omega \cdot u)^2)  \, d \omega,  \label{eq:C2} \\
C_3 &=& \frac{\kappa}{d-1} \int_{{\mathbb S}^{d-1}} M_u \, b \, \, (1- (\omega \cdot u)^2)  \, d \omega.\label{eq:C3}
\end{eqnarray}
In the last equality of \eqref{eq:I2}, we have used that $((u \cdot \nabla_x)u) \cdot u = 0$ (since $|u|=1$) and so $P_{u^\perp}\big((u\cdot\nabla_x)u\big)=(u\cdot\nabla_x)u$.

Finally for the term involving $\partial_t f_0$, we notice, thanks to \eqref{eq:formul_deriv}, that
\begin{equation}
\partial_t(\rho M_u) = \rho M_u \left[ \partial_t(\log \rho) + \kappa (\omega \cdot u) \omega_\perp \cdot \partial_tu \right].
\label{eq:formul_partial_t}
\end{equation}
Then, integrating this formula with respect to $\omega$ and using \eqref{eq:moments1}, we get 
\begin{equation}
\int_{{\mathbb S}^{d-1}}  \partial_t (\rho M_u) \ d\omega = \partial_t \rho.
\label{eq:formul_1stterm}
\end{equation}
Then, inserting \eqref{eq:I1_bis}, \eqref{eq:I2}, \eqref{eq:formul_1stterm} into \eqref{eq:solvability_equations1} yields \eqref{eq:macro_equations_rho}. Formulas \eqref{eq:C1_bis}, \eqref{eq:C2_bis}, \eqref{eq:C3_bis}, \eqref{eq:C4_bis} are easily deduced from \eqref{eq:C1}, \eqref{eq:C2}, \eqref{eq:C3}, \eqref{eq:C4} through the use of \eqref{eq:int_spheriq}.

\subsubsection{Equation for the mean direction: explicit form of Eq. \eqref{eq:solvability_equations2}}

In this section, we prove Eq. \eqref{eq:macro_equations_u}, i.e. we compute the various terms involved in \eqref{eq:solvability_equations2}. Again, since $\hat f_1$ does not have any contribution, we denote $f_{10}$ by $f_1$. First, we consider term involving $\partial_t f_0$. From 
\eqref{eq:formul_partial_t}, we get that 
$$ \int_{{\mathbb S}^{d-1}} \partial_t (\rho M_{u}) \, \vec \psi_u \ d\omega = \rho \int_{{\mathbb S}^{d-1}} M_u \, h \, \, \omega_\perp \left[ \partial_t(\log \rho) + \kappa (\omega \cdot u) \omega_\perp \cdot \partial_tu \right] \, d \omega.
$$ 
The integral in factor of $\partial_t(\log \rho)$ vanishes by antisymmetry. The other integral is computed using \eqref{eq:moments2}. With the fact the $\partial_t u$ is orthogonal to $u$ this leads to 
\begin{equation} \label{eq:time_derivative}
\int_{{\mathbb S}^{d-1}} \partial_t (\rho M_{u})\bar \psi_u\ d\omega = C_0\, \rho \ \partial_t u ,
\end{equation}
where
$$ C_0 = \frac{\kappa}{d-1} \int_{{\mathbb S}^{d-1}} M_u \, h \, \, (\omega \cdot u) \,  (1 - (\omega \cdot u)^2) \, d \omega. $$

Next, we compute the term involving $(\omega \cdot \nabla_x) f_1$. By \eqref{eq:shape_of_f1}, we have
\begin{eqnarray*}
(\omega \cdot \nabla_x) f_1 &=& A+B,  \\
A &=& \big( (\omega \cdot \nabla_x) (\rho \, M_u) \big) \, (T_{e,o} + T_{o,e}), \\
B &=& \rho \, M_u \, (\omega \cdot \nabla_x) (T_{e,o} + T_{o,e}). 
\end{eqnarray*}
We first compute the contribution of $A$. The quantity $(\omega \cdot \nabla_x) (\rho \, M_u) $ is given by \eqref{eq:transport_equilibrium_parity} together with \eqref{eq:def_Seo}, \eqref{eq:def_Soe}. By inspection of the parity of the functions with respect to $\omega \cdot u$ and $\omega_\perp$, we get 
\begin{equation} 
\int_{{\mathbb S}^{d-1}} A \, \vec \psi_u \, d\omega = \int_{{\mathbb S}^{d-1}} \rho \, M_u \, [ S_{e,o} \, T_{o,e} + S_{o,e} \, T_{e,o}]  \, \vec \psi_u \, d\omega. 
\label{eq:intApsi=}
\end{equation}
Straightforward algebra using the expression \eqref{eq:GCI_explicit} for $\vec \psi_u$ and Lemma \ref{lem:moments} leads to 
\begin{eqnarray}
\int_{{\mathbb S}^{d-1}} \rho \, M_u \, S_{e,o} \, T_{o,e}  \, \vec \psi_u \, d\omega &=& B_{32} \, (u \cdot \nabla_x \log \rho) \, P_{u^\bot} \nabla_x \rho + B_{42} \, \Sigma : \big( (\nabla_x u) \otimes \nabla_x \rho \big) \nonumber \\
&+& B_{52} \, (\nabla_x \cdot u) \, P_{u^\bot} \nabla_x \rho + A_{11} \, (u \cdot \nabla_x \rho) \, (u \cdot \nabla_x) u \nonumber \\
&+& A_{12} \, \rho \, \Sigma : \big( (\nabla_x u) \otimes (u \cdot \nabla_x) u \big) + A_{13} \, \rho \, (\nabla_x \cdot u) \, (u \cdot \nabla_x) u, \nonumber\\
\label{eq:intSeoToepsi}
\end{eqnarray}
with 
\begin{eqnarray}
B_{32} &=& \frac{1}{d-1} \int_{{\mathbb S}^{d-1}} M_u \, h \, c \, \, (1-(\omega \cdot u)^2) \, d \omega, \label{eq:B32}\\
B_{42} &=& \frac{\kappa}{(d-1)(d+1)} \int_{{\mathbb S}^{d-1}} M_u \, h \, e  \, \,  (1-(\omega \cdot u)^2)^2 \, d \omega. \label{eq:B42} \\
B_{52} &=& \frac{\kappa}{d-1} \int_{{\mathbb S}^{d-1}} M_u \, h \, k \, \, (1-(\omega \cdot u)^2) \, d \omega, \label{eq:B52}\\
A_{11} &=& \frac{\kappa}{d-1} \int_{{\mathbb S}^{d-1}} M_u \, h \, c \, \, (\omega \cdot u)^2 \, (1-(\omega \cdot u)^2) \, d \omega, \label{eq:A11}\\
A_{12} &=& \frac{\kappa^2}{(d-1)(d+1)} \int_{{\mathbb S}^{d-1}} M_u \, h \, e  \, \, (\omega \cdot u)^2 \,  (1-(\omega \cdot u)^2)^2 \, d \omega. \label{eq:A12} \\
A_{13} &=& \frac{\kappa^2}{d-1} \int_{{\mathbb S}^{d-1}} M_u \, h \, k \, \, (\omega \cdot u)^2 \, (1-(\omega \cdot u)^2) \, d \omega. \label{eq:A13}
\end{eqnarray}

Here, the tensors $(\nabla_x u) \otimes \nabla_x \rho$ and $(\nabla_x u) \otimes (u \cdot \nabla_x) u$ are order three tensors.  In \eqref{eq:intSeoToepsi}, they are contracted with respect to their three indices against three indices of the fourth order tensor $\Sigma$. The latter being symmetric, we do not need to specify which are its indices involved in the contraction. We can write:
$$ \Sigma : \big( (\nabla_x u) \otimes \nabla_x \rho \big) =  \big( \Sigma : (\nabla_x u) \big) \nabla_x \rho $$
where the right-hand side shows the multiplication of the symmetric matrix (or second order tensor) $\Sigma : (\nabla_x u)$ and of the vector $\nabla_x \log \rho$. Again, in $\Sigma : \nabla_x u$, the contracted product of the order $2$ tensor $\nabla_x u$ with the order $4$ tensor $\Sigma$, we do not need to specify the indices involved in the contraction. A simple computation using \eqref{eq:defSigma} and \eqref{eq:divu} shows that 
$$ \Sigma : \nabla_x u = (\nabla_x \cdot u) P_{u^\perp} + P_{u^\perp} (\nabla_x u) P_{u^\perp} + 
P_{u^\perp} (\nabla_x u)^T P_{u^\perp}, $$
where the exponent $T$ denotes transposition. Consequently, we have that
\begin{equation} 
\Sigma : \big( (\nabla_x u) \otimes \nabla_x \rho \big) 
= (\nabla_x \cdot u) \, P_{u^\perp} \nabla_x \rho + (P_{u^\bot} \, \nabla_x u) (P_{u^\bot} \nabla_x \rho) + \big( (P_{u^\bot} \nabla_x \rho) \cdot P_{u^\bot} \nabla_x \big) u. 
\label{eq:Sigma:naunarho}
\end{equation}
A similar computations yields:
\begin{equation}
\Sigma : \big( (\nabla_x u) \otimes (u \cdot \nabla_x) u  \big) = (\nabla_x \cdot u) \, (u \cdot \nabla_x) u + (P_{u^\bot} \nabla_x u) \big( (u \cdot \nabla_x) u \big)  + \big((u \cdot \nabla_x) u \cdot P_{u^\bot} \nabla_x \big) u, \label{eq:Sig:nauotunau} 
\end{equation}
where we have used that $(u \cdot \nabla_x) u$ is orthogonal to $u$. 

Then, we turn towards the second term of \eqref{eq:intApsi=} and get
\begin{eqnarray}
\int_{{\mathbb S}^{d-1}} \rho \, M_u \, S_{o,e} \, T_{e,o}  \, \vec \psi_u \, d\omega &=& B_{12} \, (u \cdot \nabla_x \log \rho) \, P_{u^\bot} \nabla_x \rho + B_{22} \, (u \cdot \nabla_x \rho) \, (u \cdot \nabla_x) u \nonumber \\
&+&  A_{21} \, \Sigma : \big( (\nabla_x u) \otimes \nabla_x \rho \big)  + A_{22} \, \rho \, \Sigma : \big( (\nabla_x u) \otimes (u \cdot \nabla_x) u \big) , \nonumber\\
\label{eq:intSoeTeopsi}
\end{eqnarray}
with 
\begin{eqnarray}
B_{12} &=& \frac{1}{d-1} \int_{{\mathbb S}^{d-1}} M_u \, h \, a \, \,  (\omega \cdot u) \, (1-(\omega \cdot u)^2) \, d \omega, \label{eq:B12}\\
B_{22} &=& \frac{\kappa}{d-1} \int_{{\mathbb S}^{d-1}} M_u \, h \, b \, (\, \omega \cdot u) \, (1-(\omega \cdot u)^2) \, d \omega. \label{eq:B22} \\
A_{21} &=& \frac{\kappa}{(d-1)(d+1)} \int_{{\mathbb S}^{d-1}} M_u \, h \, a \, \,  (\omega \cdot u) \, (1-(\omega \cdot u)^2)^2 \, d \omega, \label{eq:A21}\\
A_{22} &=& \frac{\kappa^2}{(d-1)(d+1)} \int_{{\mathbb S}^{d-1}} M_u \, h \, b \, (\, \omega \cdot u) \, (1-(\omega \cdot u)^2)^2 \, d \omega. \label{eq:A22} 
\end{eqnarray}

Now, we compute the contribution of $B$. To compute $(\omega \cdot \nabla_x) f_1$, we will need the following identities which follow from straightforward computations: 
\begin{eqnarray}
(\omega \cdot \nabla_x) (\omega \cdot u) &=& (\omega\cdot u) \, \Big( \omega_\bot \cdot \big( (u \cdot \nabla_x) u \big) \Big) + (\omega_\bot \otimes \omega_\bot):(\nabla_x u), \label{eq:omnaomu} \\
(\omega \cdot \nabla_x) \omega_\bot &=& - \big[ (\omega\cdot u) \, \Big( \omega_\bot \cdot \big( (u \cdot \nabla_x) u \big) \Big) + (\omega_\bot \otimes \omega_\bot):(\nabla_x u) \big] \, u \ \  \\
&& - (\omega \cdot u) \big[ (\omega \cdot u) \, (u \cdot \nabla_x) u  + (\omega_\bot \cdot \nabla_x) u \big]. \label{eq:omnaomper}
\end{eqnarray}
Now, we can write 
\begin{eqnarray}
T_{e,o} + T_{o,e} &=& T_1 + \ldots + T_5, \quad \mbox{ with: } \label{eq:def_T_i} \\
T_1&=& a(\omega \cdot u) \, \omega_\perp \cdot (\nabla_x \log \rho), \label{eq:T1}\\
T_2&=& \kappa\, b(\omega \cdot u) \, \omega_\perp \cdot [(u\cdot \nabla_x) u], \label{eq:T2}\\
T_3&=& c(\omega \cdot u) \,  (u\cdot \nabla_x) \log \rho ,   \label{eq:T3} \\
T_4&=&  \kappa \, e(\omega \cdot u) \, (\omega_\perp \otimes \omega_\perp) : (\nabla_x u) ,   \label{eq:T4} \\
T_5&=&  \kappa \, k(\omega \cdot u) \, (\nabla_x \cdot u).  \label{eq:T5}
\end{eqnarray}
Now, because of \eqref{eq:omnaomu}, \eqref{eq:omnaomper}, the expression of $(\omega \cdot \nabla_x) T_1$ involves eight different terms but only four of them have the requested parities to contribute a non zero term in $\int B \, \vec \psi_u \, d\omega$. After some algebra, we get 
\begin{eqnarray}
&&\int_{{\mathbb S}^{d-1}} \rho \, M_u \, \vec \psi_u \, (\omega \cdot \nabla_x) T_1 \, d \omega = B_{11} \, \Sigma : (\nabla_x u \otimes \nabla_x \rho)  \nonumber \\
&&\qquad\quad +  B_{12} \big[- (P_{u^\bot} \nabla_x u) (P_{u^\bot} \nabla_x \rho) - \big((u \cdot \nabla_x) \rho \big) \, (u \cdot \nabla_x) u + P_{u^\bot} (u \cdot \nabla_x) ( \nabla_x \rho) \big], \nonumber\\
&&\label{eq:intT1}
\end{eqnarray}
with
\begin{eqnarray}
B_{11} &=& \frac{1}{(d-1)(d+1)} \int_{{\mathbb S}^{d-1}} M_u \, h \, a' \, \, \, (1-(\omega \cdot u)^2)^2 \, d \omega, \label{eq:B11}
\end{eqnarray}
and $B_{12}$ is given by \eqref{eq:B12}. Furthermore, we note that 
$$P_{u^\perp}\lp (u\cdot \nabla_x) \nabla_x \rho\rp= (P_{u^\perp} \nabla_x) ((u\cdot \nabla_x) \rho) - (P_{u^\perp}\nabla_x u) (P_{u^\perp}\nabla_x \rho),$$
so that \eqref{eq:intT1} gives
\begin{eqnarray}
\int_{{\mathbb S}^{d-1}} \rho \, M_u \, \vec \psi_u \, (\omega \cdot \nabla_x) T_1 \, d \omega &=& B_{11} \, \Sigma : (\nabla_x u \otimes \nabla_x \rho) + B_{12} \big[- 2 (P_{u^\bot} \nabla_x u) (P_{u^\bot} \nabla_x \rho) \nonumber \\
&& \hspace{-15pt} + (P_{u^\bot} \nabla_x) \big((u \cdot \nabla_x) \rho \big) - \big((u \cdot \nabla_x) \rho \big) \, (u \cdot \nabla_x) u \big]. \label{eq:intT1_bis}
\end{eqnarray}
Proceeding similarly for the other terms, we have 
\begin{eqnarray}
\hspace{-1cm}
&&\int_{{\mathbb S}^{d-1}} \rho \, M_u \, \vec \psi_u \, (\omega \cdot \nabla_x) T_2 \, d \omega = B_{21} \, \rho \, \Sigma : \big(\nabla_x u \otimes (u \cdot \nabla_x) u \big)  \nonumber \\
&&\qquad\qquad + B_{22} \, \rho \, \big[- (P_{u^\bot} \nabla_x u) \big( (u \cdot \nabla_x) u \big) + P_{u^\bot} (u \cdot \nabla_x) \big( (u \cdot \nabla_x) u \big) \big], \nonumber \\
\label{eq:intT2}
\end{eqnarray}
with 
\begin{eqnarray}
B_{21} &=& \frac{\kappa}{(d-1)(d+1)} \int_{{\mathbb S}^{d-1}} M_u \, h \, b' \, \, \, (1-(\omega \cdot u)^2)^2 \, d \omega, \label{eq:B21}
\end{eqnarray}
and $B_{22}$, given by \eqref{eq:B22}; 
\begin{eqnarray}
\int_{{\mathbb S}^{d-1}} \rho \, M_u \, \vec \psi_u \, (\omega \cdot \nabla_x) T_3 \, d \omega &=& B_{31} \, (u \cdot \nabla_x \rho) \,  (u \cdot \nabla_x) u +  B_{32} \, P_{u^\bot} \nabla_x (u \cdot \nabla_x \rho),\nonumber\\
 \label{eq:intT3}
\end{eqnarray}
with 
\begin{eqnarray}
B_{31} &=& \frac{1}{d-1} \int_{{\mathbb S}^{d-1}} M_u \, h \, c' \, \, (\omega \cdot u) \, (1-(\omega \cdot u)^2) \, d \omega, \label{eq:B31}
\end{eqnarray}
and $B_{32}$, given by \eqref{eq:B32}; 
\begin{eqnarray}
\int_{{\mathbb S}^{d-1}} \rho \, M_u \, \vec \psi_u \, (\omega \cdot \nabla_x) T_4 \, d \omega &=& (B_{41} - B_{42}) \, \rho \, \Sigma : \big(\nabla_x u \otimes (u \cdot \nabla_x) u \big) \nonumber \\
&-&  B_{43} \, \rho \, \Big[ \Big( \big( (u \cdot \nabla_x) u \big) \cdot P_{u^\perp}\nabla_x \Big) u + (P_{u^\bot} \nabla_x u) \big( (u \cdot \nabla_x) u \big) \Big] \nonumber \\
&+& B_{42} \, \rho \, \Sigma : \nabla_x^2 u, \label{eq:intT4}
\end{eqnarray}
with 
\begin{eqnarray}
B_{41} &=& \frac{\kappa}{(d-1)(d+1)} \int_{{\mathbb S}^{d-1}} M_u \, h \, e' \, \, (\omega \cdot u) \, (1-(\omega \cdot u)^2)^2 \, d \omega, \label{eq:B41} \\
B_{43} &=& \frac{\kappa}{d-1} \int_{{\mathbb S}^{d-1}} M_u \, h \, e \, \, (\omega \cdot u)^2 \, (1-(\omega \cdot u)^2) \, d \omega, \label{eq:B43}
\end{eqnarray}
and $B_{42}$, given by \eqref{eq:B42};  
\begin{eqnarray}
\int_{{\mathbb S}^{d-1}} \rho \, M_u \, \vec \psi_u \, (\omega \cdot \nabla_x) T_5 \, d \omega &=& B_{51} \, \rho \, (\nabla_x \cdot u) \, (u \cdot \nabla_x) u + B_{52} \, \rho \, (P_{u^\bot} \nabla_x) (\nabla_x \cdot u), \nonumber \\
\label{eq:intT5}
\end{eqnarray}
with 
\begin{eqnarray}
B_{51} &=& \frac{\kappa}{d-1} \int_{{\mathbb S}^{d-1}} M_u \, h \, k' \, \, (\omega \cdot u) \, (1-(\omega \cdot u)^2) \, d \omega, \label{eq:B51} 
\end{eqnarray}
and $B_{52}$, given by \eqref{eq:B52}. In \eqref{eq:intT4}, the symbol $\nabla_x^2 u$ denotes the third order tensor of the second derivatives of $u$ with components $(\nabla_x^2 u)_{ijk} = \partial_{x_i} \partial_{x_j} u_k$. The symbol $\Sigma : \nabla_x^2 u$ denotes the vector obtained by contracting  $\Sigma \otimes \nabla_x^2 u$ over three indices (which ones being unsignificant due to the symmetry of $\Sigma$). We can prove that 
\begin{eqnarray}
 \Sigma : \nabla_x^2 u &=& P_{u^\bot} \big( \nabla_x \cdot ( P_{u^\bot} \nabla_x u ) \big) + (\nabla_x \cdot u) \, (u \cdot \nabla_x) u + \Big( \big( (u \cdot \nabla_x) u \big) \cdot P_{u^\perp}\nabla_x \Big) u \nonumber \\
&+& 2 P_{u^\bot} \nabla_x (\nabla_x \cdot u) + 2 (P_{u^\bot} \nabla_x u) \big( (u \cdot \nabla_x) u \big). 
\label{eq:Signa2u}
\end{eqnarray}
Indeed, with \eqref{eq:defSigma}, we have, using indices: 
\begin{equation}
 (\Sigma : \nabla_x^2 u)_\ell  
= (P_{u^\bot})_{ij} \, (P_{u^\bot})_{k \ell} \, \partial_{x_i} \partial_{x_j} u_k 
+ (P_{u^\bot})_{ik} \, (P_{u^\bot})_{j \ell} \, \partial_{x_i} \partial_{x_j} u_k 
+ (P_{u^\bot})_{i\ell} \, (P_{u^\bot})_{jk} \, \partial_{x_i} \partial_{x_j} u_k, 
\label{eq:Signa2ul}
\end{equation}
where Einstein's repeated index summation is being used. We first note that the second and third terms are equal by exchange of the dummy indices $i$ and $j$. Now, the first term can be written:
\begin{equation}
(P_{u^\bot})_{ij} \, (P_{u^\bot})_{k \ell} \, \partial_{x_i} \partial_{x_j} u_k  
= (P_{u^\bot})_{\ell k} \, \partial_{x_i} \Big( (P_{u^\bot})_{ij} \,  \partial_{x_j} u_k \Big)
- (P_{u^\bot})_{\ell k} \,  \, \big( \partial_{x_i} (P_{u^\bot})_{ij} \big) \,  \partial_{x_j} u_k. 
\label{eq:PuPunanau}
\end{equation}
With $\partial_{x_i} (P_{u^\bot})_{ij} = - (\partial_{x_i} u_i) \, u_j - u_i \, (\partial_{x_i} u_j)$, we get 
\begin{eqnarray*} (P_{u^\bot})_{\ell k} \,  \big( \partial_{x_i} (P_{u^\bot})_{ij} \big) \,  \partial_{x_j} u_k 
&=& - (\partial_{x_i} u_i) \, (P_{u^\bot})_{\ell k} \,   \big( (u_j \partial_{x_j}) u_k \big)
+ (P_{u^\bot})_{\ell k} \Big( \big( (u_i \partial_{x_i}) u_j \big) \partial_{x_j} \Big) u_k \\
&=& - \bigg( (\nabla_x \cdot u) \, (u \cdot \nabla_x) u + \Big( \big( (u \cdot \nabla_x) u \big) \cdot \nabla_x \Big) u \bigg)_\ell. 
\end{eqnarray*}
We note that the first term at the right-hand side of \eqref{eq:PuPunanau} can be written $\big(P_{u^\bot} \big( \nabla_x \cdot ( P_{u^\bot} \nabla_x u ) \big) \big)_\ell$. So, collecting all these identities, we get the first line of \eqref{eq:Signa2u}. For the second term of \eqref{eq:Signa2ul} we write:
\begin{equation} (P_{u^\bot})_{ik} \, (P_{u^\bot})_{j \ell} \, \partial_{x_i} \partial_{x_j} u_k 
= (P_{u^\bot})_{\ell j} \, \partial_{x_j} \big( (P_{u^\bot})_{ik} \, \partial_{x_i} u_k \big)
- (P_{u^\bot})_{\ell j} \, (\partial_{x_j} (P_{u^\bot})_{ik}) \, (\partial_{x_i}  u_k).  
\label{eq:PuPupapau}
\end{equation}
With $\partial_{x_j} (P_{u^\bot})_{ik} = - (\partial_{x_j} u_i) \, u_k - u_i \partial_{x_j} u_k$ and noting that $ u_k (\partial_{x_i}  u_k) = 0$, we get
$$ 
(P_{u^\bot})_{\ell j} \, (\partial_{x_j} (P_{u^\bot})_{ik}) \, (\partial_{x_i}  u_k) = 
- \big( (P_{u^\bot})_{\ell j} \, \partial_{x_j} u_k \big) \, \big( (u_i \partial_{x_i})  u_k \big) = - \bigg( (P_{u^\bot} \nabla_x u) \big( (u \cdot \nabla_x) u \big) \bigg)_\ell.
$$
Since the first term at the right-hand side of \eqref{eq:PuPupapau} can be written $\big( P_{u^\bot} \nabla_x (\nabla_x \cdot u) \big)_\ell$ (using \eqref{eq:divu}), we get the second line of 
\eqref{eq:Signa2u}, remembering that the second and third terms of \eqref{eq:Signa2ul} are equal. 

Now, we collect \eqref{eq:intSeoToepsi}, \eqref{eq:intSoeTeopsi}, \eqref{eq:intT1_bis}, \eqref{eq:intT2}, \eqref{eq:intT3}, \eqref{eq:intT4} and \eqref{eq:intT5}, and use formulas \eqref{eq:Sigma:naunarho}, \eqref{eq:Sig:nauotunau} and \eqref{eq:Signa2u} to obtain \eqref{eq:macro_equations_u} with the following formulas for the $E$, $F$, $G$, $H$ constants:
\begin{eqnarray*}
&& E_1 = \frac{B_{12} + B_{32}}{C_0}, \quad F_1 = \frac{B_{22}}{C_0}, \quad F_2 = \frac{B_{42}}{C_0}, \quad F_3 = \frac{2B_{42} + B_{52}}{C_0}, \\
&& G_1 = \frac{{A_{11} + B_{22} + B_{31} - B_{12}}}{C_0}, \quad G_2 = \frac{- 2 B_{12} + B_{42} + A_{21} + B_{11}}{C_0}, \\
&& G_3 = \frac{B_{42} + A_{21} + B_{11}}{C_0}, \quad G_4 = \frac{B_{52} + B_{42} + A_{21} + B_{11}}{C_0},  \\
&&  H_1=\frac{B_{32} + B_{12}}{C_0}, \quad H_2 = \frac{- B_{22}- B_{43} + A_{12} + A_{22} + B_{21} + B_{41} + B_{42}}{C_0},  \\
&& H_3 = \frac{- B_{43} + A_{12} + A_{22} + B_{21} + B_{41}}{C_0}, \quad H_4 = \frac{A_{13} + B_{51} + A_{12} + A_{22} + B_{21} + B_{41} }{C_0}. 
\end{eqnarray*}
The expressions \eqref{eq:E1} to \eqref{eq:H4} are deduced from the expressions of the $A_{ij}$ and $B_{ij}$ given in this section through the use of \eqref{eq:int_spheriq}. This ends the proof of Theorem \ref{th:macro}.

\section{Summary and outlook}
\label{sec:summary}

In this paper, we have derived a cross-diffusion system for the density and mean direction of a system of self-propelled particles interacting through nematic alignment. This derivation highlights the role of the generalised collision invariants in the inversion of the linearized collision operator. In the future, we may expect that this technique will be useful to derive macroscopic models for other kinds of alignment interactions in the diffusive regime. An example of this is given by nematically moving particles interacting through nematic alignment. We can also develop similar techniques for abrupt collisions leading to jumps in the particle directions. The limit system itself poses a number of challenges. The first one of course is its well-posedness. If the second-order terms prove to be elliptic as conjectured, then we may at least hope for local-in-time well-posedness. The relations between the structure of the model and the underlying symmetries of the system are worth being explored further. We may in this way provide an exhaustive list of models compatible with the underlying symmetries. As the model presents a large number of different terms, a natural question is to understand the role of each of them. A related one is to determine whether these terms are independent from each other, or if some structural relations between the coefficients are needed for well-posedness. The numerical simulation of the model will also be challenging. Indeed, given the role of the symmetries it is desirable to develop methods that preserve them. This will require the development of new methods as traditional grid-based methods break the rotational invariance of the continuous problem. Finally, it will be interesting to investigate whether the continuous model can produce similar patterns as the underlying particle model. If this is the case, this patterning ability could be analyzed through dynamical systems techniques such as bifurcation analysis.

\section*{Data statement}
No new data were collected in the course of this research.

\section*{Conflict of interest}
The authors have no conflict of interest to declare.

\appendix

\section{Proof of the identity \eqref{eq:aux_decompPDPD}}
\label{sec:proof_aux_decomp}

We consider the $i$-th component of the following expressions:
\beqar
I:= P_{u^\perp}\lp \nabla_x \cdot (P_{u^\perp}\nabla_x u) \rp_i = (P_{u^\perp})_{ij} \partial_k [(P_{u^\perp})_{kl} \partial_l u_j],\\
II:= \lp \mbox{Tr}_{[12]} (P_{u^\perp}\nabla_x) (P_{u^\perp}\nabla_x u) \rp_i = (P_{u^\perp})_{jp} \partial_p[(P_{u^\perp})_{j\ell}\partial_\ell u_i],
\eeqar
where we have used Einstein's convention (sum of terms with double indices).
A computation shows that
\beqar
II &=& (P_{u^\perp})_{jp} \partial_p[(P_{u^\perp})_{jl}(P_{u^\perp})_{is}\partial_\ell u_s] \\
&=& (P_{u^\perp})_{jp}(P_{u^\perp})_{is} \partial_p[(P_{u^\perp})_{j\ell}\partial_\ell u_s] + (P_{u^\perp})_{jp}(P_{u^\perp})_{j\ell}\partial_\ell u_s \partial_p (P_{u^\perp})_{is} \\
&=& (P_{u^\perp})_{is} \partial_j\lp (P_{u^\perp})_{j\ell} \partial_\ell u_s\rp - (P_{u^\perp})_{is} u_j u_p \partial_p[(P_{u^\perp})_{j\ell}\partial_\ell u_s] - (P_{u^\perp})_{p\ell}\partial_\ell u_s  u_i \partial_p u_s) \\
&=& (P_{u^\perp})_{ij} \partial_k[(P_{u^\perp})_{k\ell}\partial_\ell u_j]+(P_{u^\perp})_{is} u_p \partial_p u_j (P_{u^\perp})_{j\ell}\partial_\ell u_s- (P_{u^\perp})_{p\ell}\partial_\ell u_s  u_i \partial_p u_s \\
&=& (P_{u^\perp})_{ij} \partial_k [(P_{u^\perp})_{k\ell}\partial_\ell u_j]+ u_p \partial_p u_j \partial_j u_i - u_i (P_{u^\perp})_{p\ell}\partial_\ell u_s \partial_p u_s,
\eeqar
where in the first equality we  used that $(P_{u})_{is}\partial_\ell u_s = \partial_\ell u_i$, since $u_s \partial_\ell u_s =0$; in the third equality we used that $u_s\partial_\ell u_s=0$; in the fourth equality we used that
$u_j (P_{u^\perp})_{j\ell}=0$ and in the first term we changed the labels $s$ for $j$ and $j$ for $k$.\\
From this we conclude that
$$II= I + ((u\cdot\nabla_x) u)\cdot\nabla_x)u - u (P_{u^\perp}\nabla_x u: P_{u^\perp}\nabla_x u)),$$
which is, precisely, expression \eqref{eq:aux_decompPDPD}.
The last term in this expression follows from the following:
$$
(P_{u^\perp})_{p\ell}\partial_\ell u_s \partial_p u_s = (P_{u^\perp})_{p\ell}\partial_\ell u_s (P_{u^\perp})_{pq}\partial_q u_s\\
= (P_{u^\perp}\nabla_x u)_{ps}(P_{u^\perp}\nabla_x u)_{ps}.
$$

\section*{Acknowledgment}
PD acknowledges support by the Engineering and Physical Sciences Research Council (EPSRC)
under grants no. EP/M006883/1 and EP/P013651/1, by the Royal Society and the
Wolfson Foundation through a Royal Society Wolfson Research Merit Award no. WM130048
and by the National Science Foundation (NSF) under grant no. RNMS11-07444 (KI-Net).
PD is on leave from CNRS, Institut de Math\'ematiques de Toulouse, France.

\noindent SMA is supported by the Vienna Science and Technology Fund (WWTF) with a Vienna
Research Groups for Young Investigators, grant VRG17-014.

\end{document}